\documentclass[11pt,a4paper]{article}
\usepackage[utf8]{inputenc}
\usepackage[english]{babel}
\usepackage{siunitx}
\usepackage{subcaption}
\usepackage{float}
\usepackage{fancyhdr}
 \usepackage{chngcntr}
 \usepackage{cancel}
\usepackage{hyperref}
\usepackage{authblk}

\usepackage{quiver}
\usepackage{ytableau}

\usepackage{stackengine}

\usepackage{hyperref}
\usepackage{color}
\usepackage{amsthm}
\usepackage{mathtools}

\usepackage[normalem]{ulem} 
\usepackage{graphicx}
\usepackage{amssymb,latexsym,cite}
\usepackage{amsmath}
\usepackage{amsfonts}
\usepackage{mathrsfs}
\usepackage{bbm}
\usepackage{bm}
\usepackage[T1]{fontenc}

\usepackage[matrix,arrow,color]{xy}

\usepackage{enumitem}

\def\slasha#1{\setbox0=\hbox{$#1$}#1\hskip-\wd0\hbox to\wd0{\hss\sl/\/\hss}}

\def\periodb#1{\setbox0=\hbox{$#1$}#1\hskip-\wd0\hbox to\wd0{-}}

\newcommand{\ident}{\mathbbm{1}}   			
\newcommand{\ii}{\mathrm{i}}   			
\newcommand{\jj}{\mathrm{j}}   			
\newcommand{\kk}{\mathrm{k}}   			
\newcommand{\e}{\mathrm{e}}   			

\newcommand{\CA}{\mathcal{A}}    			

\newcommand{\CH}{\mathcal{H}}

\newcommand{\CCK}{\mathscr{K}}

\newcommand{\CN}{\mathcal{N}}

\newcommand{\CO}{\mathcal{O}}

\newcommand{\CP}{\mathcal{P}}

\newcommand{\CQ}{\mathcal{Q}}

\newcommand{\CR}{\mathcal{R}}
\newcommand{\CS}{\mathcal{S}}

\newcommand{\CCV}{\mathscr{V}}

\newcommand{\CE}{\mathcal{E}}

\newcommand{\frg}{\mathfrak{g}}				

\newcommand{\frF}{\mathfrak{F}}

\newcommand{\frM}{\mathfrak{M}}
\newcommand{\frB}{\mathfrak{B}}

\newcommand{\frZ}{\mathfrak{Z}}

\newcommand{\ulfour}{\underline{\sf 4}}
\newcommand{\ulthree}{\underline{\sf 3}}

\DeclareMathOperator*{\timesbig}{\scalerel*{\times}{\textstyle\sum}}
\DeclareMathOperator*{\Timesbig}{\scalerel*{\times}{\displaystyle\sum}}
\usepackage{scalerel}

\def\tv{{\textrm{\tiny $V$}}}
\def\tv1{{\textrm{\tiny $V[1]$}}}

\newcommand{\mbf}[1]{{\boldsymbol {#1} }}

\newcommand{\FR}{\mathbbm{R}}     			
\newcommand{\FC}{\mathbbm{C}}     			
\newcommand{\RZ}{\mathbbm{Z}}     			
\newcommand{\PP}{{\mathbbm{P}}}    			

\newcommand{\sone}{\mathbb{S}}

\newcommand{\dd}{\mathrm{d}}     			

\newcommand{\sU}{\mathsf{U}}     			

\newcommand{\sG}{\mathsf{G}}
\newcommand{\sT}{\mathsf{T}}

\newcommand{\sHom}{\mathsf{Hom}}

\newcommand{\sH}{\mathsf{H}}
\newcommand{\sSU}{\mathsf{SU}}
\newcommand{\sSL}{\mathsf{SL}}
\newcommand{\sGL}{\mathsf{GL}}

\newcommand{\sGamma}{\mathsf{\Gamma}}
\newcommand{\shGamma}{{\widehat{\mathsf{\Gamma}}}}
\newcommand{\sSO}{\mathsf{SO}}
\newcommand{\sSpin}{\mathsf{Spin}}

\newcommand{\sEnd}{\mathsf{End}}

\newcommand{\comment}[1]{}     				
\def\tyng(#1){\hbox{\tiny$\yng(#1)$}}			
\def\tyoung(#1){\hbox{\tiny$\young(#1)$}}			

\newcommand{\beq}{\begin{eqnarray}}
\newcommand{\eeq}{\end{eqnarray}}

\newcommand{\Hilb}{{\sf Hilb}}
\newcommand{\Ob}{{\sf Ob}}

\definecolor{outrageousorange}{rgb}{1.0, 0.43, 0.29}

\newenvironment{myitemize}{\begin{itemize}[itemsep=-0.05cm, leftmargin=*, topsep=0.1cm]}{\end{itemize}}

\newcommand{\qu}{\mathtt{q}}
\newcommand{\Qu}{\mathtt{Q}}
\newcommand{\ttO}{\mathtt{O}}
\newcommand{\tto}{\mathtt{o}}

\newcommand{\Tr}{\mathrm{Tr}}

\theoremstyle{definition}

\newtheorem{corollary}[equation]{Corollary}
\newtheorem{lemma}[equation]{Lemma}
\newtheorem{proposition}[equation]{Proposition}

\newtheorem{conjecture}[equation]{Conjecture}

\newtheorem{remark}[equation]{Remark}
\newtheorem{example}[equation]{Example}

\newcommand{\midwedge}{\text{\Large$\wedge$}}

\def\beq{\begin{equation}}
\def\bee{\begin{equation}}
\def\eeq{\end{equation}}
\def\bea{\begin{eqnarray}}
\def\eea{\end{eqnarray}}
\def\ba{\begin{align}}
\def\ea{\end{align}}

\numberwithin{equation}{section}
\catcode`@=12

\newcommand{\ttQ}{{\sf Q}}
\newcommand{\ttA}{{\sf A}}
\newcommand{\ttR}{{\sf R}}
\newcommand{\ttZ}{{\sf Z}}
\newcommand{\ttX}{{\tt X}}
\newcommand{\ttY}{{\tt Y}}
  
\newcommand{\Quot}{{\sf Quot}}
\newcommand{\ch}{\mathrm{ch}}

\begin{document}

\title{\bf Instanton Counting and Donaldson--Thomas Theory \\ on Toric Calabi--Yau Four-Orbifolds}

\author{Richard J. Szabo\thanks{R.J.Szabo@hw.ac.uk} \ }
\author{ Michelangelo Tirelli\thanks{mt2001@hw.ac.uk}}

\affil{\textit{\normalsize Department of Mathematics,
Heriot–Watt University}\\ \vspace{-1mm}
\textit{\normalsize Colin Maclaurin Building, Riccarton, Edinburgh EH14 4AS, UK}\\ \vspace{1mm}
\textit{\normalsize Maxwell Institute for Mathematical Sciences, Edinburgh, UK}\\ \vspace{1mm}
\textit{\normalsize Higgs Centre for Theoretical Physics, Edinburgh, UK}}
\date{}
\maketitle

\vspace{1cm}

\begin{abstract}
\noindent
We study rank $r$ cohomological Donaldson--Thomas theory on a toric Calabi--Yau orbifold of $\mathbbm{C}^4$ by a finite abelian subgroup $\sGamma$ of $\sSU(4)$, from the perspective of instanton counting in cohomological gauge theory on a noncommutative crepant resolution of the quotient singularity. We describe the moduli space of noncommutative instantons on $\mathbbm{C}^4/\sGamma$ and its generalized ADHM parametrization. Using toric localization, we compute the orbifold instanton partition function as a combinatorial series over $r$-vectors of $\sGamma$-coloured solid partitions. When the $\sGamma$-action fixes an affine line in $\FC^4$, we exhibit the dimensional reduction to rank $r$ Donaldson--Thomas theory on the toric K\"ahler three-orbifold $\FC^3/\sGamma$.
Based on this reduction and explicit calculations, we conjecture closed infinite product formulas, in terms of generalized MacMahon functions, for the instanton partition functions on the orbifolds $\FC^2/\RZ_n\times\FC^2$ and  $\mathbbm{C}^3/(\mathbbm{Z}_2\times\RZ_2)\times\mathbbm{C}$, finding perfect agreement with new mathematical results of Cao, Kool and Monavari.
\end{abstract}

\vspace{2cm}

\begin{flushright}
		\small
		{\sf EMPG--23--01}
	\end{flushright}

\newpage

{
\tableofcontents
}

\bigskip

\section{Introduction}
\label{sec:Intro}

The counting of BPS states in string theory and quantum field theory often leads to deep mathematical insights into counting problems of enumerative geometry. Of particular interest are vacuum moduli spaces of D-branes and instantons, which often provide alternative contructions of the relevant moduli spaces that are otherwise technically difficult to define rigorously in geometrical approaches. Looking in the other direction, mathematical constructions of enumerative geometry can shed light on structural properties of the BPS spectrum of particles in string theory and quantum field theory. 

In this paper we are concerned with the connections between BPS state counting problems and Donaldson--Thomas theory~\cite{Donaldson:1996kp,Donaldson:2009yq}. Donaldson--Thomas invariants are virtual numbers counting sheaves on a complex variety. They are defined as integrals of cohomology classes over virtual cycles of moduli spaces of sheaves. Their best understood physics connection is to type~II string theory compactified on a K\"ahler threefold $M$, where BPS states preserve half of the $\CN=2$ supersymmetry and correspond to bound states of D-branes.  The vacuum degeneracies are computed by the Witten index and the corresponding partition function reproduces the generating function for the Donaldson--Thomas invariants of $M$ (see e.g.~\cite{Yamazaki:2010fz,Cirafici:2012qc} for reviews). 

The enumeration of D-brane bound states on $M$ can be equivalently reformulated as an instanton counting problem in a six-dimensional $\CN_{\textrm{\tiny T}}=2$ cohomological gauge theory on $M$. The computation of Donaldson--Thomas partition functions from this perspective has been studied in great detail by~\cite{Iqbal:2003ds,coho,Quiver3d,Cirafici:2011cd} in the case where $M$ is a toric Calabi--Yau threefold. For this, the moduli space of $\sU(r)$ instantons is compactified in two ways: by deforming the first order partial differential equations defining BPS states to operator algebraic equations for noncommutative instantons, and by a local $\varOmega$-deformation of $M$ which preserves its $\sSU(3)$-holonomy. 

For the simplest example $M=\FC^3$, the Donaldson--Thomas partition function enumerates BPS bound states of D$p$-branes inside $r$ D$(p{+}6)$-branes. A six-dimensional version of the ADHM construction in four dimensions~\cite{Atiyah:1978ri} establishes that, from a geometric point of view, the compactified instanton moduli space is isomorphic to a moduli space of torsion-free sheaves on complex projective space $\mathbbm{P}^3$ with suitable characteristic classes and framing conditions~\cite{coho,Quiver3d}. The Coulomb branch instanton partition function can be written in a simple compact form as~\cite{Iqbal:2003ds,coho} 
\begin{align}
    Z^{r}_{\mathbbm{C}^3}(\qu)=M\big((-1)^r\,\qu\big)^r \ ,
\end{align}
where $M(q)=\prod_{n\geq1}\,(1-q^n)^{-n}$ is the MacMahon function and $\qu$ is the Boltzmann weight parameter for instantons. See~\cite{Szabo:2009vw,Szabo:2011mj} for reviews of various physical and geometrical aspects of this correspondence. 

General D$p$--D$p'$-brane systems can be considered after turning on fluxes for the Neveu--Schwarz (NS) $B$-field, in such a way that supersymmetry is restored in the vacuum~\cite{Witten:1995im,Boundstates}. In this paper we are interested in the codimension $p'-p=8$. Many new features appear in eight dimensions. The corresponding eight-dimensional cohomological gauge theories were constructed and studied in the late 1990s~\cite{special,8d,8ds}. They are in many respects similar to their four-dimensional counterparts, namely  Donaldson--Witten theory~\cite{Witten:1988ze} and Vafa--Witten theory~\cite{Vafa:1994tf} on complex surfaces. The equivariant instanton partition functions on $\FC^4$ were recently studied by Nekrasov and Piazzalunga in~\cite{m4,m4c}. 

The renewed interest in these eight-dimensional quantum field theories has been sparked by recent mathematical advances in the Donaldson--Thomas theory of Calabi--Yau fourfolds, starting with the seminal work of Cao and Leung~\cite{Cao:2014bca} which constructed a virtual fundamental class in special cases. Subsequently, virtual cycles of the Donaldson--Thomas moduli spaces for fourfolds were developed more generally by Borisov and Joyce~\cite{Borisov:2015vha} in the setting of derived differential geometry, as well as
by Oh and Thomas~\cite{Oh:2020rnj} in the setting of algebraic geometry. 
The virtual cycle is defined via a suitable choice of a local orientation at each point of the moduli space. This corresponds to a choice of \emph{signs}, a well-known phenomenon in eight dimensions that does not arise in lower dimensions. These signs are the most important and non-trivial aspects of the theory, while at the same time presenting one of the major difficulties. The computations of Cao and Kool~\cite{Cao:2017swr} show that the choice of signs seems unique (up to overall orientation), and that they always give the simplest answer for the Donaldson--Thomas partition functions. 

In the following we shall elaborate on several aspects of the construction of noncommutative instantons in eight dimensions and the evaluation of their partition functions. We give a new derivation of the equivariant instanton partition function on $\FC^4$ that carefully incorporates the correct sign choices. Our sign choice differs from other choices that have appeared so far in the literature.

The central achievement of this paper is a detailed systematic study and computations of the rank $r$ degree zero cohomological Donaldson--Thomas invariants of the toric Calabi--Yau four-orbifolds $\FC^4/\sGamma$, extending the flat space treatment on $\FC^4$.
These have so far received only limited attention in both the physics and mathematics literature. The analogous instanton counting problem in six dimensions has been studied in detail in~\cite{Quiver3d,Cirafici:2011cd}; see~\cite{Cirafici:2012qc} for a review with comparisons to the instanton counting problems in the Donaldson--Witten and Vafa--Witten theories. In eight dimensions, only the instanton partition functions for some simple cyclic orbifolds have been briefly discussed in~\cite{Bon,Kimura:2022zsm}. A thorough mathematical treatment of the rank one K-theoretic Donaldson--Thomas theory of Calabi--Yau four-orbifolds appears in parallel to our work in~\cite{CKMpreprint}, complementing the results of the present paper, as we discuss further below.

In this paper we construct and study rank $r$ cohomological gauge theory on quotient stacks $\big[\FC^4/\sGamma\big]$, where $\sGamma$ is a finite abelian subgroup of $\sSL(4,\FC)$. This is equivalent to the gauge theory on a noncommutative crepant resolution of the quotient singularity $\FC^4/\sGamma$ provided by the path algebra $\ttA$ of a generalization of the McKay quiver determined by representation theory data of $\sGamma$, with relations given by a generalized ADHM parametrization of the orbifold noncommutative instanton equations.
The topological gauge theory localizes by construction on $\sGamma$-equivariant instanton configurations. Using toric localization, the orbifold instanton partition function can be reduced to the fixed points of the moduli space under the action of the maximal torus of $\sSU(4)$ which are also $\sGamma$-invariant. These are classified by $r$-vectors of $\sGamma$-coloured solid partitions whose Boltzmann weights depend on the representations of $\sGamma$; this was also mentioned by~\cite{Bon} and is analogous to the combinatorial description of orbifold instantons in six dimensions in terms of plane partitions~\cite{Quiver3d}. We assert that the BPS partition function on $[\FC^4/\sGamma]$ conjecturally provides the corresponding orbifold Donaldson--Thomas invariants.

Noncommutative Donaldson--Thomas invariants of four-Calabi--Yau algebras  already appear in the seminal work of~\cite{Cao:2014bca}; see also~\cite{Cao:2020huo} for the example of the local resolved conifold. 
Our path algebras $\ttA$ are always four-Calabi--Yau, and in fact they are Koszul since $\sGamma\subset\sSL(4,\FC)$. In the three-orbifold case, this fact was used repeatedly in~\cite{Quiver3d} to establish the relation between equivariant sheaves on $\FC^3$ and BPS states in the noncommutative resolution chamber of resolved Calabi--Yau singularities. The wall-and-chamber structure of the K\"ahler moduli space is discussed in~\cite{Cao:2020huo} for the example of the local resolved conifold, building on the standard threefold case~\cite{Szendroi:2007nu}. 

If there exists a toric crepant resolution $X$ of the quotient singularity $\FC^4/\sGamma$, we can expect an analogous relation, as well as a version of the McKay correspondence which would establish an equivalence between the derived category of $\ttA$-modules and the derived category of $X$, when the former admits a tilting object. This should then relate the instanton partition functions in the `orbifold' and `large radius' phases by changes of variables and wall-crossing formulas. As such an analysis is out of reach with our current techniques, we defer it to future investigations.

\subsubsection*{Summary of Results}

Before starting a systematic study of our main topic, we first review in some detail the analogous problem on the flat space $\FC^4$, in order to set the stage as the extension to orbifolds is then relatively straightforward. We elaborate on the analysis of noncommutative instantons in $\sU(r)$ cohomological  gauge theory on $\mathbbm{C}^4$ with $\sSU(4)$-holonomy and their generalized ADHM parametrization. They are realized in a dimensional reduction of $\CN=1$ supersymmetric Yang--Mills theory from ten dimensions, and also in type~II string theory where the noncommutative deformation corresponds to turning on a non-zero constant background $B$-field in the flat ten-dimensional target spacetime~\cite{Seiberg:1999vs}. The instanton partition function on an $\varOmega$-background is then regarded as an equivariant integral over the instanton moduli space~\cite{SWcounting} and can be evaluated using toric localization techniques as a combinatorial expansion in random solid partitions.

We evaluate the equivariant  instanton partition function of the $\sU(r)$ cohomological gauge theory on $\FC^4$ with $r$ massive fundamental matter fields. It can be expressed in an exact closed form as ({Conjecture}~\ref{Prop1})
\begin{align}\label{eq:result_1}
Z_{\FC^4}^r(\qu;\vec{\epsilon},m)=M(-\qu)^{-\frac{r\,m\,\epsilon_{12}\,\epsilon_{13}\,\epsilon_{23}}{\epsilon_1\,\epsilon_2\,\epsilon_3\,\epsilon_4}} 
\end{align}
with
\begin{align}\label{eq:massdef}
m=\frac{1}{r} \, \sum_{l=1}^r\,(m_l-a_l) \ ,
\end{align}
where $a_l$ and $m_l$ are the Coulomb and mass parameters associated to the maximal tori of the global $\sU(r)$ colour and flavour symmetry groups, respectively, while $\epsilon_a$  are coordinates on the maximal torus of $\sSU(4)$ satisfying the Calabi--Yau constraint $\epsilon_1 + \cdots +\epsilon_4=0$; we use the notation $\epsilon_{ab}=\epsilon_a+\epsilon_b$. 

The formula \eqref{eq:result_1} for the equivariant Donaldson--Thomas partition function on $\FC^4$ is well-known and has appeared many times before in the literature. We sketch a possible alternative analytic proof which, while incomplete and lacking conceptual insight, highlights the symmetries of the theory which are not evident in other approaches and can be potentially extended to theories on more complicated spaces, like some of our orbifolds. It is based on the  fact that in six dimensions the Coulomb branch instanton partition function is known for generic $\varOmega$-deformation of $\FC^3$~\cite{Maulik:2004txy,Szaboconj,proofconj} and is recovered from the eight-dimensional theory through the mass specialization $m_l=\epsilon_4+a_l$ (Proposition~\ref{prop:ZDTgeb}). The idea then is to show that~\eqref{eq:result_1} is the unique expression determined by the instanton deformation complex which correctly reduces in the six-dimensional limit and preserves all symmetries of the cohomological matrix model representation of the instanton partition function.

Moving on to our orbifold theories, we prove an orbifold version of Proposition~\ref{prop:ZDTgeb} which relates the equivariant instanton partition function on an orbifold of the form $\FC^3/\sGamma\times\FC$, where $\sGamma\subset\sSL(3,\FC)$, to the noncommutative Donaldson--Thomas partition function for the toric K\"ahler orbifold $\FC^3/\sGamma$ through mass specialisation (Proposition~\ref{prop3}). In the rank one case $r=1$, the noncommutative Donaldson--Thomas partition function  on $\FC^3/\sGamma$ is expressed by a closed formula precisely for the the Kleinian group $\sGamma=\RZ_n$ in $\sSL(2,\FC)$ and the orbifold group $\sGamma=\RZ_2\times \RZ_2$ in $\sSL(3,\FC)$ at the Calabi--Yau specialization $\epsilon_1+\epsilon_2+\epsilon_3=0$ on~$\FC^3$~\cite{Young:2008hn,Quiver3d}. 

For $\sGamma=\RZ_n$  the infinite product formula is extended to generic triples $(\epsilon_1,\epsilon_2,\epsilon_3)$ in \cite{Zhou18}. From this we assert that our proposed proof of {Conjecture}~\ref{Prop1} should be possible to adapt to show that the equivariant instanton partition function  of the cohomological $\sU(1)$ gauge theory with massive fundamental matter on $[\mathbb{C}^2/\mathbb{Z}_n]\times \mathbb{C}^2$ is given by (Conjecture~\ref{con2})
\begin{align}
\begin{split}
Z_{[\FC^2/\RZ_n]\times \FC^2}(\vec \qu;\vec{\epsilon},m) & = M\big((-1)^n\,\Qu\big)^{-n\,\frac{m\,\epsilon_{12}\,\epsilon_{13}\,\epsilon_{23}}{\epsilon_1\,\epsilon_2\,\epsilon_3\,\epsilon_4}-\frac{n^2-1}{n}\,\frac{m\,\epsilon_{12}}{\epsilon_1\,\epsilon_2}} \\
& \quad \, \times \prod_{0<p\leq s<n}\,\widetilde{M}\big((-1)^{p-s+1}\,\qu_{[p,s]},(-1)^n\,\Qu\big)^{-\frac{m\,\epsilon_{12}}{\epsilon_3\,\epsilon_4}} \ ,
\end{split}
\end{align}
where $\Qu=\qu_0\,\qu_1\cdots \qu_{n-1}$, $\qu_{[p,s]}=\qu_p\,\qu_{p+1}\cdots \qu_{s-1}\,\qu_s$, and we weigh the fractional instanton contributions with fugacities $\qu_s$ indexed by the irreducible representations of $\sGamma=\RZ_n$. The infinite product \smash{$M(x,q)=\prod_{k\geq1}\,(1-x\,q^k)^{-k}$} is the generalized MacMahon function and we have set \smash{$\widetilde{M}(x,q) = M(x,q) \, M(x^{-1},q)$}. 

For $\sGamma=\RZ_2\times\RZ_2$ we do not yet have available results for a generic $\varOmega$-deformation of $\FC^3$, but we provide strong evidence in favour of the closed formula of Conjecture~\ref{con3}. It seems unlikely that there are any other orbifolds $\FC^4/\sGamma$ for which exact infinite product expressions for the instanton partition functions  are even conjecturally possible; see~\cite{CKMpreprint} for a geometric explanation of this.
 
For higher rank theories we formulate various conjectural closed formulas for orbifold instanton partition functions with particular framing decompositions $\vec r$ of the rank $r$ according to the irreducible representations of $\sGamma$, for the orbifolds  $\FC^2/\RZ_n\times\FC^2$ and $\FC^3/(\RZ_2\times\RZ_2)\times\FC$ (Conjectures~\ref{con4}, \ref{con:r0r1} and~\ref{con4a}). These conjectures are developed from the $\sGamma$-equivariant instanton deformation complex associated to the ADHM-type parametrization of the orbifold instanton moduli space. From this we obtain a combinatorial expression for the orbifold instanton partition function as a sum over $r$-vectors of $\sGamma$-coloured solid partitions, from which we deduce its symmetries and properties that are enforced in our conjectures. Then we require that they correctly reduce to the known expressions on $\FC^3/\sGamma$. 

One of the novelties in eight dimensions is that, contrary to the standard noncommutative Donaldson--Thomas invariants~\cite{Szendroi:2007nu}, our invariants depend explicitly on the framing vector $\vec r$, and contrary to the invariants of~\cite{Quiver3d,Cirafici:2011cd}, the $\vec r$ dependence here is genuine and cannot simply be encoded in a multiplicative factor. 
In  the case of the orbifold $\FC^2/ \RZ_2\times \FC^2$ we obtain a conjectural formula for  the equivariant  instanton partition function  of type $\vec{r}=(r_0,r_1)$ for the rank $r$ cohomological gauge theory with $r$ massive fundamental matter fields on $[\FC^2/\RZ_2]\times \FC^2$. It is given by  (Conjecture~\ref{con:r0r1})
\begin{align}
\begin{split}
Z_{[\FC^2/\RZ_2]\times\FC^2}^{\vec{r}}(\vec \qu;\vec{\epsilon},\vec{m}) &= M(\Qu)^{-2\,\frac{m\,r\,\epsilon_{12}\,\epsilon_{13}\,\epsilon_{3}}{\epsilon_1\,\epsilon_2\,\epsilon_3\,\epsilon_4}-\frac{3}{2}\,\frac{m\,r\,\epsilon_{12}}{\epsilon_1\,\epsilon_2}} \\
& \quad\, \times \widetilde{M}(-\qu_1,\qu_0\,\qu_1)^{-\frac{m^0\,r_0\,\epsilon_{12}}{\epsilon_3\,\epsilon_4}} \ \widetilde{M}(-\qu_0,\qu_0\,\qu_1)^{-\frac{m^1\,r_1\, \epsilon_{12}}{\epsilon_3\,\epsilon_4}} \ ,
\end{split}
\end{align}
where $\vec m=(m^0,m^1)$ is the $\sGamma$-module decomposition of the center of mass parameter \eqref{eq:massdef}.

As offsprings of our conjectures, we obtain predictions for the rank $r$ Donaldson--Thomas partition functions on the K\"ahler three-orbifolds $\FC^2/\RZ_n\times\FC$ and $\FC^3/\RZ_2\times\RZ_2$ with generic $\sU(3)$-holonomy, in the cases where they are not yet known. See Propositions~\ref{prop:C2ZnU3}, \ref{prop:C2Z2CU3}, \ref{prop:6dZC2C2} and~\ref{prop:C3Z2Z2U3}.

The $\varOmega$-deformation also enables the definition of the ``pure'' cohomological gauge theory on $\FC^4$ without the massive fundamental matter fields. The instanton partition function for the pure gauge theory is related to~\eqref{eq:result_1} by the infinite mass limit of Proposition~\ref{prop:puremassiverel} which decouples the fundamental hypermultiplet.  For rank $r=1$ this gives (Corollary~\ref{prop:pureC4})
\begin{align}
Z_{\FC^4}(\Lambda;\vec\epsilon\,)^{\rm pure}= \exp\Big(-\Lambda\ \frac{\epsilon_{12}\,\epsilon_{13}\,\epsilon_{23}}{\epsilon_1\,\epsilon_2\,\epsilon_3\,\epsilon_4}\,\Big) \ ,
\end{align}
where $\Lambda$ is the ultraviolet scale of the quantum field theory, whereas for higher rank $r>1$ the instanton contributions all vanish. 

Analogously, we define the ``pure'' gauge theory on orbifolds. The orbifold instanton partition function for the pure gauge theory is related to the orbifold instanton partition function of the cohomological gauge theory with massive fundamental matter by the infinite mass limit of Proposition~\ref{Prop4}. For rank $r>1$ the orbifold instanton contributions also all vanish (Propositions~\ref{prop:pureorbZn}~and~\ref{prop:pureorbZ2Z2}).

\begin{remark}
As the writing of this paper was nearing completion, a related work by Kimura appeared~\cite{Kimura:2022zsm}, which writes down the same matrix integral expression \eqref{eq:Orb_Zin} from a different perspective for the instanton partition function on the orbifold $\FC^4/\sGamma$, in the special case where $\sGamma=\RZ_n$ is a generic cyclic subgroup of $\sSU(4)$. Kimura also gives a free field representation of the orbifold instanton partition function in that case, in the spirit of the BPS/CFT correspondence~\cite{Nekrasov:2015wsu}, but does not discuss the explicit evaluation of the contour integrals nor the geometrical applications to Donaldson--Thomas theory.
\end{remark}

\subsubsection*{Relation to the Work of Cao, Kool and Monavari}

Our work is complementary to new independent mathematical work by Cao, Kool and Monavari~\cite{CKMpreprint} who consider Donaldson--Thomas invariants of general Calabi--Yau four-orbifolds using an algebro-geometric approach. In the local toric model, they compute rank one Donaldson--Thomas invariants using equivariant localization and a vertex formalism on quotient singularities $\FC^4/\sGamma$, where $\sGamma$ is a finite abelian subgroup of $\sSL(4,\FC)$. When $\sGamma$ is the Kleinian subgroup $\RZ_n\subset \sSL(2,\FC$) or the subgroup $\RZ_2\times \RZ_2\subset \sSL(3,\FC)$, they also conjecture closed formulas for the equivariant K-theoretic partition functions, generalizing the Nekrasov--Piazzalunga partition function~\cite{m4c}. They recover the cohomological invariants as well as three-orbifold invariants via the dimensional reduction of~\cite{Cao:2019tvv}. 

They also consider the Pandharipande--Thomas invariants of the crepant resolutions of $\FC^4/\sGamma$ (when they exist), including the cases of non-abelian orbifold groups $\sGamma$. They again conjecture closed formulas for the equivariant K-theoretic partition functions, as well as relate the Donaldson--Thomas invariants for orbifold groups $\RZ_n$ and $\RZ_2\times\RZ_2$ with the Pandharipande--Thomas invariants of the corresponding crepant resolutions by certain changes of variables.

Their work complements ours by developing the rank one K-theory version of the story from a rigorous purely algebro-geometric perspective, whereas we also develop the higher rank cohomological Donaldson--Thomas invariants of toric Calabi--Yau four-orbifolds from the viewpoint of quantum field theory. Up to slightly different conventions, our results perfectly match. More precisely, we can summarise the main agreements as follows:
\begin{myitemize}
\item Conjecture~\ref{con2} agrees with the cohomological limit of~\cite[Conjecture~{5.13}]{CKMpreprint} (see~\cite[{Corollary~6.6}]{CKMpreprint}).
\item The rank one case of Proposition~\ref{prop:pureorbZn} agrees with~\cite[{Corollary~6.10}]{CKMpreprint}.
\item Conjecture~\ref{con3} agrees with the cohomological limit of~\cite[Conjecture~{5.14}]{CKMpreprint} (see~\cite[{Corollary~6.6}]{CKMpreprint}).
\item The rank one case of Proposition~\ref{prop:pureorbZ2Z2} agrees with~\cite[{Corollary~6.10}]{CKMpreprint}.
\end{myitemize}

\subsubsection*{Outline}

Throughout our presentation we gloss over numerous technical subtleties without comment. The organisation of the remainder of this paper can be briefly summarised as follows:
\begin{myitemize}
\item In Section~\ref{sec:Instantons} we provide an elaborate review of the construction of an eight-dimensional cohomological gauge theory for the holonomy group $\sSU(4)$. We study its instanton solutions, discuss properties of the instanton moduli space and define an ADHM-type parametrization. We evaluate the equivariant instanton partition function using a quiver matrix model for the ADHM data and the tangent-obstruction complex of the instanton moduli space. We discuss the relations to Donaldson--Thomas invariants of $\FC^4$ and the analogous enumerative theory in six dimensions.
\item In Section~\ref{Orb8d} we study instantons on orbifolds $\mathbbm{C}^4/\sGamma$ with $\sGamma$ a finite abelian subgroup of $\sSL(4,\FC)$. We extend the results of Section~\ref{sec:Instantons} to the cohomological gauge theories on quotient stacks $[\FC^4/\sGamma]$ for generic $\sGamma$-actions, providing numerous explicit examples. We also give a supplementary overview of crepant resolutions of~$\FC^4/\sGamma$ as well as their relevance to the BPS state counting problems for D-brane bound states and instantons.
\item In the final two Sections~\ref{sec:20C2ZnC2}~and~\ref{sec:30C3Z2Z2C} we focus on two examples of orbifolds where it appears possible to explicitly sum the instanton series, namely $\FC^2/{\RZ_n}\times\FC^2$ and $\FC^3/(\RZ_2\times\RZ_2)\times\FC$. We conjecture closed infinite product formulas for the orbifold instanton partition functions in terms of combinations of generalized MacMahon functions. 
\item Appendix~\ref{app:ADHMconstruction} presents some details of the generalized ADHM construction of the instanton moduli space, while in Appendix~\ref{app:closedformulas} we sketch possible analytic proofs of the closed expressions for the instanton partition functions on $\FC^4$ for general rank $r$, as well as on the orbifold $\FC^2/{\RZ_n}\times\FC^2$ in the rank one case.
\end{myitemize}

\subsubsection*{Acknowledgements}

We thank Yalong Cao, Martijn Kool and Sergej Monavari for helpful discussions and correspondence, and for sharing a draft of their manuscript~\cite{CKMpreprint} with us prior to submission. {M.T.} is particularly grateful to Sergej Monavari for valuable comments  on our manuscript and enlightening conversations. We thank an anonymous referee whose comments prompted us to clarify various aspects of our presentation. R.J.S. is grateful to the Mainz Institute for Theoretical Physics (MITP) of the DFG Cluster of Excellence PRISMA${}^+$ (Project ID 39083149) for its hospitality and support during the completion of this work. The work of {\sc R.J.S.} was supported in part by the STFC Consolidated Grant ST/P000363/1. The work of {\sc M.T.} is supported by an EPSRC Doctoral Training Partnership grant. 

\section{Eight-Dimensional Cohomological Gauge Theory}\label{sec:Instantons}

In this section we review and elaborate on several well-known aspects of generalized instantons and their equivariant partition functions in cohomological gauge theory on $\FC^4$. We do this in some detail, as the treatment is needed and is parallel to the set-up used in discussing the orbifold theories. In particular, we discuss in detail $\sSpin(7)$-instantons in noncommutative gauge theory and their ADHM parametrizations, which also figure in the generalized gauge theories of~\cite{Nekrasov:2016qym,Pomoni:2021hkn}, where some explicit solutions are also discussed. It is instructive for the reader to keep in mind the analogous treatments in lower dimensions; see~\cite{Szabo:2022zyn} for a review and comparison of instanton counting theories in four, six and eight dimensions, and~\cite{Kanno:2020ybd} for the analogue discussion of ADHM-type quiver matrix models.

\subsection{$\sSpin(7)$-Instanton Equations}
\label{sec:noncom_inst}

Let $M\simeq\FR^8$ be an oriented eight-dimensional real vector space endowed with the standard flat Euclidean metric. Instantons on $M$ are solutions to first order self-duality equations for gauge fields, which are preserved by a holonomy group $\sH\subset \sSO(8)$~\cite{8d,8ds}. 
The largest possible holonomy group is $\sSpin(7)\subset \sSO(8)$~\cite{Matrix}, where  the $\sSpin(7)$-structure is specified by the closed non-degenerate Cayley four-form $\Phi$.  It can be expressed in an oriented orthonormal coframe $\{e^\mu\}_{\mu=1}^8$ of $M\simeq\mathbbm{R}^8$, where the metric is $\delta_{\mu\nu}\,e^\mu\otimes e^\nu$, as~\cite{spin}
\begin{equation}\begin{split}
    \Phi&=e^{1256}+ e^{1278} + e^{3456} + e^{3478} + e^{1257}-e^{1368}-e^{2456} \\ & \quad \, +e^{2468} -e^{1458}-e^{1467}-e^{2358} -e^{2367}+e^{1234}+e^{5678} \ , \label{omega}\end{split}
\end{equation}
where we generally use the shorthand notation $e^{\mu_1\cdots\mu_n}:=e^{\mu_1}\wedge \cdots\wedge e^{\mu_n}$.
The four-form $\Phi$ is self-dual, $\ast\,\Phi=\Phi$, where $\ast:\midwedge^k\,\FR^8\longrightarrow\midwedge^{8-k}\,\FR^8$ is the Hodge duality operator on the space of $k$-forms associated to the metric on $M$.

The generalized instanton equations can  be written as
\begin{align}\label{eq:spin7inst}
\lambda\, F = \ast\,(\Phi\wedge F) \ ,
\end{align}
where $F=\dd A+A\wedge A$ is the curvature two-form of a $\sU(r)$ gauge connection $A$ on $M$. The  eigenvalues $\lambda=-3$ and $\lambda=1$ correspond respectively to the two eigenspaces of the self-adjoint operator $\ast\,(\Phi\wedge-)$ in the decomposition of the $28$-dimensional space of two-forms $\midwedge^2\,\FR^8$ into two irreducible representations $\mathbf{28}|_{\sSpin(7)}=\mathbf{7}\oplus \mathbf{21}$ of $\sSpin(7)$. In this paper we focus on solutions of \eqref{eq:spin7inst} with $\lambda=1$, which are called $\sSpin(7)$-instantons and form the basis of Donaldson--Thomas theory in eight dimensions~\cite{Donaldson:1996kp,Donaldson:2009yq}.  

The Bianchi identity $DF=\dd F+A\wedge F=0$ together with \eqref{eq:spin7inst} for $\lambda=1$ imply
\begin{align}
   D \ast F = D(\Phi\wedge F) = \Phi\wedge DF = 0 \ ,
\end{align}
where we used $\dd\Phi=0$. Thus
every solution of the $\sSpin(7)$-instanton equations is a solution of the second order Yang--Mills equations. A typical solution has moduli parametrizing the center of the instanton in $\FR^8$ and its scale $\ell\in\FR_{>0}$ (see e.g.~\cite{Lotay:2022dty}). Thus the instanton moduli space $\frM_r$, i.e.~the space of smooth solutions to \eqref{eq:spin7inst} for $\lambda=1$ modulo gauge transformations, is not compact because instantons can be shrunk to zero size ($\ell\rightarrow0^+$) where they become singular. Deforming the equations \eqref{eq:spin7inst} to noncommutative instanton equations provides an ultraviolet regularization which resolves these small instanton singularities. 

The noncommutative deformation amounts to quantizing a constant Poisson structure on $\FR^8$ defined by a bivector $\theta$ of maximal rank. After using an $\sSO(8)$ rotation to transform $\theta$ into its Jordan canonical form and a suitable rescaling of coordinates, it defines a complex structure on $\FR^8$. Conversely, we can choose one of the seven complex structures on $\FR^8$, consider the associated K\"ahler form $\omega$ corresponding to the metric $\delta_{\mu\nu}\,e^\mu\otimes e^\nu$, and set $\theta=\zeta\,\omega^{-1}$, where $\zeta$ is a positive real parameter. The noncommutative deformation introduces an additional length scale $\zeta$ which serves as a minimum size for the instantons, so it resolves the moduli space singularity as instantons can no longer reach the singularity where their size vanishes. The introduction of a complex structure $J$ breaks the holonomy group to $\sSU(4)\simeq\sSpin(6)\subset \sSpin(7)$, so that the instanton equations are now defined on the Calabi--Yau fourfold~$\FC^4$. 

Let $\Theta^a=e^{2a-1}+\ii\,e^{2a}$ and $\Theta^{\Bar a}=\Bar\Theta^a$ be a coframe on $M=\FC^4$ which forms a basis of $(1,0)$-forms and $(0,1)$-forms with respect to $J$, that is
\begin{align}\label{eq:complexJ}
J\,\Theta^a = \ii\,\Theta^a \qquad \mbox{for} \quad 
a \ \in \ \ulfour := \{1,2,3,4\} \ .
\end{align}
Then there is a non-degenerate $(4,0)$-form $\Omega$ and a $(1,1)$-form $\omega$ defined by~\cite{Popov:2010rf}
\begin{align}\label{eq:Omegaomega}
\Omega = \Theta^1\wedge\Theta^2\wedge\Theta^3\wedge\Theta^4 \qquad \mbox{and} \qquad \omega = \tfrac\ii2\,\big(\Theta^1\wedge\Theta^{\Bar1} + \Theta^2\wedge\Theta^{\Bar2} +\Theta^3\wedge\Theta^{\Bar3} +\Theta^4\wedge\Theta^{\Bar4}\big) \ ,
\end{align}
which obey
\begin{align}\label{eq:omegarels}
\Omega\wedge\omega = 0 \qquad \mbox{and} \qquad \Omega\wedge\Bar\Omega = \tfrac23\, \omega^4 = 16 \, {\tiny\tt vol} \ .
\end{align}
The pair $(\omega,\Omega)$ defines an $\sSU(4)$-structure on $M$ and there is a compatible $\sSpin(7)$-structure determined by
\begin{align}\label{eq:Spin7SU4}
    \Phi=\tfrac{1}{2}\,\omega\wedge\omega-{\rm Re}\,\Omega \ ,
\end{align}
 which coincides with the Cayley four-form \eqref{omega}. The corresponding K\"ahler metric is $\omega\circ J = \delta_{a\bar a}\,\Theta^a\otimes\Theta^{\bar a}$.
 
We can write the field strength $F$ in the basis of $(1,0)$-forms and $(0,1)$-forms as
\begin{align}\label{eq:Fdecomp}
F = F^{2,0} + F^{1,1} + F^{0,2} = \tfrac12\, F_{ab}\,\Theta^a\wedge\Theta^b + F_{a\Bar b}\,\Theta^a\wedge\Theta^{\Bar b} + \tfrac12\,F_{\Bar a\Bar b} \, \Theta^{\Bar a}\wedge \Theta^{\Bar b} \ .
\end{align}
When the $\sSpin(7)$-structure is determined by an $\sSU(4)$-structure via \eqref{eq:Spin7SU4}, the $\sSpin(7)$-instanton equations \eqref{eq:spin7inst} can be reduced via the inclusion $\sSU(4)\subset \sSpin(7)$ to the equations~\cite{Donaldson:1996kp,Donaldson:2009yq}
\begin{align}\label{in}
F_{ab} = \tfrac12\,\varepsilon_{ab\Bar c\Bar d} \, F_{\Bar c\Bar d} \qquad \mbox{and} \qquad \delta^{a\Bar a} \, F_{a\Bar a} = 0 \ ,
\end{align}
for $a,b\in\ulfour$, where $\varepsilon_{ab\bar c\bar d}$ is the Levi--Civita tensor in four dimensions.

The $\sSpin(7)$-instanton equations \eqref{in} are weaker than the Hermitian Yang--Mills equations that appear in the analogue instanton equations in four and six dimensions, which amount to replacing the first equation of \eqref{in} with $F_{ab}=0$ and describe semistable holomorphic vector bundles on $M$. Any solution of the Hermitian Yang--Mills equations is automatically a solution of the $\sSpin(7)$-instanton equations, but not conversely in general. 

\subsection{Cohomological Gauge Theory}\label{sec:coho_gauge_theory}

The cohomological Yang--Mills theory which describes the instanton moduli problem in eight dimensions corresponding to the $\sSU(4)$-structure defined by \eqref{eq:Omegaomega} was constructed in~\cite{top,special} using the BRST formalism to gauge fix the topological action functional
\begin{align}\label{eq:topaction}
    S_0:=\int_{M} \,\Omega\wedge \Tr\big(F^{0,2}\wedge F^{0,2} \big) \ ,
\end{align}
where $\Tr$ denotes the trace in the fundamental representation of the (complexified) gauge group. It can be equivalently obtained by dimensional reduction from ten-dimensional $\mathcal{N}=1$ supersymmetric Yang--Mills theory with gauge group $\sU(r)$~\cite{reduction}, whose field content consists of the gauge field $A$ and a Majorana--Weyl spinor in the adjoint representation of $\sU(r)$. This also identifies it as the low-energy effective field theory on D-branes in type II string theory. 

One starts by compactifying the time dimension and the ninth spatial dimension of $\FR^{1,9}$ on a square torus $\mathbbm{T}^{1,1} = \FR^{1,1}\,\big/\,2\pi\,R\,\RZ^2$ of radius $R$. Accordingly, the field theory is defined on the spacetime $\mathbbm{R}^{8}\times \mathbbm{T}^{1,1}$, and in the limit $R\rightarrow 0$ the unbroken global symmetry group is
\begin{align}\label{eq:unbroken}
    \sSO(1,9) \ \supset  \ \sSO(8)\times \sSO(1,1) \ .
\end{align}
The components of the gauge field $A_0$ and $A_9$ become a pair of adjoint scalar fields on $M=\FR^8$ after dimensional reduction. The reduced field theory has a non-compact R-symmetry group $\sSO(1,1)$ which in the bosonic sector acts only on these scalars. The Majorana--Weyl spinor field reduces according to the decomposition of the  positive chirality real spinor representation of $\sSO(1,9)$ into irreducible representations of the subgroup \eqref{eq:unbroken} given by
\begin{align}
    {\mbf 16}_+\big|_{\sSO(8)\times \sSO(1,1)} = \mbf{( 8}{}_{\rm s}\mbf{,+1)} \, \oplus \, \mbf{( 8}{}_{\rm c}\mbf{,-1)} \ , \label{deco}
\end{align}
where ${\mbf 8}_{\rm s}$ and ${\mbf 8}_{\rm c}$ are respectively the  chiral and antichiral real spinor representations of $\sSO(8)$. 

With the reduced holonomy group $\sSU(4)\subset \sSpin(7)$, the ambient space becomes the Calabi--Yau fourfold $\mathbbm{C}^4$, characterized by a complex structure $J$ and a compatible $\sSU(4)$-structure $(\omega,\Omega)$ from \eqref{eq:complexJ}--\eqref{eq:omegarels}. The doublet of real scalars are naturally grouped into a complex Higgs field $\varphi$.
Under $\sSU(4)$ the curvature two-form $F$ decomposes as in \eqref{eq:Fdecomp}.

The $\sSU(4)$ holonomy preserves two supercharges. Under $\sSU(4)$ the spinor representations $\mathbf{8}_{\rm s}$ and $\mathbf{8}_{\rm c}$ branch into irreducible representations as
\begin{align}\label{eq:SU4branching}
\begin{split}
\mathbf{8}_{\rm s}\big|_{\sSU(4)}&=\mbf 1\oplus \mbf 6\oplus\mbf 1 \simeq \midwedge^{0}\,\FC^4 \,\oplus\, \midwedge^{2,0}\,\FC^4 \,\oplus\,\midwedge^{4,0}\,\FC^4 \ , \\[4pt]
\mathbf{8}_{\rm c}\big|_{\sSU(4)}&=\mbf 4\oplus\Bar{\mbf 4} \simeq \midwedge^{1,0}\,\FC^4 \,\oplus\,\midwedge^{3,0}\,\FC^4 \ .
\end{split}
\end{align}
Accordingly, in the fermionic sector the spectrum consists of a complex scalar and a $(2,0)$-form, along with $(1,0)$-forms and $(0,1)$-forms. This matches the field content of the bosonic spectrum, which consists of the complex gauge field
\begin{align}
A=A^{1,0}+A^{0,1}=A_a\,\Theta^a+A_{\Bar{a}}\,\Theta^{\Bar{a}}
\end{align}
and Higgs field $\varphi$ along with their complex conjugates. This recovers the field content of the holomorphic $\CN_{\textrm{\tiny T}}=2$ topological Yang--Mills theory in eight dimensions~\cite{top,special}.

Here we shall focus on the purely bosonic sector of the dimensionally reduced field theory on $\FC^4$. For this, let $\ast:\midwedge^{k,l}\,\FC^4\longrightarrow \midwedge^{4-l,4-k}\,\FC^4$ be the map induced by the Hodge duality operator. Together with the holomorphic four-form $\Omega$, it defines an anti-linear operator $\star_\Omega:\midwedge^{0,k}\,\FC^4\longrightarrow\midwedge^{0,4-k}\,\FC^4$ by $\star_\Omega\,\beta:=\overline{\ast(\beta\wedge\Omega)}$ for $\beta\in\midwedge^{0,k}\,\FC^4$. With $\frg:= \mathfrak{gl}(r,\FC)$, this operation has the property that the standard inner product of $\alpha,\beta\in\midwedge^{0,k}\,T^*\FC^4\otimes\frg$ can be expressed as
\begin{align}\label{eq:Omegainnerproduct}
\langle\alpha,\beta\rangle := \int_{M}\, \Tr\big(\alpha\wedge\ast\,\beta^\dag\big) = \int_{M}\,\Omega\wedge\Tr\big(\alpha\wedge\star_\Omega\,\beta\big) \ .
\end{align}
There is a related gauge-invariant complex quadratic form  defined for $\alpha\in\midwedge^{0,2}\,T^*\FC^4\otimes\frg$ by
\begin{align}\label{eq:Omegaquadraticform}
(\alpha,\alpha)_\Omega := \int_M\,\Omega\wedge\Tr(\alpha\wedge \alpha) \ .
\end{align}

The Yang--Mills Lagrangian density on $M$ can be written as ${\tiny\tt vol} \, \Tr\big(-\frac14\,F_{\mu\nu}^2\big)=\frac12\,\Tr(F\wedge\ast F)$, and so the Yang--Mills action functional is expressed in terms of the norm induced by the inner product \eqref{eq:Omegainnerproduct} as $\frac12\,\|F\|^2$. This may be rewritten for gauge fields with suitable boundary conditions at infinity and vanishing first Chern form (i.e. $F\in\midwedge^2\,T^*\FR^8\otimes\mathfrak{su}(r)$)  using the useful identity~\cite{prop}
\begin{align}
\tfrac12\,\|F\|^2= 2\,\big\|F^{0,2}\big\|^2+\tfrac12\,\|\omega\,\lrcorner\,F\|^2+\mathcal{K}\ ,
\end{align}
where $\omega\,\lrcorner\,F = \ast\,(\omega\wedge\ast\, F)$ is the contraction of the two-form $F$ with the $(1,1)$-form $\omega$, while $\mathcal{K}=6\,\int_M\,\omega^2\wedge\Tr(F^2)$ is a topological term.
This enables us to write the purely bosonic part of the dimensionally reduced action functional in the form
\begin{align}\label{eq:Sbos}
S = 2\,\big\|F^{0,2}\big\|^2 +\tfrac12\,\|\omega\,\lrcorner\,F\|^2 + \|D\varphi\|^2 + V(\varphi,\bar\varphi) + \mathcal{K} \ ,
\end{align}
where $V(\varphi,\bar\varphi)=\frac12\,\int_M\,{\tiny\tt vol} \ \Tr\big([\varphi,\bar\varphi]^2\big)$ is the Higgs potential. 

The involution $\star_\Omega:\midwedge^{0,2}\,\FC^4\longrightarrow\midwedge^{0,2}\,\FC^4$ enables an orthogonal decomposition of the field strength \smash{$F^{0,2} =\bar\partial A^{0,1}+A^{0,1}\wedge A^{0,1} = F_+^{0,2}+F_-^{0,2}$} into eigencurvatures as
\begin{align}
F_\pm^{0,2} = \tfrac12\,\big(F^{0,2}\pm\star_\Omega\, F^{0,2}\big) \qquad \mbox{with} \quad \star_\Omega F_\pm^{0,2} = \pm\,F_\pm^{0,2} \ .
\end{align}
This self-duality condition arises from the reality of the representation $\bf 6$ in the branching of ${\bf 8}_{\rm s}$ from \eqref{eq:SU4branching}, which yields a splitting 
\begin{align}
\midwedge^{0,2}\,\FC^4=\midwedge_+^{0,2}\,\FC^4\oplus\midwedge_-^{0,2}\,\FC^4
\end{align}
into \emph{real} $\pm1$-eigenspaces $\midwedge^{0,2}_\pm\,\FC^4$ of $\star_\Omega$, since $(\star_\Omega)^2=\ident$. 

The quadratic form \eqref{eq:Omegaquadraticform} is maximally positive/negative definite on \smash{$\midwedge_\pm^{0,2}\,T^*\FC^4\otimes\frg$}. This enables us to write the topological action functional \eqref{eq:topaction} as
\begin{align}\label{eq:S0def}
S_0=\big(F^{0,2}, F^{0,2} \big)_\Omega = \big\|F_+^{0,2}\big\|^2 - \big\|F_-^{0,2}\big\|^2 \ .
\end{align}
From $\big\|F^{0,2}\big\|^2 = \big\|F_+^{0,2}\big\|^2 + \big\|F_-^{0,2}\big\|^2 = 2\, \big\|F_-^{0,2}\big\|^2 + S_0$ we can then bring \eqref{eq:Sbos} to the form
\begin{align}\label{eq:Sbosfinal}
S = 4\,\big\|F_-^{0,2}\big\|^2 +\tfrac12\,\|\omega\,\lrcorner\,F\|^2 + \|D\varphi\|^2 + V(\varphi,\bar\varphi) + 2\,S_0 +  \mathcal{K} \ .
\end{align}
In particular, it follows that $S \geq 2\,S_0 + \mathcal{K}$.

Since $\sSU(4)\subset  \sSpin(7)$, and dimensional reduction of ten-dimensional $\CN=1$ supersymmetric Yang--Mills theory on a manifold of $\sSpin(7)$-holonomy is equivalent to a topological twist of the resulting eight-dimensional supersymmetric gauge theory~\cite{special}, this field theory is cohomological and has a BRST symmetry. Hence in the topological sector with $2\,S_0+\mathcal{K}$ fixed, it localizes onto the moduli space of solutions of the BRST fixed point equations given by
\begin{align}\label{eq:F-02}
    F_{-}^{0,2}=0 \ , \quad \omega\,\lrcorner\,F=0 \qquad \mbox{and} \qquad 
    D\varphi=0 \ .
\end{align}
The first equation implies $F_{ab}=\frac{1}{2}\,\varepsilon_{ab\bar{c}\bar{d}}\,F_{\Bar{c}\Bar{{d}}}$, while the second equation  gives $\delta^{a\bar a}\,F_{a\bar a}=0$. 
In other words, the gauge theory localizes on the instanton equations \eqref{in} on the Calabi--Yau fourfold $\mathbbm{C}^4$. The solutions of these equations minimize the action functional \eqref{eq:Sbosfinal}. When $S_0=0$, the $\sSpin(7)$-instanton equations are equivalent to the Hermitian Yang--Mills equations.

\subsection{Equivariant Gauge Theory}
\label{sec:instantoncounting}

We shall now describe the moduli space $\frM_r$ of solutions to the $\sSpin(7)$-instanton equations \eqref{in}, and use this to compute the instanton partition function of the eight-dimensional cohomological gauge theory. We begin with an informal discussion of `integration' over instanton moduli spaces, and then quantify how to compute these formal expressions practically in the remainder of this section through a generalized ADHM parametrization of $\frM_r$. 

Generally, the BRST symmetry of a rank $r$ cohomological gauge theory localizes its path integral (in a fixed topological sector) to an integral over a virtual fundamental class $[\frM_r]^{\rm vir}$,  which is represented by a coarse moduli space $\frM_r$ and a vector bundle $\Ob_r\longrightarrow\frM_r$ called the obstruction bundle, whose fibres are spanned by the antighost zero modes; it is non-trivial when there is a quadratic Kuranishi map obstructing the extension of first order deformations, parametrized by the fibres of the tangent bundle $T\frM_r\longrightarrow\frM_r$, to second order. If $\frM_r$ is smooth and oriented (which is seldom the case), then $[\frM_r]^{\rm vir}$ is the Poincar\'e dual of the Euler class $e(\Ob_r)$. Partition functions of the topological field theory are then defined geometrically by integrating the Euler class $e(\CCK)$ of a `matter bundle' $\CCK\longrightarrow\frM_r$ whose rank equals the virtual dimension of the moduli space, i.e.~the rank of the virtual tangent bundle $T^{\rm vir}\frM_r=T\frM_r\ominus\Ob_r$. 
 
In our eight-dimensional theory, we let $P_\Omega^-=\frac12\,(\ident-\star_\Omega):\midwedge^{0,2}\,\FC^4\longrightarrow \midwedge_-^{0,2}\,\FC^4$ be the projection to the $-1$-eigenspace of the involution $\star_\Omega$ induced by the K\"ahler metric and the holomorphic $(4,0)$-form $\Omega$ on the Calabi--Yau fourfold $M=\FC^4$. The local geometry of the instanton moduli space $\frM_r$ is captured by the instanton deformation complex~\cite{special}
\begin{align}\label{eq:defcomplex}
0\longrightarrow\midwedge^0\,T^*\FC^4\otimes\frg \xrightarrow{ \ \bar\partial_A \ } \midwedge^{0,1}\,T^*\FC^4\otimes\frg \xrightarrow{ \ \bar\partial_A^- \ } \midwedge_-^{0,2}\,T^*\FC^4\otimes\frg\longrightarrow 0 \ ,
\end{align}
where the first arrow is an infinitesimal complex gauge transformation, while the second arrow with $\bar\partial_A^-:=P_\Omega^-\circ\bar\partial_A$ is the linearization of the first equation in \eqref{eq:F-02}. Associated to this complex is a local cyclic $L_\infty$-algebra which describes the full cohomological field theory in the BV formalism and $\frM_r$ as the corresponding Maurer--Cartan moduli space.

The cochain complex \eqref{eq:defcomplex} is elliptic and its degree one cohomology represents the (complex) tangent space to the moduli space $\frM_r$ at a fixed holomorphic self-dual connection $A$. We assume that the degree zero cohomology vanishes, i.e.~$\ker\big(\bar\partial_A\big)=0$, which amounts to restricting to irreducible connections with only trivial automorphisms. There is also a non-trivial degree two cohomology which defines the real self-dual obstruction bundle \smash{$\Ob_r^-\longrightarrow\frM_r$}, whose fibre over $A$ is given by ${\rm coker}\big(\bar\partial_A^-\big)\subset \midwedge_-^{0,2}\,T^*\FC^4\otimes\frg$. 

This is the starting point for the construction of a corresponding (real) virtual fundamental class $[\frM_r]^{\rm vir}$ in~\cite{Cao:2014bca,Cao:2015gra,Borisov:2015vha,Cao:2018rft,Bojko:2020rfg,Oh:2020rnj, KiemPark20}. We consider the stratification of the moduli space $\frM_r$ into its connected components $\frM_{r,k}$ which are labelled by the instanton number (fourth Chern class)
\begin{align}
k = \frac1{384\pi^4} \, \int_M \, \Tr(F\wedge F\wedge F\wedge F) \ .
\end{align}
The obstruction bundle correspondingly restricts to real vector bundles \smash{$\Ob_{r,k}^-\longrightarrow\frM_{r,k}$}, which are orientable. The real virtual dimension is $2\,r\,k$ and the virtual fundamental class $[\frM_{r,k}]_\tto^{\rm vir}$ depends on the choice of an orientation $\tto$ of $\Ob_{r,k}^-$. 

The Euler class of the self-dual obstruction bundle, which is naturally induced by integration over the BRST antighost fields in the cohomological gauge theory, can be thought of as defining a square root Euler class for the complexification $\Ob_{r,k}:=\Ob_{r,k}^-\otimes_\FR\FC$ through
\begin{align}
\sqrt{e}(\Ob_{r,k}):= e(\Ob^-_{r,k}) \ .
\end{align}
In~\cite{Oh:2020rnj} an alternative definition of the square root Euler class is given as follows. The complex vector bundle $\Ob_{r,k}$ carries a natural quadratic form induced by \eqref{eq:Omegaquadraticform}, which restricts to a metric on \smash{$\Ob_{r,k}^-$}. Then $\sqrt{e}(\Ob_{r,k})$ can be defined using a maximally isotropic holomorphic sub-bundle of $\Ob_{r,k}$ with respect to this quadratic form, which is isomorphic to \smash{$\Ob_{r,k}^-$} as a real bundle. This comes with a sign determined by the choice of orientation. It has the advantage of yielding a class that lifts to Chow cohomology (with $\RZ[\frac12]$-coefficients) and provides a construction of \smash{$[\frM_{r,k}]_\tto^{\rm vir}$} as an algebraic cycle. 

A natural choice of `matter bundle' $\CCK$ in this case is provided by the bundle $\CCV_{r,k}\longrightarrow\frM_{r,k}$ whose fibre over a gauge orbit $[A]$ is the complex vector space $V$ entering the generalized ADHM parametrization of $\frM_{r,k}$ in Section~\ref{sec:ADHMC4} below; this is a complex vector bundle of rank $k$, and taking $r$ tensor powers of it gives a vector bundle $\CCK$ of the desired real rank $2\,r\,k$. To incorporate masses $\vec m=(m_1,\dots,m_r)$ for the matter fields we use the Chern polynomial of the vector bundle $\CCV_{r,k}$ which is defined by the formula
\begin{align}
c(\CCV_{r,k};m) = \sum_{i=0}^k \, m^i \, c_{k-i}(\CCV_{r,k}) \ ,
\end{align} 
where $c_i(\CCV_{r,k})$ is the $i$-th Chern class of $\CCV_{r,k}$. The Euler class is the top Chern class $e(\CCV_{r,k}) = c_k(\CCV_{r,k})$.

The instanton partition function of the eight-dimensional cohomological gauge theory then schematically has the form
\begin{align}\label{eq:ZC4rqformal}
Z_{\FC^4}^{r}(\qu;\vec m) = \sum_{k=0}^\infty \, \qu^k \ \int_{[\frM_{r,k}]_\tto^{\rm vir}} \ \prod_{l=1}^r \, c(\CCV_{r,k};m_l) \ ,
\end{align}
where the counting parameter $\qu$ weighs the instanton number and is determined by the couplings of the gauge theory.

To make sense of the symbolic expression \eqref{eq:ZC4rqformal}, we will work in the setting of equivariant gauge theory and \emph{define} it via an equivariant integral over the instanton moduli space. The global symmetry group of the cohomological field theory with matter is
\begin{align}
\sG = \sU(r)_{\rm col} \, \times \, \sSU(4) \, \times \, \sU(r)_{\rm fla} \ ,
\end{align}
where the colour symmetry $\sU(r)_{\rm col}$ acts by rotating the framing of the gauge bundle at infinity in $\FC^4$, and the flavour symmetry $\sU(r)_{\rm fla}$ acts on its vector representation $\mbf r$. This group can be rotated into its maximal torus
\begin{align}
\sT = \sT_{\vec a} \, \times \, \sT_{\vec\epsilon} \, \times \, \sT_{\vec m} \ ,
\end{align}
where we label the torus factors by the corresponding complexified Cartan subalgebra elements, which we assume to consist of generic complex parameters in the following. 

The framing rotation parameters $\vec a=(a_1,\dots,a_r)$ are vacuum expectation values of the complex Higgs field $\varphi$, while $\vec\epsilon = (\epsilon_1,\epsilon_2,\epsilon_3,\epsilon_4)$ satisfy
\begin{align}\label{eq:CY4constraint}
\epsilon_1+\epsilon_2+\epsilon_3+\epsilon_4=0
\end{align}
and correspond to the natural complex scaling action of the three-torus $\sT_{\vec\epsilon}$ on $\FC^4$. This couples the gauge theory to Nekrasov’s $\varOmega$-background~\cite{SWcounting} through an $\sSU(4)$-invariant deformation of the BRST supercharges. The $\varOmega$-deformation does not change the instanton equations \eqref{in}, but it confines their solutions to the fixed point $0\in\FC^4$ and so provides an infrared regularization of the singularities of $\frM_{r,k}$ due to instantons that escape to infinity.
The flavour rotation parameters $\vec m = (m_1,\dots,m_r)$ are masses of $r$ fundamental matter fields. 

The torus $\sT$ acts on the moduli space $\frM_{r,k}$ and its obstruction bundle $\Ob_{r,k}$ (with $\sT_{\vec m}$ acting trivially), as well as on the vector bundle $\CCV_{r,k}\otimes\mbf r$ (with $\sT_{\vec m}$ acting trivially on $\CCV_{r,k}$ and $\sT_{\vec a,\vec\epsilon}:=\sT_{\vec a}\times\sT_{\vec\epsilon}$ acting trivially on $\mbf r$). 
The product of Chern polynomials in \eqref{eq:ZC4rqformal} can be regarded as the equivariant Euler class $e_\sT(\CCV_{r,k}\otimes\mbf r)$. The integrations are then interpreted as the pushforwards \smash{$\int_{[\frM_{r,k}]_{\tto}^{\rm vir}}^\sT \, e_\sT(\CCV_{r,k}\otimes\mbf r)$} to a point in the $\sT$-equivariant Chow cohomology of $\frM_{r,k}$, whose coefficient ring is $\FC[\vec a,\vec\epsilon,\vec m]/\langle\epsilon_1+\cdots+\epsilon_4\rangle$. The $\varOmega$-deformation localizes the moduli space $\frM_{r,k}$ onto its isolated $\sT$-fixed points $\vec\sigma\in\frM_{r,k}^\sT$. The virtual localization formula of~\cite{Graber,Oh:2020rnj} computes the integrals in \eqref{eq:ZC4rqformal} as a sum over these fixed points, giving the instanton partition function as a combinatorial expansion in $\qu$ whose terms are rational functions of the equivariant parameters $(\vec a,\vec\epsilon,\vec m)$. It reads as
\begin{align}\label{eq:ZC4matter}
Z_{\FC^4}^r(\qu;\vec a,\vec\epsilon,\vec m) = \sum_{k=0}^\infty \, \qu^k \ \sum_{\vec\sigma\in\frM_{r,k}^{\sT}} \, \frac{\sqrt{e_{\sT}}\big((\Ob_{r,k})_{\vec\sigma}\big) \ e_{\sT}\big((\CCV_{r,k})_{\vec\sigma}\otimes\mbf r\big)}{e_{\sT}\big(T_{\vec\sigma}\frM_{r,k}\big)} \ .
\end{align}

The fixed points of the instanton moduli spaces have a combinatorial significance that we discuss below. The  square root Euler class $\sqrt{e_\sT}\big((\Ob_{r,k})_{\vec\sigma}\big)$ in \eqref{eq:ZC4matter} is defined up to a sign which depends explicitly on the orientation of \smash{$\Ob_{r,k}^-$} at the fixed point $\vec\sigma$. The formula \eqref{eq:ZC4matter} assumes that the spaces $T_{\vec\sigma}\frM_{r,k}$ and $(\Ob_{r,k})_{\vec\sigma}$ carry only non-zero weights for the action of the torus $\sT_{\vec a,\vec\epsilon}\,$, so that the equivariant Chow cohomology classes $e_\sT(T_{\vec\sigma}\frM_{r,k})$ and $\sqrt{e_\sT}\big((\Ob_{r,k})_{\vec\sigma}\big)$ are invertible. We shall elaborate further on these points later on in this section.
 
The equivariant theory also allows for the definition of the instanton partition function of the cohomological gauge theory without matter~\cite{m4}, that is, the ``pure'' holomorphic $\CN_{\textrm{\tiny T}}=2$ topological Yang--Mills theory in eight dimensions defined by integrating $1$ over $[\frM_{r,k}]_\tto^{\rm vir}$. This is analogous to the definition of the Nekrasov partition function for the pure $\CN=2$ supersymmetric Yang--Mills theory in four dimensions~\cite{SWcounting}. The field theory is now equivariant with respect to the torus $\sT_{\vec a,\vec\epsilon}\subset \sT$, and the localization theorem gives
\begin{align}\label{eq:ZC4pure}
Z_{\FC^4}^r(\Lambda;\vec a,\vec\epsilon\,)^{\rm pure} := \sum_{k=0}^\infty\,\Lambda^k \ \int_{[\frM_{r,k}]_\tto^{\rm vir}}^{\sT_{\vec a,\vec\epsilon}} \, 1 = \sum_{k=0}^\infty\,\Lambda^k \ \sum_{\vec\sigma\in\frM_{r,k}^{\sT}} \, \frac{\sqrt{e_{\sT}}\big((\Ob_{r,k})_{\vec\sigma}\big)}{e_{\sT}\big(T_{\vec\sigma}\frM_{r,k}\big)} \ .
\end{align}

\subsection{Noncommutative Field Theory}\label{sec:ADHMC4}

As we discussed in Section~\ref{sec:noncom_inst}, the noncommutative deformation is equivalent to the choice of a complex structure on $\mathbbm{R}^8$, which induces a Poisson bivector $\theta=\zeta\,\omega^{-1}$ that we wish to quantize, where $\zeta\in\FR_{>0}$ and $\omega$ is the corresponding K\"ahler form of $\FR^8$. 
The quantization map sends the local complex coordinates $z_a,\bar z_{\bar a}$ of $M=\FC^4$ to operators with the commutation relations
\begin{align}
[z_a,z_b]=0=[\bar z_{\bar a},\bar z_{\bar b}] \qquad \mbox{and} \qquad [z_a,\bar z_{\bar b}]=\zeta \, \delta_{a\bar b} \ ,
\end{align}
for $a,\bar a,b, \bar b\in\ulfour$. We write $\CA$ for the noncommutative algebra generated by these operators over $\FC$. The unique irreducible representation of this algebra is given by the Fock module
\begin{align}\label{eq:fock_space}
\CH = \FC[\bar z_1,\bar z_2,\bar z_3,\bar z_4]|\vec0\,\rangle = \bigoplus_{\vec n\in\RZ_{\geq0}^4} \, \FC|\vec n\,\rangle \ ,
\end{align}
with vacuum vector $|\vec 0\,\rangle$.

Let $W\simeq\FC^r$ be the fundamental representation of the colour symmetry $\sU(r)_{\rm col}$, regarded as a Hermitian vector space.
Using the quantization map, we send all fields of the eight-dimensional cohomological gauge theory to operators acting on the separable Hilbert space $\CH^r:=W\otimes\CH$, i.e.~to elements of the algebra ${\sf Mat}_{r{\times}r}(\CA):=\sEnd_\FC(W)\otimes\CA$ of $r{\times}r$ matrices valued in $\CA$. This turns the gauge theory into an infinite-dimensional matrix model. The fixed point equations \eqref{eq:F-02} become operator algebraic equations for the noncommutative fields given by
\begin{align}
    [Z_a,Z_b]=\tfrac{1}{2}\,\varepsilon_{ab\bar c\bar d}\,\big[Z^\dagger_{\Bar{c}},Z^\dagger_{\Bar{d}}\big] \ , \quad \sum_{a=1}^4 \, \big[Z_{\Bar{a}}^\dagger,Z_a\big]= \zeta \, \ident_{\CH^r}  \qquad \mbox{and} \qquad [Z_a,\varphi]=0 \, \label{Zin}
\end{align}
for $a,b\in\ulfour$, where the operators \smash{$Z_a:=\frac1{2}\,(z_a+\ii\,A_{\bar a})$} are called \emph{covariant coordinates}.
The $\varOmega$-deformation leaves unchanged the instanton equations but alters the last equation of \eqref{Zin} to 
\begin{align}
[Z_a,\varphi]=\epsilon_a\,Z_a \ .
\end{align}

It follows that charge $k$ noncommutative $\sU(r)$ instantons in eight dimensions are described as elements of the algebra ${\sf Mat}_{r{\times}r}( \mathcal{A})$ acting on the free $\CA$-module $\CA^r:=W\otimes\mathcal{A}$. Through the natural isomorphism of $\CA$-modules $\mathcal{E}^{r,k}\simeq \mathcal{A}^r$, they can be related to connections $\nabla: \mathcal{E}^{r,k}\longrightarrow  \mathcal{E}^{r,k}\otimes_\CA \Omega_{\mathcal{A}}^1$ on the projective modules $\mathcal{E}^{r,k}:=\mathcal{H}^ { k}\oplus\mathcal{A}^{r}$ over $\CA$, where $\mathcal{H}^ { k}:=V\otimes\CH$ for a Hermitian vector space $V\simeq \FC^k$ and $\Omega_\CA^1$ is the bimodule of one-forms over the algebra $\CA$.
This induces a decomposition of the covariant coordinates $Z_a\in \sEnd_{\mathcal{A}}(\mathcal{E}^{r,k})$ as
\begin{align}\label{eq:Zamatrices}
Z_1={\small\bigg(
    \begin{matrix}
    B_1 & \hat I_1\\
    \hat I'_1 & R_1
        \end{matrix} \bigg)} \normalsize \ ,\quad
        Z_2={\small\bigg(
    \begin{matrix}
    B_2 & \hat I_2'\\
    \hat I_2 & R_2
        \end{matrix}\bigg) } \normalsize \ ,\quad
          Z_3={\small\bigg(
    \begin{matrix}
    B_3 & \hat I_3'\\
    \hat I_3 & R_3
        \end{matrix} \bigg)} \normalsize \qquad \mbox{and} \qquad
          Z_4={\small\bigg(
    \begin{matrix}
    B_4 & \hat I_4'\\
    \hat I_4 & R_4
        \end{matrix} \bigg)} \normalsize \ ,
\end{align}
where the diagonal blocks consist of linear maps $B_{a} \in\sEnd_\mathbbm{C}(V)$ and operators $R_a\in \sEnd_{\mathcal{A}}(W\otimes \mathcal{A})$, while the off-diagonal blocks are operators
\begin{align}\label{off_diagonal_op}
\hat I_1,\hat I_2',\hat I_3',\hat I_4' \ \in \ \sHom_{\mathcal{A}}(W\otimes \mathcal{A},V) \qquad \mbox{and} \qquad \hat I_1',\hat I_2,\hat I_3,\hat I_4 \ \in \ \sHom_{\mathcal{A}}(V,W\otimes \mathcal{A}) \ .
\end{align}
We can set  $\hat I_a'=0$  for $a\in\ulfour$ thanks to the freedom of gauge choice.  

Using the isomorphisms
\begin{align}
\begin{split}
\sHom_{\mathcal{A}}(W\otimes \mathcal{A},V) & \simeq 	\sHom_\FC(W,V) \otimes \sEnd_{\mathcal{A}} (\mathcal{A}^{r})
\end{split}
\end{align}
and
\begin{align}
\begin{split}
\sHom_{\mathcal{A}}(V,W\otimes \mathcal{A}) & \simeq \sHom_\FC(V,W) \otimes \sEnd_{\mathcal{A}} (\mathcal{A}^{ r})
\end{split}
\end{align}
we decompose the operators from \eqref{off_diagonal_op} as $\hat I_a=I_a\otimes \psi_a$, where $I_1\in\sHom_\FC(W,V)$, $I_\alpha\in\sHom_\FC(V,W)$ for $\alpha\in\{2,3,4\}$ and the operators $\psi_a\in \sEnd_{\mathcal{A}} (\mathcal{A}^{ r}) \simeq {\sf Mat}_{r{\times}r}( \mathcal{A})$ satisfy \smash{$\psi_1\,\psi_\alpha = \ident_{\mathcal{A}^r}$}. Relabelling $I:=I_1$ and $J_\alpha:=I_{\alpha+1}$, it is then straightforward to see that the first and second equations of \eqref{Zin} yield the matrix equations
\begin{align} \begin{split}\label{eq:ADHM_IJKL}
[B_1,B_\alpha]+I\,J_\alpha -\tfrac12\,\varepsilon_{1\alpha\bar\beta\bar\gamma}\,\big[B_{\bar\beta}^\dagger,B_{\bar\gamma}^\dagger\big]&=0 \ , \\[4pt]
[B_\alpha,B_\beta]-\tfrac12\,\varepsilon_{1\alpha\bar\beta\bar\gamma}\,\big(\big[B_{\bar 1}^\dagger,B_{\bar\gamma}^\dagger\big] - J_{\bar\gamma}^\dag\,I^\dag\big) &=0 \ , \\[4pt]
\sum_{a=1}^4\,\big[B_a,B_{\bar a}^\dagger\big]+I\,I^\dagger-\sum_{\alpha=1}^3\,J_{\bar\alpha}^\dagger\,J_\alpha&=\zeta\, \ident_{ k} \ , \end{split}
\end{align}
for $\alpha,\beta,\gamma\in\{1,2,3\}$, along with other relations which are not relevant here.  

We observe that the space of solutions to the system of equations \eqref{eq:ADHM_IJKL} is independent of the value of $\zeta>0$ as a consequence of the scaling symmetry $B_a\to\kappa\,B_a$, $I\to\kappa\,I$, $J_\alpha\to\kappa\,J_\alpha$ and $\zeta\to\kappa^2\,\zeta$ for $\kappa\in\FR$. If the $\sSpin(7)$-instanton equations reduce to the Hermitian Yang--Mills equations (i.e. $S_0=0$), then the first instanton equation of \eqref{Zin} would reduce to the holomorphic equations $[Z_a,Z_b]=0$ and the non-holomorphic terms in the first two matrix equations of \eqref{eq:ADHM_IJKL} would drop out. In this instance one can repeat the stability argument of~\cite{coho} to infer that, in the equivariant gauge theory of Section~\ref{sec:instantoncounting}, one may set $J_\alpha=0$ for $\alpha\in\{1,2,3\}$. In fact, it is the equations $[B_a,B_b]=0$ that arise for the vacua in the string theory setting discussed below, which we shall see are equivalent to the non-holomorphic equations with $J_\alpha=0$ that retain the information of the original $\sSpin(7)$-instantons. With this in mind, we restrict the solutions of \eqref{eq:ADHM_IJKL} to the closed subvariety defined by the condition
\begin{align}
J_\alpha=0 \qquad \mbox{for} \quad \alpha\in\{1,2,3\} \ .
\end{align}

Altogether we have shown that the noncommutative instanton equations \eqref{Zin} can be reduced to the ADHM-type equations
\begin{align}
\mu^\FC_{ab} :=[B_a,B_b]-\tfrac{1}{2}\,\varepsilon_{ab\bar c\bar d}\,\big[B_{\bar c}^\dagger,B_{\bar d}^\dagger\big]=0 \qquad \mbox{and} \qquad     \mu^\FR:=\sum_{a=1}^4\,\big[B_a,B_{\bar a}^\dagger\big]+I\,I^\dagger=\zeta\, \ident_{ k} \ .
    \label{ADHMeq}
\end{align}
{In Appendix~\ref{app:ADHMconstruction} we describe how to explicitly construct an $\sSU(4)$-instanton connection $A$ from the generalized ADHM data.}
For later considerations, it is convenient to encode the ADHM data in a framed representation of the four-loop quiver $\mathsf{L}_4$:
\begin{equation}\label{eq:ADHMquiverC4}
{\small
\begin{tikzcd}
	{\boxed{W}} && {\boxed{V}}\arrow["B_1"',out=150,in=120,loop,swap,]\arrow["B_2"',out=60,in=30,loop,swap,]\arrow["B_3"',out=330,in=300,loop,swap,]\arrow["B_4"',out=240,in=210,loop,swap,]
	\arrow["I", from=1-1, to=1-3] 
\end{tikzcd}  }\normalsize
\end{equation}

\subsubsection*{Physical Interpretation}

The ADHM equations \eqref{ADHMeq} describe supersymmetric bound states of $k$ coincident D$1$-branes along $\FR^{1,1}$ with $r$ D$9$-branes filling $\FR^{1,9}\simeq\FR^{1,1}\times\FC^4$ in the low energy limit  of type~IIB string theory (or the T-dual D$0$--D$8$  and D$(-1)$--D$7$ systems), with a large constant Kalb--Ramond field preserving $\frac18$ of the supercharges~\cite{Hiraoka:2002wm,Nekrasov:2016qym,m4,m4c,Fucito:2020bjd,Billo:2021xzh, Pomoni:2021hkn}. The massless spectrum of D9--D9 strings yields ten-dimensional $\mathcal{N}=1$ Yang--Mills theory with gauge group $\sU(r)$, and the noncommutative instanton equations \eqref{Zin} describe D$1$--D$9$-brane bound states from the perspective of the D$9$-branes. Our choice of vanishing first Chern form in Section \ref{sec:coho_gauge_theory} excludes D$7$--D$9$ bound states, while the restriction of vanishing second Chern form, i.e. $S_0=0$, also excludes D$5$--D$9$ bound states. 

From the point of view of the D$1$-branes, the D$9$-branes are heavy and the degrees of freedom supported on them are frozen to their vacuum expectation values. Then the Chan--Paton gauge symmetry becomes a global $\sU(r)_{\rm col}$ colour symmetry on the D$9$-branes.
The parameter $\zeta$ is determined by the constant background $B$-field along $\FC^4$ and it plays the role of the coupling of a Fayet--Iliopoulos term in the low energy $\CN=(0,2)$ field theory  with gauge group $\sU(k)$ on the D$1$-branes. The arrows $B_a$ of the ADHM quiver representation \eqref{eq:ADHMquiverC4} for $a\in\ulfour$ are the lowest components of  chiral superfields on $\FR^{1,1}$ that arise from quantizing the D$1$--D$1$ strings which generate four complex scalar fields in the adjoint representation of $\sU(k)$, while the arrow $I$ is the lowest component of a chiral superfield on $\FR^{1,1}$ associated to the massless spectrum of D$1$--D$9$ open strings in the Neveu--Schwarz sector which give rise to a complex scalar field transforming in  the bifundamental representation of~$\sU(k)\times\sU(r)_{\rm col}$. 

The ADHM equations \eqref{ADHMeq} are the equations for the classical Higgs branch of the D1-brane theory~\cite{Nekrasov:2016qym,Pomoni:2021hkn}. Since
    \begin{align}
        \Tr\big(\mu^\FR\big)=\Tr\big(I\,I^\dagger\big)=\zeta \, k \ ,
    \end{align}
and the term in the middle is non-negative, a solution to $\mu^\FR=\zeta\,\ident_k$ exists only if $\zeta\geq 0$. If $\zeta=0$ then $I=0$, which corresponds to the two-dimensional $\CN=(8,8)$ supersymmetric Yang--Mills theory with gauge group $\sU(k)$ on $\FR^{1,1}$ associated to the massless spectrum of D1--D1 strings alone. It can be obtained as the dimensional reduction of ten-dimensional $\CN=1$ supersymmetric Yang--Mills theory with gauge group $\sU(k)$.

Thus the interesting case for us is $\zeta>0$; in this case the D1--D9-brane system is tachyonic and unstable~\cite{Boundstates}, so it decays to a supersymmetric vacuum via tachyon condensation. The presence of the Fayet--Iliopoulos term with coupling $\zeta$ in string theory corresponds to a noncommutative deformation of the $\sU(r)$ gauge theory on the D9-branes~\cite{Nekrasov:1998ss,Seiberg:1999vs}, and 
this is why one has to work with noncommutative instantons. Thus the situation here is somewhat different from the four-dimensional case: for instantons on $\FR^4\simeq\FC^2$ a similar analysis, now done for a D1--D5-brane system, leads to a theory where $\zeta=0$ is a permitted and non-trivial choice, though $\zeta\neq0$ prevents the system from entering the Coulomb branch through the small instanton singularity~\cite{Nekrasov:1998ss}.

Finally, as discussed in Section~\ref{sec:instantoncounting}, the partition function of the topological field theory is defined by adding $r$ fundamental fermions whose functional integration brings down the Euler class of the matter bundle $\mathscr{V}_{r,k}\otimes\mbf r\longrightarrow\frM_{r,k}$. In string theory language this is conjecturally equivalent to adding $r$ anti-D9-branes to the interacting system of D-branes~\cite{m4,m4c}, with the masses $\vec m=(m_1,\dots,m_r)$ corresponding to the equivariant Chern roots of the Chan--Paton bundle on the $\overline{\rm D9}$-branes. Similarly to the D9--D9 strings, the degrees of freedom supported on the $\overline{\rm D9}$--$\overline{\rm D9}$ strings  are frozen to their vacuum expectation values, leaving a global $\sU(r)_{\rm fla}$ flavour symmetry on the $\overline{\rm D9}$-branes. Although the D9-branes and anti-D9-branes need not annihilate completely to a supersymmetric vacuum in the background of a $B$-field~\cite{Sen:1998sm}, any remaining open string ground states will be invisible to the D1-branes.

Altogether the spectrum of stable states (free from tachyonic modes) in the presence of $r$ anti-D9-branes involves reversed GSO projections and changes only the Fermi field content in the $\CN=(0,2)$ field theory on the D$1$-branes, adding to the ADHM data an antichiral spinor on $\FR^{1,1}$ in the bifundamental representation of $\sU(k)\times\sU(r)_{\rm fla}$ associated to the D1--$\overline{\rm D9}$-strings from the Ramond sector, whose lowest component we denote by $\bar I$. It does not modify the form of the ADHM equations \eqref{ADHMeq} nor the non-trivial values of the Fayet--Iliopoulos coupling $\zeta$.

\subsubsection*{Geometrical Interpretation}

The moduli space $\frM_{r,k}$ of $\sU(r)$ noncommutative $k$-instantons can be described as the {ADHM moduli space}, that is, the space of quintuples $(B_a,I)_{a\in\ulfour}$ satisfying the ADHM equations \eqref{ADHMeq} modulo the action  of  the unitary group $\sG^{\textrm{\tiny$V$}}=\sU(k)$ of the vector space $V$ given by
\begin{align}\label{eq:GVaction}
    g \cdot (B_{a}, I)_{a\in\ulfour}= (g\,B_{a}\,g^{-1}, g\,I)_{a\in\ulfour} \qquad \mbox{with} \quad g\in \sU(k) \ ,
\end{align}
which leaves invariant the ADHM equations \eqref{ADHMeq}. In this parametrization, the vector space $V$ descends to $\frM_{r,k}$ as the vector bundle $\CCV_{r,k}$ discussed in Section~\ref{sec:instantoncounting}.

Let us 
consider the complex moment map equations in \eqref{ADHMeq}, given by $\mu^\FC_{ab}=0$ for $1\leq a<b\leq 4$. As noted by~\cite{m4}, there is an identity
\begin{align}\label{eq:muabCholomorphic}
    \sum_{1\leq a< b\leq 4}\, \big\| \mu^\FC_{ab}\big\|_{\textrm{\tiny F}}^2&=\sum_{1\leq a< b\leq 4}\, \big\| [B_a,B_b] \big\|_{\textrm{\tiny F}}^2 \ ,
\end{align}
where $\|-\|_{\textrm{\tiny F}}$ is the Frobenius norm on $\sEnd_\FC(V)$. From \eqref{eq:muabCholomorphic} it follows that the equations $\mu^\FC_{ab}=0$ are equivalent to the equations $[B_a,B_b]=0$. Instead, the real moment map equation $\mu^\FR=\zeta\,\ident_k$ in \eqref{ADHMeq} for $\zeta>0$ is equivalent to a \emph{stability} condition:
There is no proper subspace $S\subset\mathbbm{C}^k$ such that $B_a(S)\subset S$ for $a\in\ulfour$ and ${\rm im}(I)\subset S$. The stability condition is equivalent to the condition that the operators $B_a$ for $a\in\ulfour$ acting on ${\rm im}(I)$ generate the vector space $V=\FC[B_1,B_2,B_3,B_4]\,I(W)$; it can be interpreted physically as the requirement that all D1-branes are bound to D9-branes.

Writing $\frB^{\rm st}$ for the space of stable quintuples, it follows that the instanton moduli space has an explicit holomorphic parametrization as the geometric invariant theory (GIT) quotient
\begin{align}\label{eq:Mkrholo}
    \mathfrak{M}_{r,k}=\big\{(B_a,I)_{a\in\ulfour}\in\frB^{\rm st} \ \big| \ [B_a,B_b]=0 \ , \ 1\leq a<b\leq 4\big\} \, \big/ \, \sGL(k,\mathbbm{C}) \ .
\end{align}
From this one easily sees that the real virtual dimension of $\frM_{r,k}$ is $(8\,k^2+2\,r\,k) - 6\,k^2-2\,k^2 = 2\,r\,k$. 

As a consequence, by results of~\cite{Henni:2017,sheaf} the instanton moduli space $\frM_{r,k}$ is isomorphic to the Quot scheme $  \Quot^k_{r}(\FC^4)$ of zero-dimensional quotients of the free sheaf $\CO_{\FC^4}^{\oplus r}$ with length $k$, 
\begin{align}\label{eq:MkrQuot}
\frM_{r,k} \, \simeq \,   \Quot^k_{r}(\FC^4) \ ,
\end{align}
which parametrizes framed torsion free sheaves $\CE$ on complex projective space $\PP^4$ of rank $r$ and $\ch_4(\CE)=k$. Any such sheaf sits in a short exact sequence
\begin{align}
0\longrightarrow \CE\longrightarrow \CO_{\PP^4}^{\oplus r}\longrightarrow \CS_Z \longrightarrow 0 \ ,
\end{align}
where $\CS_Z$ is a pure torsion sheaf of length $k$ supported on a zero-dimensional subscheme $Z\subset\PP^4$. 

\begin{remark}
In the rank one case $r=1$, it follows from \eqref{eq:MkrQuot} that the instanton moduli space is isomorphic to the Hilbert scheme of $k$ points on $\FC^4$,
\begin{align}
\frM_{1,k} \, \simeq \, \Hilb^k(\FC^4) \ ,
\end{align}
which is parametrized by ideals $J$ of codimension $k$ in the polynomial ring $\FC[z_1,z_2,z_3,z_4]$. The correspondence follows from defining the $k$-dimensional vector space $V=\FC[z_1,z_2,z_3,z_4]/J$ for an ideal $J\in\Hilb^k(\FC^4)$. We then take $B_a\in\sEnd_{\FC}(V)$ to be given as multiplication by $z_a$ mod~$J$ and $I\in\sHom_\FC(\FC,V)$ as $I(1) = 1$ mod~$J$.
\end{remark}

\subsection{Quiver Matrix Model}\label{BRST}

From the perspective of D$9$-branes, the noncommutative field theory of Section~\ref{sec:ADHMC4} computes the instanton partition function \eqref{eq:ZC4rqformal} for the cohomological field theory of Section~\ref{sec:coho_gauge_theory}~\cite{Szabo:2022zyn}. Here we shall study the theory from the perspective of the D$1$-branes. Following~\cite{coho,Quiver3d}, this uses the generalized ADHM parametrization of the instanton moduli space $\frM_{r,k}$ to construct a matrix model representation of the integral in~\eqref{eq:ZC4rqformal}, which is based on the fields of the quiver~\eqref{eq:ADHMquiverC4}. The construction of the cohomological matrix model proceeds analogously to the construction of the eight-dimensional cohomological gauge theory from~\cite{top,special}.

The symmetry group of the quiver matrix model is $\sU(k)\times\sU(r)_{\rm col}\times\sSU(4)$. The gauge and global framing symmetries act on the ADHM variables $(B_a,I)_{a\in\ulfour}$ as
\begin{align}
(g,h)\cdot (B_a,I)_{a\in\ulfour} = (g\,B_a\,g^{-1},g\,I\,h^{-1})_{a\in\ulfour} \qquad \mbox{for} \quad (g,h)\in\sU(k)\times\sU(r)_{\rm col} \ .
\end{align}
We work equivariantly with respect to the (complex) maximal torus $\sT_{\vec\epsilon}$ of $\sSU(4)$, which  acts on the ADHM data as
 \begin{align}\label{eq:torusquiver}
(t_a)_{a\in\ulfour}\cdot (B_a,I)_{a\in\ulfour}=(t_a^{-1}\,B_a,I)_{a\in\ulfour} \qquad \mbox{for} \quad (t_a)_{a\in\ulfour}=\big(\e^{\,\ii\,\epsilon_a}\big)_{a\in\ulfour}\in \sT_{\vec\epsilon} \ .
\end{align}
The BRST multiplets are $(B_a,\psi_a)_{a\in\ulfour}$ and $(I,\varrho)$ with the BRST transformations
\begin{align}
\begin{split}
    \mathcal{Q}B_a=\psi_a\qquad , & \qquad \mathcal{Q}\psi_a=[\phi,B_a]-\epsilon_a\,B_a \ , \\[4pt]
    \mathcal{Q}I=\varrho\qquad , & \qquad  \mathcal{Q}\varrho=\phi \, I-I\,\mbf{a} \ ,
\end{split}
\end{align}
for $a\in\ulfour$,
where $\phi$ is the generator of $\sU(k)$ gauge transformations and $\mbf{a}={\rm diag}(a_1,\dots, a_r)$ is a background field which parametrizes an element of the (complex) Cartan subalgebra of $\sU(r)_{\rm col}$.

The matrix fields $(B_{a},I)_{a\in\ulfour}$ are required to satisfy the seven constraints \eqref{ADHMeq}. Note that the complex moment map equations $\mu^\FC_{ab} = 0$,  for $a,b\in\ulfour$ with $a<b$, only contribute three independent constraints, which we choose to be 
\begin{align}\label{eq:EabC4prime}
\mu^\FC_{\alpha\beta}=0 \qquad \mbox{with} \quad (\alpha,\beta) \ \in \  \ulthree^{\perp}:=\big\{(1,2)\,,\,(1,3)\,,\,(2,3)\big\} \ . 
\end{align}
To implement the equations \eqref{eq:EabC4prime} and $\mu^\FR=\zeta\,\ident_k$ in \eqref{ADHMeq}, we add the Fermi multiplets $(\vec\chi,\vec H)$, where \smash{$\vec\chi=(\chi^\FC_{\alpha\beta},\chi^\FR)_{(\alpha,\beta)\in\ulthree^\perp}$} are the antighost fields in $\sEnd_\FC(V)$ and \smash{$\vec H=(H_{\alpha\beta}^\FC,H^\FR)_{(\alpha,\beta)\in\ulthree^\perp}$} are the auxiliary fields. Their BRST transformations are
\begin{align}
\begin{split}
\mathcal{Q}\chi^\FC_{\alpha\beta}=H^\FC_{\alpha\beta} \qquad , & \qquad \mathcal{Q}H^\FC_{\alpha\beta}=[\phi,\chi^\FC_{\alpha\beta}]-\epsilon_{\alpha\beta}\,\chi^\FC_{\alpha\beta} \ , \\[4pt]
\mathcal{Q}\chi^\FR=H^\FR \qquad , & \qquad \mathcal{Q}H^\FR=[\phi,\chi^\FR] \ ,
      \end{split}
\end{align}
for $(\alpha,\beta)\in\ulthree^\perp$, where generally we abbreviate $\epsilon_{ab\cdots}:= \epsilon_a+\epsilon_b+\cdots$ throughout. Finally, we add the scalar gauge multiplet $(\phi,\bar{\phi},\eta)$ to close the algebra of BRST transformations as
\begin{align}
     \mathcal{Q}\phi=0 \ , \quad \mathcal{Q}\bar{\phi}=\eta \qquad \mbox{and} \qquad \mathcal{Q}\eta=[\bar{\phi},\phi] \ .
\end{align}

We also add matter fields to the theory, associated to the matter bundle $\mathscr{V}_{r,k}\otimes\mbf r$, which transform under a $\sU(r)_{\rm fla}$ flavour symmetry. Corresponding to $r$ anti-D$9$-branes, in the cohomological matrix model we include the Fermi multiplet $(\bar{I},\bar{\varrho})$ of fields in $\sHom_\FC(\mbf r,V)$ with the BRST transformations
 \begin{align}
     \mathcal{Q}\bar{I}=\bar{\varrho}\qquad \mbox{and} \qquad \mathcal{Q}\bar{\varrho}=\phi\,\bar{I}-\bar{I}\,\mbf m \ ,
 \end{align}
where the masses $\mbf{m}={\rm diag}(m_1,\dots, m_r)$ parametrize an element of the (complex) Cartan subalgebra of $\sU(r)_{\rm fla}$.

The action functional which corresponds to this system of fields and equations is given by
\begin{align}\label{eq:action}
\begin{split}
    S&=\mathcal{Q} \, \Tr\Big( \sum_{(\alpha,\beta)\in\ulthree^\perp} \, {\chi}_{\alpha\beta}^{\FC\,\dagger}\,\big(g_\FC\,{H}^\FC_{\alpha\beta}-\mu^\FC_{\alpha\beta}\big)+\chi^\FR\,\big(g_\FR\,H^\FR -\mu^\FR-\zeta\,\ident_k\big) \\
    &\hspace{2cm} + \sum_{a\in\ulfour} \, \psi_a\,\big[\bar{\phi},B^\dagger_a\big]+I^\dagger\,\bar{\phi}\,\varrho+\bar{I}^\dagger\,\bar{\phi}\,\bar{\varrho}+\eta\,\big[\phi,\bar{\phi}\,\big]+{\rm c.c.}\Big) \\
    & \quad \, + g_1\,\mathcal{Q}\,\Tr \big(\chi^\FR\,\bar{\phi}\,\big) + g_2\,\mathcal{Q}\,\Tr\Big(\sum_{a\in\ulfour} \, \big(B_a\,\psi_a^\dagger-\psi_a\, B_a^\dagger\big)+I\,\varrho^\dagger-\varrho\, I^\dagger + \bar I\,\bar\varrho^\dagger-\bar\varrho\, \bar I^\dagger\Big) \ ,
    \end{split}
\end{align}
where ${\rm c.c.}$ means complex conjugate and we introduced appropriate coupling constants for convenience. The last two terms have been added by hand to respectively give a nondegenerate mass matrix for $\chi^\FR$ as well as kinetic terms for $\psi_a$, $\varrho$ and $\bar\varrho$.
After dividing out by the volume of the gauge group, the path integral form of the equivariant integral over the ADHM variables can then be written symbolically as
\begin{align}
    \int \, \frac{\mathscr{D}[{\rm fields}]}{{\tiny\tt vol}\big(\sU(k)\big)} \ \e^{-S} \ . \label{pf}
\end{align}

To evaluate the integral~\eqref{pf}, we follow the technique employed in~\cite{coho}.  
The first step is to use the $\sU(k)$ gauge invariance to diagonalize the gauge generator $\phi$. This produces a measure on $\FR^k$ with Vandermonde determinant
\begin{align}
    \frac{1}{k!}\, \prod_{i=1}^k \, \frac{\dd\phi_i}{2\pi \,\ii} \ \prod^k_{\stackrel{\scriptstyle i,j=1}{\scriptstyle i\neq j}} \, (\phi_i-\phi_j) \ ,
\end{align}
where $k!$ is the order of the Weyl group of $\sU(k)$. 

By construction of the cohomological theory, the path integral \eqref{pf} is independent of the coupling constants $g_\FC$, $g_\FR$, $g_1$ and $g_2$. Taking the limit $g_1\rightarrow\infty$, the relevant part of the action functional is $\Tr(H^\FR\,\bar\phi+\chi^\FR\,\eta)$, which shows that the fields $(H^\FR,\bar{\phi},\chi^\FR,\eta)$ can simply be integrated out and do not contribute non-trivially to the path integral. 

Sending $g_\FC\rightarrow\infty$, the relevant part of the action functional is
\begin{align}
\sum_{(\alpha,\beta)\in\ulthree^\perp} \, \Tr\Big(H_{\alpha\beta}^{\FC\,\dag} \, H^\FC_{\alpha\beta} + \chi_{\alpha\beta}^{\FC\,\dag}\,\big([\phi,\chi^\FC_{\alpha\beta}]-\epsilon_{\alpha\beta}\,\chi^\FC_{\alpha\beta}\big)\Big) \ .
\end{align}
This shows that the auxiliary fields $H^\FC_{\alpha\beta}$ do not contribute non-trivially.  On the other hand, integrating the antighost fields \smash{$\chi^\FC_{\alpha\beta}$} produces
\begin{align}\label{eq:antighostsign}
\prod_{i, j=1}^k \ \prod_{(\alpha,\beta)\in\ulthree^\perp} \, (\phi_i-\phi_j-\epsilon_{\alpha\beta}) \ ,
 \end{align}
which contributes to the equivariant Euler class of the self-dual obstruction bundle ${\sf Ob}_{r,k}^-$, with global orientation corresponding to the choice \eqref{eq:EabC4prime} and our choice of ordering of antighosts in the Grassmann integration measure. 

Finally, for $g_2\rightarrow\infty$ the contribution of the ADHM fields $(B_a,I)_{a\in\ulfour}$ to the path integral \eqref{pf} is
 \begin{align}
\prod_{i,j=1}^k \ \prod_{a\in\ulfour} \, \frac{1}{\phi_i-\phi_j-\epsilon_a} \ \prod_{i=1}^{k } \ \prod_{l=1}^{r} \, \frac{1}{\phi_i-a_l} \ ,
 \end{align}
while the contribution of the matter fields $\bar I$ is
 \begin{align}\label{eq:matterintegral}
     \prod_{i=1}^{k} \ \prod_{l=1}^{r} \, (\phi_i-m_l)\ .
 \end{align}

Putting everything together we arrive at the matrix model representation of the equivariant integral in \eqref{eq:ZC4matter} given by
 \begin{equation}
     \begin{split}
Z_{\FC^4}^{r,k}(\vec a,\vec\epsilon,\vec m) := \int^\sT_{[\frM_{r,k}]_\tto^{\rm vir}} \, e_\sT(\CCV_{r,k}\otimes\mbf r) &= \frac{(-1)^k}{k!} \, \Big(\frac{\epsilon_{12} \, \epsilon_{13} \, \epsilon_{23}}{\epsilon_1\,\epsilon_2\,\epsilon_3\,\epsilon_4}\Big)^k \\ & \quad \, \times \ \oint_{\varGamma_{r,k}} \ \prod_{i=1}^k \, \frac{\dd\phi_i}{2\pi\,\ii} \, \frac{\CP_r(\phi_i|\vec m)}{\CP_r(\phi_i|\vec a)} \ \prod_{\stackrel{\scriptstyle i,j=1}{\scriptstyle i\neq j}}^k \, \CR(\phi_i-\phi_j|\vec\epsilon\,) 
   \end{split}\label{partition}
   \end{equation}
at a generic point $\vec\epsilon$ of the $\varOmega$-deformation, where
\begin{align}\label{eq:CPrCR}
\CP_r(x|\vec w) = \prod_{l=1}^r \, (x-w_l) \qquad \mbox{and} \qquad \CR(x|\vec\epsilon\,) = \frac{x\,(x-\epsilon_{12})\,(x-\epsilon_{13})\,(x-\epsilon_{23})}{(x-\epsilon_1)\,(x-\epsilon_2)\,(x-\epsilon_3)\,(x-\epsilon_4)} \ .
\end{align}
The instanton partition function \eqref{eq:ZC4matter} is then the generating function for the integrals \eqref{partition}, obtained from a weighted sum over all the instanton sectors:
  \begin{align}\label{partitionfull}
      Z_{\FC^4}^r(\qu;\vec a,\vec\epsilon,\vec m) = 1+\sum_{k=1}^\infty \, \qu^k \ Z_{\FC^4}^{r,k}(\vec a,\vec\epsilon,\vec m) \ .
  \end{align}
  
The polynomial $\CP_r(x|\vec w)$ encodes the $S_r$-invariants for the action of the Weyl group $S_r$ of $\sU(r)$ by permutations of the entries of $\vec w=(w_1,\dots,w_r)$. There is also a symmetry under the group $\RZ_2$ which acts on the $\varOmega$-deformation parameters $\vec\epsilon$ by permuting $\epsilon_1\leftrightarrow\epsilon_4=-\epsilon_{123}$. 
Note that the action of the center $\sU(1)\subset \sU(r)_{\rm col}$ on the partition function \eqref{partition} can be absorbed into an action of the center $\sU(1)\subset \sU(r)_{\rm fla}$: an overall shift $a_l\to a_l+a$ of the Coulomb parameters can be absorbed by a simultaneous overall shift $\phi_i\to\phi_i+a$ of the integration variables and redefinition $m_l\to m_l+a$ of the masses. Hence there is a natural action of ${\sf P}\big(\sU(r)_{\rm col}\times \sU(r)_{\rm fla}\big)$ on $Z_{\FC^4}^{r,k}(\vec a,\vec\epsilon,\vec m)$.

The integral \eqref{partition} also appears in~\cite{m4c,Fucito:2020bjd,Billo:2021xzh,Kimura:2022zsm} from different perspectives. It is ill-defined as a Lebesgue integral over $\FR^k$ because its integrand does not converge to zero at infinity. Instead, as suggested by the notation, it should be interpreted as a contour integral where the sum of residues reproduces the sum over $\sT$-fixed points in \eqref{eq:ZC4matter}. However, assigning a choice of contour $\varGamma_{r,k}\subset\FC^k$ to the integral \eqref{partition} directly is not straightforward. We shall discuss this point further below.

\subsubsection*{Dimensional Reduction}

From \eqref{partition} we may immediately deduce

\begin{proposition}\label{prop:ZDTgeb}
The equivariant instanton partition function of the cohomological gauge theory with a massive fundamental hypermultiplet on $\FC^4$ is related to the Coulomb branch partition function \smash{$Z_{\FC^3}^{r}(\qu;\vec a,\epsilon_1,\epsilon_2,\epsilon_3)$} for Donaldson--Thomas theory on the toric K\"ahler threefold $\FC^3$ through the mass specialisations
\begin{align}\label{eq:ZDT}
Z_{\FC^4}^{r}(\qu;\vec a,\vec\epsilon,m_l=a_l+\epsilon_4) =  Z_{\FC^3}^{r}\big((-1)^{r+1}\,\qu;\epsilon_1,\epsilon_2,\epsilon_3\big)={M}(-\qu)^{-\frac{r\,\epsilon_{12}\,\epsilon_{13}\,\epsilon_{23}}{\epsilon_1\,\epsilon_2\,\epsilon_3}} \ ,
\end{align}
where
\begin{align}\label{eq:MacMahon}
{M}(q)=\prod_{n=1}^\infty\,\frac{1}{(1-q^n)^n}
\end{align}
is the MacMahon function.
\end{proposition}

\begin{proof}
Using the Calabi--Yau condition $\epsilon_4=-\epsilon_{123}$ on $\FC^4$ one finds that the matrix integral \eqref{partition} gives
\begin{align}
\begin{split}
& Z_{\FC^4}^{r,k}(\vec a,\vec\epsilon, m_l=a_l-\epsilon_{123}) \\[4pt] & \qquad =
\frac{1}{k!} \, \Big(\frac{\epsilon_{12} \, \epsilon_{13} \, \epsilon_{23}}{\epsilon_1\,\epsilon_2\,\epsilon_3\,\epsilon_{123}}\Big)^k \, \oint_{\varGamma_{r,k}} \ \prod_{i=1}^k \, \frac{\dd\phi_i}{2\pi\,\ii} \, \frac{\CP_r(\phi_i+\epsilon_{123}|\vec a)}{\CP_r(\phi_i|\vec a)}  \ \prod_{\stackrel{\scriptstyle i,j=1}{\scriptstyle i\neq j}}^k \, \CR(\phi_i-\phi_j|\epsilon_1,\epsilon_2,\epsilon_3,-\epsilon_{123}) \ .
\end{split}
\end{align}
Up to an overall sign $(-1)^{(r+1)\,k}$, this is exactly the charge $k$ integral contribution \smash{$Z_{\FC^3}^{r,k}(\vec a,\epsilon_1,\epsilon_2,\epsilon_3)$} from~\cite[eq.~(5.23)]{coho} to the rank $r$ instanton partition function on $\FC^3$ at a generic point $(\epsilon_1,\epsilon_2,\epsilon_3)$ of the $\varOmega$-deformation. This proves the first equality in \eqref{eq:ZDT}. 

In the rank one case $r=1$, the second equality in \eqref{eq:ZDT} was originally proven in~\cite[Theorem~1]{Maulik:2004txy}, where \smash{$Z_{\FC^3}^{1,k}(\epsilon_1,\epsilon_2,\epsilon_3)$} computes the equivariant volume of the Hilbert scheme of $k$ points on $\FC^3$:
\begin{align}\label{eq:ZC4C3k1}
Z_{\FC^4}^{1,k}(\vec\epsilon,m=a-\epsilon_{123}) = Z_{\FC^3}^{1,k}(\epsilon_1,\epsilon_2,\epsilon_3) = \int^{\sT_{\vec\epsilon}}_{[\Hilb^k(\FC^3)]^{\rm vir}} \, 1 \ .
\end{align}
The general higher rank result for $r>1$ was conjectured in {\cite{Awata:2009dd,Szaboconj}} and proven in \cite{proofconj}, where \smash{$Z_{\FC^3}^{r,k}(\vec a,\epsilon_1,\epsilon_2,\epsilon_3)$} computes the equivariant volume of the Quot scheme \smash{$\Quot^k_{r}(\FC^3)$} of zero-dimensional quotients of $\CO_{\FC^3}^{\oplus r}$ of length $k$. 
\end{proof}

\begin{remark}\label{rem:8d6d2}
From the string theory perspective, Proposition~\ref{prop:ZDTgeb} superficially supports Sen's conjecture~\cite{Sen:1998sm}: it implies that the specification $m_l=a_l+\epsilon_4$, which corresponds to the diagonal coordinates for the maximal torus $\sU(1)^r_{\rm fla}\times\sU(1)^r_{\rm col}$ of the global symmetry group $\sU(r)_{\rm fla}\times\sU(r)_{\rm col}$, can be interpreted as a particular  configuration of D9-branes and anti-D9-branes that decay into D7-branes, whose bound states with D1-branes correspond to instantons on $\FC^3$. This was noted in the K-theory version of the theory for $r=1$ by Nekrasov~\cite{m4}. The rank one version of the dimensional reduction of Proposition~\ref{prop:ZDTgeb} was also studied rigorously by~\cite{Cao:2017swr} and its K-theory counterpart by~\cite{Cao:2019tvv}.
\end{remark}

\subsubsection*{Fixed Points and Solid Partitions}

Starting from the integral \eqref{partition} directly, we can identify the combinatorial enumeration problem computed by the sum over residues. The integrand has poles along the hyperplanes
  \begin{align}
      \phi_i-\phi_j-\epsilon_a=0 \qquad \mbox{and} \qquad \phi_i-a_l=0 \ .\label{fixpoint}
  \end{align}
These are exactly the fixed points of the equivariant action of $\mathsf{T}_{\vec\epsilon}$ on the ADHM variables, which are given by
  \begin{align}\label{eq:fixedpointsC4}
      t_a\,B_a=g\,B_a\,g^{-1} \qquad \mbox{and} \qquad
      I=g\,I\,h^{-1} \ ,
  \end{align}
for $t_a=\e^{\,\ii\,\epsilon_a}\in\sT_{\vec\epsilon}$, $g=\exp(\ii\,\phi)\in\sU(k)$ and $h=\exp(\ii\,\mbf a)\in\sU(r)_{\rm col}$. These equations define a group representation $\sT_{\vec a}\times\sT_{\vec\epsilon}\longrightarrow\sU(k)$, and they are equivalent to
  \begin{align}
      (B_a)_{ij}\,(\phi_i-\phi_j-\epsilon_a)=0\qquad \mbox{and} \qquad I_{il}\,(\phi_i-a_l)=0 \ .
  \end{align}
Together with $\bar\varrho_{il}=0$, these are the equations for the fixed points of the BRST charge $\CQ$, onto which the path integral of the topological field theory localizes.

The fixed points \eqref{fixpoint} have a standard combinatorial description. A subtlety here is that the fixed point equations \eqref{eq:fixedpointsC4} are given for the equivariant action of the (complex) maximal torus $\sT_{\vec\epsilon}$ of $\sSU(4)$, rather than the maximal torus $\sT_{\vec\epsilon\,'}$ of $\sU(4)$. Following Nekrasov's approach from~\cite{m4}, one can argue that the two sets of fixed point equations are equivalent: The gauge transformations generated by $\phi$ can be replaced with  $\sU(k)$ gauge transformations generated by $\phi'$ such that
\begin{align}\label{eq:fixedpointprime}
[B_a,\phi']=\epsilon'_a\,B_a \ ,
\end{align}
for $a\in\ulfour$ and \emph{generic} $\vec\epsilon\,'=(\epsilon_1',\epsilon_2',\epsilon_3',\epsilon'_4)$. Clearly if \eqref{eq:fixedpointprime} holds generically, then it holds in particular at the Calabi--Yau specialisation with $\epsilon_4=-\epsilon_{123}$. 

Conversely, if $[B_a,\phi]=\epsilon_a\,B_a$ with $\epsilon_{1234}=0$, then $X:=B_1\,B_2\,B_3\,B_4$ is nilpotent, that is, $[X,\phi]=0$. By the Jacobson--Morozov theorem, it follows that $X$ can be extended to an $\mathfrak{sl}(2,\FC)$ triple $(X,Y,H)$ in $\sEnd_\FC(V)$:
\begin{align}
    [X,H]=2\,X \ , \quad [H,Y]=-2\,Y \qquad \mbox{and} \qquad [X,Y]=H \ ,
\end{align}
where $Y$ is nilpotent and $H$ is the generator of the associated Cartan subalgebra. Given arbitrary $\vec\epsilon\,'$, we define $\phi'=\frac12\,\epsilon'_{1234}\, H$ with $\epsilon'_{1234}\neq0$. Then $[X,\phi']=\epsilon'_{1234}\, X$, and in particular
\begin{align}
\begin{split}
& [B_1,\phi']\,B_2\,B_3\,B_4+B_1\,[B_2,\phi']\,B_3\,B_4+B_1\,B_2\,[B_3,\phi']\,B_4+B_1\,B_2\,B_3\,[B_4,\phi']  =\epsilon'_{1234}\, X \ . 
\end{split}
\end{align}
It follows that $[B_a,\phi']=\epsilon'_a \,B_a$ (up to relabelling of the entries of $\vec\epsilon\,'$), as required. A geometric proof of this fact (at the level of fixed points of Hilbert schemes of points on $\FC^4$) is found in~\cite{Cao:2017swr}.

With this equivalence in mind, we write a generic fixed point \eqref{fixpoint} by setting the eigenvalues of $\phi$ to
\begin{align}\label{eq:phifixedpoints}
    \phi_{(a_l;\vec p\,)}=a_l+\vec p\cdot \vec\epsilon \ ,
\end{align}
where $l\in\{1,\dots,r\}$ and $\vec p=(p_1,p_2,p_3,p_4)\in\RZ^4_{>0}$ is a quadruple of positive integers; here and in the following we denote $\vec p\cdot\vec\epsilon:=\sum_{a\in\ulfour}\,p_a\,\epsilon_a$ and we used the Calabi--Yau condition \eqref{eq:CY4constraint}. In other words, a fixed point is parametrized by a set of tuples $(a_l;\vec p\,)$. There are $k$ different tuples, because there are $k$ eigenvalues. Given an instanton number $k$, we can partition it into integers
$k=k_1+\cdots + k_r$ such that for fixed $l$, and hence fixed $a_l$, $k_l\geq0$ is the total number of quadruples $\vec p$ of integers. 

Since the quadruple $(1,1,1,1)$ is always present, the collection of tuples $(a_l;\vec p\,)$ for fixed $l$ may be represented by a \emph{solid partition} \smash{$\sigma_l = (\sigma_{i,j,n})_{i,j,n\geq1}$}, that is, a sequence of non-negative integers $\sigma_{i,j,n}\in\RZ_{\geq0}$ satisfying
\begin{align}
 \sigma_{i,j,n}\geq  \sigma_{i+1,j,n} \ , \quad  \sigma_{i,j,n}\geq\sigma_{i,j+1,n} \qquad \mbox{and} \qquad \sigma_{i,j,n}\geq  \sigma_{i,j,n+1} \ ,
\end{align}
for all $i,j,n\geq1$. The \emph{size} of the solid partition $\sigma_l$ is
\begin{align}
|\sigma_l| := \sum_{i,j,n\geq1} \, \sigma_{i,j,n} = k_l \ .
\end{align}

It follows that there is a one-to-one correspondence between fixed points $(a_l;\vec p\,)$ and solid partitions $\sigma_l=\{\vec p\in\RZ^4_{>0}~|~1\leq p_4\leq\sigma_{p_1,p_2,p_3}\}$, and therefore one may associate to any fixed point an array of solid partitions 
$
\vec{\sigma}=(\sigma_1,\dots, \sigma_r)
$
whose total size is the instanton number:
\begin{align}
|\vec{\sigma}|=\sum_{l=1}^r \, |\sigma_l|=\sum_{l=1}^r \, k_l=k \ .
\end{align}
Note that here the only role of the ADHM field $I$, which arises from D1--D9 strings, is to assign a `colouring' of the solid partitions into the $r$-tuple~$\vec\sigma$.

\begin{remark}\label{not:solidpartitions}
Recall that a {plane partition} is a collection of triples of positive integers $(p_1,p_2,p_3)$ where the pairs $(p_1,p_2)$ define a Young diagram $\lambda$ and $p_3\leq\pi_{p_1,p_2}$, with $\pi$ a $\RZ_{>0}$-valued  function of $(i,j)\in\lambda$  such that $\pi_{i,j}\geq\pi_{i+1,j}$ and $\pi_{i,j}\geq\pi_{i,j+1}$. The rank one instanton partition function for the $\CN_{\textrm{\tiny T}}=2$ cohomological gauge theory on $\FC^3$ at the Calabi--Yau specialisation of the $\varOmega$-deformation, or equivalently the MacMahon function \eqref{eq:MacMahon}, is the generating function for the number of plane partitions $\pi$ of fixed size $|\pi|=\sum_{i,j\geq1}\,\pi_{i,j}$:
\begin{align}
Z_{\FC^3}^{r=1}\big(\qu;\epsilon_1,\epsilon_2,\epsilon_3\big)\big|_{\epsilon_{123}=0} = M(-\qu) = \sum_\pi \, (-\qu)^{|\pi|} \ .
\end{align}
More generally, $Z_{\FC^3}^{r}\big(\qu;\epsilon_1,\epsilon_2,\epsilon_3\big)\big|_{\epsilon_{123}=0} = M\big((-1)^r\,\qu\big)$ is the generating function for $r$-tuples $\vec\pi=(\pi_1,\dots,\pi_r)$ of plane partitions.

Given a \emph{solid} partition $\sigma=(\sigma_{i,j,n})_{i,j,n\geq 1}$, we can view it as a stack of \emph{plane} partitions, with cubes piled in the positive $16$-multiant $(x_1,x_2,x_3,x_4)\in\RZ_{\geq0}^4$, where $\sigma_{i,j,n}$ is the height function of the stack of hypercubes defined on the $(x_1,x_2,x_3)$ hyperplane. 
\end{remark}

\subsection{Instanton Partition Function}\label{sec:Coho_Part_Func}

We shall now discuss the explicit evaluation of the equivariant partition function \eqref{eq:ZC4matter} as a combinatorial expansion based on the generalized ADHM parametrization of the instanton moduli spaces~$\frM_{r,k}$.
As discussed in Section~\ref{BRST}, the partition function \eqref{partition} receives contributions from the fixed points $[(B_a,I)_{a\in\ulfour}]\in\frM^\sT_{r,k}$ of the ADHM moduli space. Thus we need to describe the local geometry of the instanton moduli space $\frM_{r,k}$ around these fixed points, and suitably incorporate the moduli of the matter fields. 

Let $\vec{\sigma}$ be a fixed point, corresponding to an $r$-tuple of solid partitions, and consider the ADHM deformation complex
\begin{align}
    0\longrightarrow \sEnd_\FC(V_{\vec{\sigma}})\xrightarrow{ \ \dd_1 \ }\begin{matrix}
        \sEnd_\FC(V_{\vec{\sigma}}) \otimes  Q \\[1ex] \oplus\\[1ex] \sHom_\FC(W_{\vec{\sigma}},V_{\vec{\sigma}})
    \end{matrix}\xrightarrow{ \ \dd_2 \ } \sEnd_\FC(V_{\vec{\sigma}})\otimes\midwedge_-^{0,2}\,  Q  \longrightarrow 0 \ , \label{complex}
\end{align}
where $ Q \simeq\FC^4$ is the four-dimensional fundamental $\sT_{\vec\epsilon}\,$-module with weight decomposition 
\begin{align}
 Q  = t_1^{-1} + t_2^{-1} + t_3^{-1} + t_4^{-1} \ ,
\end{align}
and the weights $t_a=\e^{\,\ii\,\epsilon_a}$ satisfy the Calabi--Yau condition 
\begin{align}
t_1\,t_2\,t_3\,t_4=1 \ . 
\end{align}
The map $\dd_1$ is an infinitesimal $\sGL(k,\FC)$ gauge transformation, while $\dd_2$ is the linearization of the holomorphic ADHM equations 
$[B_\alpha,B_\beta]=0$; explicitly
\begin{align}
    \dd_1\phi=\bigg(\begin{matrix}
    \big([\phi ,B_a]\big)_{a\in\ulfour} \\
    \phi \, I
    \end{matrix}\bigg) \qquad \mbox{and} \qquad
   \dd_2\bigg(\begin{matrix}
    (b_a)_{a\in\ulfour}\\i 
    \end{matrix}\bigg)=\big(
    [b_\alpha,B_\beta]+[B_\alpha,b_\beta]
    \big)_{(\alpha,\beta)\in\ulthree^\perp} \ .
\end{align}

By construction, the first cohomology of the cochain complex \eqref{complex} is a local model for the tangent space $T_{\vec{\sigma}}\mathfrak{M}_{r,k}$ at the fixed point $\vec\sigma$, while the second cohomology parametrizes the local obstruction space \smash{$(\Ob_{r,k}^-)_{\vec\sigma}$}. We assume that $\ker(\dd_1)=0$. Then the equivariant index of this complex computes the virtual sum 
\begin{align}
\sqrt{\ch_{\sT}}\big(T_{[\CE]}^{\rm vir}\,\frM_{r,k}\big) = {\sf Ext}_{\CO_{\PP^4}}^1(\CE,\CE) \, \ominus \, {\sf Ext}_{\CO_{\PP^4}}^{2\,-}(\CE,\CE)
\end{align} 
of cohomology groups, where $[\CE]$ is the isomorphism class of a framed torsion free sheaf on $\PP^4$ corresponding to the fixed point $\vec\sigma$. This gives 
\begin{align}\label{eq:chTvir}
\begin{split}
\sqrt{\ch_{\sT}}\big(T_{\vec\sigma}^{\rm vir}\,\frM_{r,k}\big)&=V_{\vec{\sigma}}^*\otimes V_{\vec{\sigma}} \, \big(t^{-1}_1+t^{-1}_2+t^{-1}_3+t^{-1}_4\big)+W_{\vec{\sigma}}^*\otimes V_{\vec{\sigma}}\\
    & \quad \, -V_{\vec{\sigma}}^*\otimes V_{\vec{\sigma}}\,\big(1+t^{-1}_1\,t^{-1}_2+t^{-1}_1\,t^{-1}_3+t^{-1}_2\,t^{-1}_3\big) \ . \end{split}
\end{align}

Recalling that the fibre of the matter bundle $(\CCV_{r,k})_{\vec\sigma}\otimes\,\mbf r$ at the fixed point $\vec\sigma$ is $\sHom_\FC(\mbf r,V_{\vec\sigma})$, the total index we wish to compute is
\begin{align}\begin{split}
    \chi_{\vec{\sigma}}:\!&= \sqrt{\ch_\sT}\big(T^{\rm vir}_{\vec{\sigma}}\,\mathfrak{M}_{r,k}\big)-\ch_\sT\big((\mathscr{V}_{r,k})_{\vec\sigma}\otimes\mbf r\big) \\[4pt]
    &= -V_{\vec{\sigma}}^*\otimes V_{\vec{\sigma}} \, \big(1-t^{-1}_1-t^{-1}_2-t^{-1}_3-t^{-1}_4+t^{-1}_1\,t^{-1}_2+t^{-1}_1\,t^{-1}_3+t^{-1}_2\,t^{-1}_3 \big)\\
    & \quad \,  +\big(W_{\vec{\sigma}}^*-\mbf r^*\big) \otimes V_{\vec{\sigma}} \ . \end{split} \label{chi}
\end{align}
In these expressions we regard the various vector spaces in \eqref{eq:chTvir} and \eqref{chi} as elements of the representation ring of the torus \smash{$\sT=\sT_{\vec a}\times\sT_{\vec\epsilon}\times\sT_{\vec m}$}, i.e. as polynomials in $t_a$ for $a\in\ulfour$ as well as in $e_l:=\e^{\,\ii\,a_l}$ and $f_l:=\e^{\,\ii\,m_l}$ for $l=1,\dots,r$. The dual involution acts on the weights as $t_a^*=t_a^{-1}$, $e_l^*=e_l^{-1}$ and $f_l^*=f_l^{-1}$.

As the notation suggests, the character \eqref{eq:chTvir} can be identified as a square root of the character of the virtual tangent bundle at the fixed point $\vec\sigma$:
\begin{align}\begin{split} \label{eq:chTvirfull}
   \ch_{\sT}\big(T_{\vec{\sigma}}^{\rm vir}\frM_{r,k}\big)&= -V_{\vec{\sigma}}^*\otimes V_{\vec{\sigma}} \ \prod_{a\in\ulfour }\,\big(1-t^{-1}_a\big)+W_{\vec{\sigma}}^*\otimes V_{\vec{\sigma}}+V_{\vec{\sigma}}^*\otimes W_{\vec{\sigma}} \\[4pt]
   &=\sqrt{\ch_{\sT}}\big(T_{\vec{\sigma}}^{\rm vir}\frM_{r,k}\big)+\sqrt{\ch_{\sT}}\big(T_{\vec{\sigma}}^{\rm vir}\frM_{r,k}\big)^* \ . \end{split}
\end{align}
The choice of square root is not unique and different choices will produce contributions to the $k$-instanton partition functions $Z_{\FC^4}^{r,k}(\vec a,\vec\epsilon,\vec m)$ below which coincide up to a sign $\pm1$ which we fix by hand, see e.g.~\cite{Cao:2017swr,m4c,Cao:2019tvv,Monavari:2022rtf,CKMpreprint}. Every sign choice is equivalent to a choice of local orientation at each $\sT_{\vec a,\vec\epsilon}\,$-fixed point $\vec\sigma$ of the instanton moduli space~\cite{Oh:2020rnj}, and it produces a sign factor $(-1)^{\ttO_{\vec\sigma}}$.

With suitable choices of bases we can decompose the $\sU(r)_{\rm col}$-module $W$ at a fixed point $\vec\sigma$ and the $\sU(r)_{\rm fla}$-module $\mbf r$ into one-dimensional vector spaces for the $\sT_{\vec a}\,$-action and the $\sT_{\vec m}$-action, respectively, with the weight decompositions
\begin{align}
W_{\vec{\sigma}}=\sum_{l=1}^r \, e_l \qquad \mbox{and} \qquad \mbf r=\sum_{l=1}^r \, f_l \ .
\label{decom_WM1}
\end{align}
On the other hand,  from the fixed point equations \eqref{eq:fixedpointsC4} and the stability condition (cf. Section~\ref{sec:ADHMC4}) it follows that the $\sT_{\vec a,\vec\epsilon}\,$-module decomposition of the vector space $V$ at the fixed point $\vec\sigma$ is given by
\begin{align}\label{decom_V}
V_{\vec{\sigma}}=\sum_{l=1}^r \, e_l \ \sum_{\vec p\,\in\sigma_l} \, t_1^{p_1}\,t_2^{p_2}\,t_3^{p_3}\,t_4^{p_4} \ .
\end{align}
Hence each term in the weight decomposition of the vector space $V_{\vec\sigma}$ corresponds to an element in the collection of solid partitions $\vec\sigma=(\sigma_1,\dots,\sigma_r)$.

After inserting the decompositions \eqref{decom_WM1} and \eqref{decom_V} into \eqref{chi}, we can use the equivariant characters $\chi_{\vec\sigma}$ to evaluate the instanton partition function \eqref{eq:ZC4matter} by using the operation
\begin{align}\label{eq:top}
\widehat{\tt e}\Big[\sum_I \, n_I\,\e^{\,{\sf w}_I}\Big] = \prod_{{\sf w}_I\neq0} \, {\sf w}_I^{n_I}
\end{align}
that converts the additive Chern characters to multiplicative (top form) Euler classes. 
Then ${\sf w}_I(\vec a,\vec\epsilon,\vec m)$ gives the weights of the $\sT$-action on the virtual tangent bundle and on the matter bundle.

Altogether this leads to the general combinatorial formula
\begin{align}
\begin{split}
    Z_{\FC^4}^{r,k}(\vec a,\vec\epsilon,\vec m) &=\sum_{|\vec{\sigma}|=k} \, (-1)^{{\tt O}_{\vec\sigma}} \ \widehat{\tt e}[-\chi_{\vec{\sigma}}] \\[4pt]
&=\sum_{|\vec{\sigma}|=k} \, (-1)^{{\tt O}_{\vec\sigma}} \ \prod_{l=1}^r \ \prod_{\vec p_l\in\sigma_l}^{\neq0} \, \frac{\CP_r(a_l+\vec p_l\cdot\vec\epsilon\,|\vec m)}{\CP_r(a_l+\vec p_l\cdot\vec\epsilon\,|\vec a)} \\
& \hspace{5cm} \times \prod_{l'=1}^r \ \prod_{\vec p^{\,\prime}_{l'}\in\sigma_{l'}}^{\neq0} \, \CR\big(a_l-a_{l'}+(\vec p_l-\vec p^{\,\prime}_{l'})\cdot\vec\epsilon\,\big|\vec\epsilon\,\big)
\end{split}\label{Zk}
\end{align}
for the partition function \eqref{partition}, where the polynomial $\CP_r$ and the rational function $\CR$ are defined in \eqref{eq:CPrCR}. The superscripts ${}^{\neq0}$ on the products indicate the omission of terms with zero numerator or denominator from the formula \eqref{Zk}, according to the prescription in \eqref{eq:top}. The sign factors $(-1)^{\ttO_{\vec\sigma}}$ are given by sums of cardinalities
\begin{align}\label{eq:signfactor}
{\tt O}_{\vec\sigma} = \sum_{l=1}^r \, \Big|\Big\{(\vec p_l,\vec p_{l}^{\,\prime})\in\sigma_l\times\sigma_{l} \ \Big| \ {\small \begin{matrix}(p_l)_\alpha\neq(p_{l}')_\alpha \\ (p_l)_\alpha-(p_l)_4+1 = (p_{l}')_\alpha - (p_{l}')_4 \end{matrix} } \normalsize \ , \ \alpha\in\{1,2,3\}\Big\}\Big| \ .
\end{align}

The full partition function \eqref{partitionfull}, including all instanton sectors, is given by
\begin{align}\label{Zfull}
Z_{\FC^4}^r(\qu;\vec a,\vec\epsilon,\vec m) = \sum_{\vec\sigma} \, (-1)^{{\tt O}_{\vec\sigma}} \ \qu^{|\vec\sigma|} \ \widehat{\tt e}[-\chi_{\vec\sigma}] \ ,
\end{align}
where the sum runs over all $r$-tuples $\vec\sigma=(\sigma_1,\dots,\sigma_r)$ of solid partitions $\sigma_l$ (including the empty solid partitions of sizes $|\sigma_l|=0$). 

\begin{remark}
The expression for ${\ttO_{\vec\sigma}}$ in \eqref{eq:signfactor} differs from the sign factor conjectured by Nekrasov and Piazzalunga in~\cite{m4c}, which is proved by Kool and Rennemo in~\cite{KRinprep}, and also from those of~\cite{Cao:2019tvv,Oh:2020rnj,CKMpreprint}, because our choice of square root \eqref{eq:chTvir} is different. It was found by comparing the character  $\sqrt{\ch_{\sT}}\big(T_{\vec\sigma}^{\rm vir}\frM_{r,k}\big)$ in \eqref{eq:chTvir} with the character of \cite[Section~{5.7}]{CKMpreprint} and {\cite{KRinprep}}: they differ in the sign factor $(-1)^{k\,(r-1) + {\rm rank}\, \ttA_{\vec\sigma}}$ given by
\begin{align}\label{eq:rank_A}
    \ttA_{\vec\sigma}= \big(V_{\vec\sigma}^*\otimes V_{\vec\sigma} \ t_4^{-1}\big)^{\rm fix} \ ,
\end{align}
where the superscript ${}^{\rm fix}$ stands for the \smash{$\sT_{\vec a,\vec\epsilon}\,$-fixed} part. This is tantamount to counting the zeroes of \smash{$a_l-a_{l'}+(\vec p_l-\vec p^{\,\prime}_{l'})\cdot\vec\epsilon-\epsilon_4$},
for $\vec p_l\in \sigma_l$ and $\vec p^{\,\prime}_{l'}\in\sigma_{l'}$.
Then the sign factor \eqref{eq:signfactor} is calculated as the difference of \eqref{eq:rank_A} and the sign factor from~\cite{CKMpreprint}.

{Since our choice of square root differs from the square root of~\cite{CKMpreprint} by conjugation of a finite number of terms, which is shown to be movable in~\cite{m4c}, our square root is movable as well.\footnote{We thank an anonymous referee for pointing this out to us.} That is, each term contains weights $t_a$ and $e_l$.}

With our choice of square root $\sqrt{\ch_{\sT}}\big(T_{\vec\sigma}^{\rm vir}\frM_{r,k}\big)$ in \eqref{eq:chTvir}, and subsequently of the index $\chi_{\vec\sigma}$, the first term affected by the sign $(-1)^{\ttO_{\vec\sigma}}$ appears at instanton number $k=16$.
\end{remark}

The formula \eqref{Zk} makes evident the evaluation of \eqref{partition} as a sum over residues; the signs $(-1)^{\ttO_{\vec\sigma}}$ should then come from a careful residue calculation. From this perspective the dimensional reduction of Proposition~\ref{prop:ZDTgeb} can be seen to arise through the introduction of extra fixed parts to the index $\chi_{\vec\sigma}$ from terms involving $f^{-1}_l\,e_l\,t_4$, which keeps only contributions to \eqref{Zfull} from solid partitions $\vec\sigma$ with $(p_l)_4=1$ for all $\vec p_l\in\sigma_l$, $l=1,\dots,r$; these correspond to arrays $\vec\pi$ of plane partitions (cf.~Remark~\ref{not:solidpartitions}) and the sign factors \eqref{eq:signfactor} become $\ttO_{\vec\pi}=|\vec\pi|$.
In fact, it is possible to ``uplift'' the result of Proposition~\ref{prop:ZDTgeb} to infer

\begin{conjecture}\label{Prop1}
The equivariant instanton partition function of the cohomological gauge theory with a massive fundamental hypermultiplet on $\FC^4$ is given by
\begin{align}
Z_{\FC^4}^r(\qu;\vec{a},\vec{\epsilon},\vec{m})=Z_{\FC^4}^r(\qu;\vec{\epsilon},m)=M(-\qu)^{-\frac{r\,m\,\epsilon_{12}\,\epsilon_{13}\,\epsilon_{23}}{\epsilon_1\,\epsilon_2\,\epsilon_3\,\epsilon_4}} \qquad \mbox{with} \quad m=\frac{1}{r} \, \sum_{l=1}^r\,(m_l-a_l) \ . \label{PfRank}
\end{align}
\end{conjecture}

\begin{remark}
We checked explicitly that the order $\qu$ and $\qu^2$ terms of the series \eqref{Zfull} agree with \eqref{PfRank} in the rank one case $r=1$, as well as with an elementary Cauchy residue evaluation of the contour integrals in \eqref{partition} if one considers contours $\varGamma_{r,k}\subset\FC^k$ which enclose all singularities of the integrand. The rank one case of {Conjecture}~\ref{Prop1} was conjectured by Cao and Kool in~\cite{Cao:2017swr}. The $k$-instanton contributions with $k\geq1$ give the rank $r$ equivariant Donaldson--Thomas invariants of~$\FC^4$:
\begin{align}
\int^\sT_{[\frM_{r,k}]_\tto^{\rm vir}} \, e_\sT(\CCV_{r,k}\otimes\mbf r) = \sum_{n=1}^k \, \frac{1}{n!} \, \Big(\frac{r\,m\,\epsilon_{12} \, \epsilon_{13} \, \epsilon_{23}}{\epsilon_1\,\epsilon_2\,\epsilon_3\,\epsilon_4}\Big)^n \ \sum_{\substack{k_1,\dots,k_n\geq1 \\ k_1+\cdots+ k_n=k}} \ \prod_{i=1}^n \ \sum_{d|k_i} \, \frac{k_1\cdots k_n}{d^2} \ .
\end{align}

Following~\cite{Moore:1997dj,Moore:1998et}, one can write a well-defined general integration prescription starting from \eqref{partition} which calculates the partition function~\emph{$Z_{\FC^4}^{r,k}(\vec a,\vec\epsilon,\vec m)$} by uplifting the theory to an $8{+}1$-dimensional theory on $\FC^4\times \sone_\beta^1$, where $\sone_\beta^1$ is the circle of circumference $\beta$. This defines the K-theory version of the gauge theory, wherein the eight-dimensional instantons are viewed as constant loops on \smash{$\FC^4\times\sone_\beta^1$}; the original theory is recovered in the cohomological limit $\beta\rightarrow 0$ where the circle shrinks to a point.  We choose as integration contour for the uplifted integral any contour large enough to enclose all singularities of the integrand, and finally we apply a residue theorem to evaluate the contour integrals. 

Nekrasov and Piazzalunga propose a plethystic exponential formula in~\cite{m4c} for the uplifted instanton partition function on \smash{$\FC^4\times \sone_\beta^1$}, generalizing the rank one conjecture of~\cite{m4}; in their approach, the sign factors arise from evaluation of the Jeffrey--Kirwan residue formula. Their conjectured formula is proven by Kool and Rennemo in~\cite{KRinprep}. It is straightforward to show that this formula reproduces {Conjecture}~\ref{Prop1} in the cohomological limit (see~\cite{m4,m4c,Cao:2019tvv} for the rank one  case). 

In Appendix~\ref{app:Prop1} we sketch the steps of a direct but less conceptual proof of the formula \eqref{PfRank} using the combinatorial evaluation \eqref{Zk} of the quiver matrix model \eqref{partition}, which is based on the dimensional reduction of Proposition~\ref{prop:ZDTgeb}. We believe that such an argument, while presently incomplete, provides useful insights into the symmetries of the theory, which are not evident through other approaches. These arguments can also be extended to more general settings, such as some of our orbifold theories below, for which rigorous results are not yet available.
\end{remark}

\subsection{Pure $\mathcal{N}_{\textrm{\tiny T}}=2$ Gauge Theory}
\label{sec:puregaugetheory}

As we discussed in Section~\ref{sec:instantoncounting}, Nekrasov's $\varOmega$-deformation permits the definition of the equivariant partition function \eqref{eq:ZC4pure} for the  eight-dimensional cohomological gauge theory even in the absence of fundamental matter fields. It should come as no surprise that the only difference from our previous calculations is the absence of the matter bundle $\CCV_{r,k}\otimes\mbf r$, i.e. the theory is based entirely on the generalized ADHM data. The same BRST construction from Section~\ref{BRST} applies by dropping the matter field contribution to the path integral, giving the equivariant integral
\begin{align}\label{eq:Zrkpurecontour}
\begin{split}
Z_{\FC^4}^{r,k}(\vec a,\vec\epsilon\,)^{\rm pure} :\!&=\int^{\sT_{\vec a,\vec\epsilon}}_{[\frM_{r,k}]_\tto^{\rm vir}} \, 1 \\[4pt]
&= \frac{(-1)^k}{k!} \, \Big(\frac{\epsilon_{12} \, \epsilon_{13} \, \epsilon_{23}}{\epsilon_1\,\epsilon_2\,\epsilon_3\,\epsilon_4}\Big)^k \, \oint_{\varGamma_{r,k}} \ \prod_{i=1}^k \, \frac{\dd\phi_i}{2\pi\,\ii} \, \frac{1}{\CP_r(\phi_i|\vec a)} \ \prod_{\stackrel{\scriptstyle i,j=1}{\scriptstyle i\neq j}}^k \, \CR(\phi_i-\phi_j|\vec\epsilon\,)
\end{split}
\end{align}
from the field theory perspective.

The contour integral \eqref{eq:Zrkpurecontour} can be computed from the Chern character \eqref{eq:chTvir} alone, giving the combinatorial expansion into solid partitions
\begin{equation}\label{ZNkpure}
\begin{split}
Z_{\FC^4}^{r,k}(\vec a,\vec\epsilon\,)^{\rm pure} &= \sum_{|\vec{\sigma}|=k} \, (-1)^{{\tt O}_{\vec\sigma}} \ \prod_{l=1}^r \ \prod_{\vec p_l\in\sigma_l}^{\neq0} \, \frac{1}{\CP_r(a_l+\vec p_l\cdot\vec\epsilon\,|\vec a)} \\
& \hspace{5cm} \times \prod_{l'=1}^r \ \prod_{\vec p^{\,\prime}_{l'}\in\sigma_{l'}}^{\neq0} \, \CR\big(a_l-a_{l'}+(\vec p_l-\vec p^{\,\prime}_{l'})\cdot\vec\epsilon\,\big|\vec\epsilon\,\big) \ .
\end{split}
   \end{equation}
The instanton partition function \eqref{eq:ZC4pure} is then given by
\begin{align}\label{Zfullpure}
Z_{\FC^4}^r(\Lambda;\vec a,\vec\epsilon\,)^{\rm pure} = 1 + \sum_{k=1}^\infty\,\Lambda^k \ Z_{\FC^4}^{r,k}(\vec a,\vec\epsilon\,)^{\rm pure} \ .
\end{align}

From the field theory perspective, we expect that \eqref{eq:ZC4matter} reduces to \eqref{eq:ZC4pure} in a suitable infinite mass limit which decouples the fundamental matter hypermultiplets~\cite{m4}. This physical expectation is confirmed by

\begin{proposition}\label{prop:puremassiverel}
The equivariant instanton partition function for the pure cohomological gauge theory is related to the partition function with a massive fundamental hypermultiplet on $\FC^4$ through the double scaling limit
\begin{align}\label{eq:puremassiverel}
Z_{\FC^4}^r(\Lambda;\vec a,\vec\epsilon\,)^{\rm pure} = \lim_{m_1,\dots,m_r\to\infty} \ \lim_{\qu\to0} \, Z_{\FC^4}^r(\qu;\vec a,\vec\epsilon,\vec m)\,\Big|_{\Lambda=(-1)^r\,m_1\cdots m_r\,\qu} \ .
\end{align}
\end{proposition}

\proof
From \eqref{Zk} the relevant terms of the combinatorial expansion \eqref{Zfull} in the limit of large masses are given by
\begin{align}
\begin{split}
\qu^{|\vec\sigma|} \, \prod_{l=1}^r \ \prod_{\vec p_l\in\sigma_l}^{\neq0} \, \CP_r(a_l+\vec p_l\cdot\vec\epsilon\,|\vec m) \xrightarrow{m_1,\dots,m_r\gg1}\qu^{|\vec\sigma|} \, \prod_{l=1}^r \ \prod_{\vec p_l\in\sigma_l} \, (-1)^r \, & m_1\cdots m_r \\[4pt]
&= \big( (-1)^r \, m_1\cdots m_r\,\qu\big)^{|\vec\sigma|} \ ,
\end{split}
\end{align}
and the result now follows from \eqref{ZNkpure} and \eqref{Zfullpure}.
\endproof

\begin{corollary}\label{prop:pureC4}
{Assume Conjecture \ref{Prop1} is true}. Then the equivariant instanton partition function of the pure cohomological gauge theory on $\FC^4$ is given by
\begin{align}\label{PfRankpure}
Z_{\FC^4}^{r}(\Lambda;\vec\epsilon\,)^{\rm pure}=\left\{ \begin{matrix} \displaystyle \exp\Big(-\Lambda\ \frac{\epsilon_{12}\,\epsilon_{13}\,\epsilon_{23}}{\epsilon_1\,\epsilon_2\,\epsilon_3\,\epsilon_4}\,\Big) \qquad \mbox{for} \quad r=1 \ , \\[3ex] 1 \qquad \mbox{for} \quad r>1  \ . \end{matrix} \right. 
\end{align}
\end{corollary}

\proof 
We combine Propositions~\ref{prop:puremassiverel} and {Conjecture}~\ref{Prop1} using the series representation
\begin{align}
\log M(q) = \sum_{n=1}^\infty\,\frac{1}{n}\,\frac{q^n}{(1-q^n)^2}
\end{align}
for the logarithm of the MacMahon function \eqref{eq:MacMahon}, which converges for any $q\neq1$. For $\qu\ll1$ this yields the behaviour $\log M(-\qu)=-\qu+O(\qu^2)$. For $r=1$ the result follows by fixing $\Lambda=-m\,\qu$ in the limits $m\to\infty$ and $\qu\to0$. For $r>1$, the partition function is identically equal to one in the double scaling limit $m_l\to\infty$ and $\qu\to0$ with $(-1)^r\,m_1\cdots m_r\,\qu=\Lambda$ fixed.
\endproof

\begin{remark}
We checked explicitly that the order $\Lambda$ and $\Lambda^2$ terms of the series \eqref{Zfullpure} agree with \eqref{PfRankpure} in the rank one case $r=1$, and that the contributions for $r=2$ with $k=1,2$ and for $r=3$ with $k=1$ are trivial: $
Z_{\FC^4}^{r,k}(\vec a,\vec\epsilon\,)^{\rm pure}=0$.

Geometrically, the large mass limit of Proposition~\ref{prop:puremassiverel} defines the insertion-free equivariant Donaldson--Thomas invariants of $\FC^4$~\cite{Cao:2019tvv}. The rank one formula of Corollary~\ref{prop:pureC4} was originally discussed in~\cite{Cao:2017swr,Cao:2019tvv}. It implies that the equivariant volume of the Hilbert scheme of $k$ points on $\FC^4$ is given by
\begin{align}
\int_{[\Hilb^k(\FC^4)]_\tto^{\rm vir}}^{\sT_{\vec\epsilon}} \, 1 = \frac{(-1)^k}{k!} \, \Big(\frac{\epsilon_{12} \, \epsilon_{13} \, \epsilon_{23}}{\epsilon_1\,\epsilon_2\,\epsilon_3\,\epsilon_4}\Big)^k \ ,
\end{align}
whereas the equivariant volume of the Quot scheme \smash{$\Quot_{r}^k(\FC^4)$} of zero-dimensional quotients of $\CO_{\FC^4}^{\oplus r}$ with length $k$ vanishes for all $r>1$. This is in marked contrast to the equivariant volumes of the Quot schemes of zero-dimensional quotients of $\CO_{\FC^2}^{\oplus r}$ and $\CO_{\FC^3}^{\oplus r}$, which are all non-zero.
\end{remark}

\section{Orbifolds of the Eight-Dimensional Theory}\label{Orb8d}

In Section~\ref{sec:Instantons} we have studied eight-dimensional noncommutative instantons on flat space. As a natural non-trivial generalization beyond flat space $\FC^4$, on the way towards local Calabi--Yau fourfolds, in this section we will extend these considerations to quotients $\FC^4/\sGamma$ by a finite group $\sGamma$ acting linearly on $\FC^4$. In order to preserve the holonomy group $\sSU(4)$ of the cohomological gauge theory, the group $\sGamma$ has to be a subgroup of $\sSU(4)$. To define the equivariant instanton partition functions, we will further restrict to toric Calabi--Yau orbifolds $\FC^4/\sGamma$, which requires that $\sGamma$ is abelian and that the $\sGamma$-action commutes with the action of the maximal torus $\sT_{\vec\epsilon}\subset \sSU(4)$ on $\FC^4$. This restricts to orbifold groups of the form $\sGamma=\RZ_{n_1}\times\RZ_{n_2}\times\RZ_{n_3}$, with order $n=n_1\,n_2\,n_3$. Such an orbifold theory was studied in~\cite{Bon,Kimura:2022zsm} for a specific example; here we will vastly extend and generalize their considerations using the framework of Section~\ref{sec:Instantons}.
 
\subsection{Noncommutative Gauge Theory on $\FC^4/\sGamma$}
\label{sec:NCorbifold}

A conventional quantum field theory on an orbifold $\FC^4/\sGamma$ is singular. However, an orbifold field theory can be constructed by allowing fields to be equivariant under the action of $\sGamma$ and gauging the $\sGamma$-action, together with a subsequent projection to $\sGamma$-invariant states of the cohomological gauge theory of Section~\ref{sec:coho_gauge_theory}. Such a construction incorporates ``twisted sectors'' into the theory, analogously to  string orbifolds, and can be thought of as a cohomological gauge theory on the quotient stack $[\FC^4/\sGamma]$. It is naturally realised in noncommutative field theory.

Suppose that the generators of the orbifold group $\sGamma$ act on the coordinates of $\FC^4$ as
\begin{align}
(z_1,z_2,z_3,z_4)\longmapsto(\e^{\,2\pi\,\ii\, s_1/n}\,z_1,\e^{\,2\pi\,\ii\, s_2/n}\,z_2,\e^{\,2\pi\,\ii\, s_3/n}\,z_3,\e^{\,2\pi\,\ii\, s_4/n}\, z_4) \ ,
\end{align}
where $n$ is the order of $\sGamma$.
This induces a decomposition of the fundamental representation $Q\simeq\FC^4$ of $\sSL(4,\FC)$ as 
\begin{align}
Q=\rho_{s_1}\oplus\rho_{s_2}\oplus\rho_{s_3}\oplus\rho_{s_4} \ , 
\end{align}
where $\rho_{s_a}$ denotes the irreducible representation of $\sGamma$ with weight $s_a$. The Calabi--Yau constraint implies $\rho_{s_1}\otimes\cdots\otimes\rho_{s_4}\simeq\rho_0$, where $\rho_0$ is the trivial representation of weight zero.

Let $\widehat{\sGamma}$ denote the finite abelian group of irreducible representations of $\sGamma$; since each representation $\rho_s$ for $s\in\widehat{\sGamma}$ is one-dimensional, we can regard elements of $\widehat{\sGamma}$ as characters $\chi_s:\sGamma\longrightarrow\FC^\times$ of $\sGamma$. Then the decomposition of $Q$ also defines a $\widehat{\sGamma}$-colouring $\RZ_{\geq0}^{\oplus 4}\longrightarrow \sGamma$ through the identification \smash{$\widehat{\sGamma}\simeq \sGamma$} by
\begin{align}
    (n_1,n_2,n_3,n_4)\longmapsto\rho_{s_1}^{\otimes n_1}\otimes \rho_{s_2}^{\otimes n_2}\otimes \rho_{s_3}^{\otimes n_3}\otimes\rho_{s_4}^{\otimes n_4} \ .\label{coloring}
\end{align}
Via this colouring, the $\sGamma$-action on $\FC^4$ induces an isotopical decomposition of the Fock space \eqref{eq:fock_space} into irreducible representations as
\begin{align}
\CH=\bigoplus_{s\in\widehat{\sGamma}}\, \CH_s \qquad \mbox{with} \quad
\CH_s={\sf Span}_{\FC}\big\{|\vec n\,\rangle \ \big| \ \rho_{s_1}^{\otimes n_1}\otimes\cdots\otimes\rho_{s_4}^{\otimes n_4}\simeq \rho_s  \big\} \ .
\end{align}

Consequently, in the rank one case $r=1$ the covariant coordinates $Z_a$ decompose into maps $Z=\bigoplus_{s\in\widehat{\sGamma}}\, (Z_a^s)_{a\in\ulfour}:\CH\longrightarrow Q\otimes\CH$ with $Z_a^{s}: \CH_s\longrightarrow \CH_{s+s_a}$  for $a\in\ulfour\,$, and the Higgs field $\varphi$ into maps $\varphi=\bigoplus_{s\in\widehat{\sGamma}}\,\varphi^s:\CH\longrightarrow\CH$ with $\varphi^{s}:\CH_s\longrightarrow \CH_s$. The operator algebraic equations \eqref{Zin} now read as
\begin{align}\label{eq:Zorb}
\begin{split}
Z_{a}^{s+s_b}\,Z_b^{s}-Z_b^{s+s_a}\,Z_a^{s}&=\tfrac{1}{2}\,\varepsilon_{ab\bar{c}\bar{d}}\,\big(Z_{\bar{c}}^{s-s_d-s_c\dagger }\,Z_{\bar{d}}^{s-s_d\dagger}-Z_{\bar{d}}^{s-s_d-s_c\dagger}\,Z_{\bar{c}}^{s-s_c\dagger}\big) \ , \\[4pt]
\sum_{a=1}^4\,\big( Z_{\bar a}^{s\dagger}\,Z_a^{s}-Z_{a}^{s-s_a}\,Z_{\bar{a}}^{s-s_a\dagger}\big)&= \zeta_s \, \ident_{\CH_s} \ , \\[4pt]
Z_a^{s}\,\varphi^{s}-\varphi^{s+s_a}\, Z_{a}^{s}&=\epsilon_a\,Z_a^{s} \ ,
\end{split}
\end{align}
where the Fayet--Iliopoulos parameters $\zeta_s>0$ for $s\in\shGamma$ are determined by the decomposition of the $B$-field into NS--NS twisted sectors of type~II string theory on $\FC^4/\sGamma$. Note that the first equation of \eqref{eq:Zorb} makes sense because of the Calabi--Yau condition $\rho_{s_{1234}}\simeq\rho_0$.

The generalization of the set of equations \eqref{eq:Zorb} to higher rank $r>1$ involves an action of the orbifold group $\sGamma$ on the Chan--Paton space $W\simeq\FC^r$ that we describe below, which is defined by a homomorphism $\sGamma\longrightarrow\sU(r)_{\rm col}$. This breaks the colour symmetry to the centralizer of the image of $\sGamma$. 
Solutions of these equations describe noncommutative $\sU(r)$ instantons on the orbifold $\FC^4/\sGamma$.

The splitting of the covariant coordinates induces an equivariant decomposition of the ADHM data $({B}_a,I)_{a\in\ulfour}$ from \eqref{eq:Zamatrices}. 
With respect to the decomposition of the orbifold group action into irreducible representations, the vector spaces on which the ADHM variables are defined decompose as 
\begin{align}
    V=\bigoplus_{s\in\widehat{\sGamma}} \, V_s\otimes \rho_s^* \qquad \mbox{and} \qquad
    W=\bigoplus_{s\in\widehat{\sGamma}} \, W_{s}\otimes\rho_s^* \ .
\end{align}
The multiplicity spaces $V_s=\sHom_{\FC[\sGamma]}(\rho_s,V)$ and $W_s=\sHom_{\FC[\sGamma]}(\rho_s,W)$ consist of $\sGamma$-equivariant homomorphisms, where $\FC[\sGamma]$ is the group ring of $\sGamma$ over $\FC$.
The dimensions $k= \dim_\FC V$ and $ r= \dim_{\FC} W$ correspondingly decompose as
\begin{align}
    k=\sum_{s\in\widehat{\sGamma}} \, k_s=\sum_{s\in\widehat{\sGamma}}\, \dim_{\FC} V_s \qquad \mbox{and} \qquad r=\sum_{s\in\widehat{\sGamma}}\, r_s = \sum_{s\in\widehat{\sGamma}}\, \dim_{\FC} W_s \ ,
\end{align}
which define arrays of fractional instanton charges $\vec k=(k_s)_{s\in\widehat\sGamma}$ with ranks $\vec r = (r_s)_{s\in\shGamma}$ in $\RZ_{\geq0}^{|\sGamma|}$, whose \emph{size} is \smash{$|\vec k\,|:=\sum_{s\in\shGamma}\,k_s=k$} and $|\vec r\,|:=\sum_{s\in\shGamma} \, r_s:=r$. 

By Schur's lemma this implies the decompositions
\begin{align}\label{eq:ADHMdecompB}
B=\bigoplus_{s\in\widehat{\sGamma}}\,(B_a^s)_{a\in\ulfour} \ \in \  \sHom_{\FC[\sGamma]}(V, V\otimes Q) \qquad \mbox{with} & \quad B^s_a: V_{s}\longrightarrow V_{s+s_a} \ , \\[4pt] \label{eq:ADHMdecompI}
I=\bigoplus_{s\in\widehat{\sGamma}}\, I^s \ \in \ \sHom_{\FC[\sGamma]}(W,V) \qquad \mbox{with} & \quad I^s: W_{s}\longrightarrow V_s \ .
\end{align}
Starting from the equations \eqref{eq:Zorb}, the same argument from Section~\ref{sec:ADHMC4} then defines the orbifold version of the ADHM equations \eqref{ADHMeq}:
\begin{align}\label{eq:ADHMOrb1}
\begin{split}
\mu_{ab}^{\FC s} &:= B_a^{s+s_b}\,B_b^{s}-B_b^{s+s_a}B_a^{s}-\tfrac{1}{2}\,\varepsilon_{ab\bar c\bar d}\, \big(B_{\bar c}^{s-s_d-s_c\dagger}\,B_{\bar d}^{s-s_d \dagger}-B_{\bar d}^{s-s_d-s_c\dagger}\,B_{\bar c}^{s-s_c \dagger} \big) = 0 \ , \\[4pt]
\mu^{\FR s} &:= \sum_{a=1}^4\, \big( B_a^{s-s_a}\,B_{\bar a}^{s-s_a \dagger}-B_{\bar a}^{s \dagger}\,B_a^s \big)+I^{s}\,I^{s \dagger}=\zeta_s \, \ident_{ k_s} \ ,
\end{split}
\end{align}
for each $s\in\shGamma$ and $1\leq a<b\leq4$.

The $\CN=(0,2)$ gauge theories on D1-branes at Calabi--Yau orbifolds of $\FC^4$ are also obtained by orbifold projection and have been studied by e.g.~\cite{Mohri:1997ef,Garcia-Compean:1998sla,Franco:2015tna} (in the absence of D1--D9 strings). This theory is anomaly-free with suitable Chern--Simons couplings to background R--R chiral scalars. It consists of four types of chiral superfields on $\FR^{1,1}$, transforming in the bifundamental representation of $\sU(k_{s+a})\times\sU(k_s)$ for $a\in\ulfour$ and $s\in\shGamma$, whose lowest components are the linear maps $B_a^s$. The choice of $k_s=n$ for all $s\in\shGamma$ corresponds to a stack of $n$ regular D1-branes, while more general dimension vectors $\vec k$ correspond to fractional D1-branes. The ADHM equations \eqref{eq:ADHMOrb1} are equations for the Higgs branch of this $\CN=(0,2)$ theory.

\subsubsection*{Moduli Spaces of Orbifold Instantons}

The moduli space $\frM_{r,k}^\sGamma$ of charge $k$ noncommutative $\sU(r)$ instantons on the orbifold $\FC^4/\sGamma$ stratifies into connected components 
\begin{align}
\frM_{r,k}^\sGamma = \bigsqcup_{|\vec r\,|=r \,,\, |\vec k\,|=k} \, \frM_{\vec r,\vec k}
\end{align}
according to the decompositions above, where the disjoint union runs over all $\sGamma$-representations $W$ and $V$ of dimensions $r$ and $k$, respectively. The $\sGamma$-action breaks the $\sU(k)$ gauge symmetry of $V$ to the subgroup $\timesbig_{s\in\shGamma}\,\sU(k_s)$, whose action on the ADHM variables is given by
\begin{align}
g\cdot\big(B_a^s\,,\,I^s\big)_{\substack{s\in\shGamma \\ a\in\ulfour}} = \big(g_{s+s_a}\, B_a^s\,g_s^{-1}\,,\,g_s\,I^s\big)_{\substack{s\in\shGamma \\ a\in\ulfour}} \qquad \mbox{with} \quad g=(g_s)_{s\in\shGamma} \ \in \ \Timesbig_{s\in\widehat{\sGamma}}\,\sU(k_s) \ .
\end{align}
Then \smash{$\frM_{\vec r,\vec k}$} can be described as the space of solutions to \eqref{eq:ADHMOrb1} modulo these gauge transformations.

The moduli space \smash{$\frM_{\vec r,\vec k}$} admits an equivalent holomorphic parametrization as a GIT quotient. 
For this, we note that the equations \smash{$ \mu^{\FC s}_{ab}=0$} from \eqref{eq:ADHMOrb1} arise as the complex moment map equations $\mu_{ab}^\FC=0$ from \eqref{ADHMeq} for the equivariant decomposition $\mu_{ab}^\FC=\bigoplus_{s\in\shGamma}\,\mu_{ab}^{\FC s}$. Since the latter are equivalent to the equations $[B_a,B_b]=0$ by the identity \eqref{eq:muabCholomorphic}, the equivariant decomposition \eqref{eq:ADHMdecompB} shows that the former can be substituted with the holomorphic equations
\begin{align}\label{eq:orbADHM}
B_a^{s+s_b}\,B_b^{s}=B_b^{s+s_a}\,B_a^s \ ,
\end{align}
for $1\leq a<b\leq4$ and $s\in\widehat{\sGamma}$.

On the other hand, the real moment map equations $\mu^{\FR s}=\zeta_s\,\ident_{k_s}$ in \eqref{eq:ADHMOrb1} for generic Fayet--Iliopoulos parameters $\zeta_s>0$ can be traded for a stability condition: A set of maps \eqref{eq:ADHMdecompB} and \eqref{eq:ADHMdecompI} is \emph{stable} if there are no proper $\sGamma$-subrepresentations
\begin{align}
S = \bigoplus_{s\in\shGamma} \, S_s\otimes\rho_s^*
\end{align}
of $V$ such that $B_a^s(S_s)\subset S_{s+s_a}$ for $a\in\ulfour$ and ${\rm im}(I^s)\subset S_s$ for all $s\in\shGamma$. 

The proof is similar to the stability proof of~\cite{Nekrasov:2015wsu}: Let $P_{s}$ be the orthogonal projection of $V_s$ to the orthogonal complement $S_s^\perp$ of the invariant subspace $S_s\subset V_s$, for each $s\in\shGamma$. Then $P_{s}\,I^s=0$ and $P_{s+s_a}\,B_a^s\,P_{s}=P_{s+s_a}\,B_a^s$, so
\begin{align}
\begin{split}
0 \ \leq \ \sum_{s\in\shGamma}\,\zeta_s \dim_\FC S_s^\perp = \sum_{s\in\shGamma} \, \Tr\big(P_{s}\,\mu^{\FR s}\big) 
&= \sum_{s\in\shGamma} \ \sum_{a\in\ulfour}\,\Tr\big(P_{s+s_a}\,B_a^{s}\,B_{\bar a}^{s \dagger}-B_a^s\, P_s\, B_{\bar a}^{s \dagger}\big) \\[4pt]
&= -\sum_{s\in\shGamma} \ \sum_{a\in\ulfour}\, \big\|(\ident_{k_{s+s_a}}-P_{s+s_a})\,B^s_a\,P_{s}\big\|_{\textrm{\tiny F}}^2 \ \leq \ 0 \ ,
\end{split}
\end{align}
which implies $S_s=V_s$ for all $s\in{\shGamma}$. 
The stability condition is equivalent to the condition that the actions of the operators $B_a^s$ for $a\in\ulfour$ and $s\in\shGamma$ on \smash{$I^{s'}(W_{s'})$} generate the isotopical subspaces $V_{s''}$. Conversely, the orbit of any set of stable maps $(B,I)$ under 
\begin{align}\label{eq:sGvec}
\sG_{\vec k}:=\Timesbig_{s\in\widehat{\sGamma}}\,\sGL(k_s,\FC)
\end{align}
intersects the locus~\smash{$\bigoplus_{s\in\shGamma}\,\mu^{\FR s}=\bigoplus_{s\in\shGamma}\,\zeta_s\,\ident_{k_s}$}. 

Then \smash{$\frM_{\vec r,\vec k}$} can be described as the space of stable maps $(B,I)$ which satisfy the holomorphic equations \eqref{eq:orbADHM}, modulo the natural action of the complex gauge group \smash{$\sG_{\vec k}$} of the $\sGamma$-module $V$.
It follows that the real virtual dimension of the moduli space  \smash{$\frM_{\vec r,\vec k}$} is given by
\begin{align}
{\rm vdim}_\FR\, \frM_{\vec r,\vec k} = \sum_{s\in\shGamma} \, \Big(2 \,  \sum_{a=1}^4 \, k_s\,k_{s+s_a} + 2 \, r_s\,k_s - \sum_{1\leq a<b\leq4} \, k_s\,k_{s+s_{ab}} - 2 \, k_s^2\Big) \ .
\end{align}

The moduli space  \smash{$\frM_{\vec r,\vec k}$} has a remaining symmetry under framing rotations: The global colour symmetry $\sU(r)_{\rm col}$ is broken to the centralizer $\timesbig_{s\in\shGamma}\,\sU(r_s)_{\rm col}$, which acts on the orbifold ADHM data as
\begin{align}
h\cdot\big(B_a^s\,,\,I^s\big)_{\substack{s\in\shGamma \\ a\in\ulfour}} = \big( B_a^s\,,\,I^s\,h_s^{-1}\big)_{\substack{s\in\shGamma \\ a\in\ulfour}} \qquad \mbox{with} \quad h=(h_s)_{s\in\shGamma} \ \in \ \Timesbig_{s\in\widehat{\sGamma}}\,\sU(r_s)_{\rm col} \ .
\end{align}
Its maximal torus is
\begin{align}
\sT_{\vec a} = \Timesbig_{s\in\shGamma}\,\sT_{\vec a^s} \ , 
\end{align}
where $\sT_{\vec a^s}$ is the maximal torus of $\sU(r_s)_{\rm col}$.

In addition,  \smash{$\frM_{\vec r,\vec k}$} has an $\sSU(4)$ symmetry inherited from the holonomy group of the toric Calabi--Yau four-orbifold \smash{$\FC^4/\sGamma$}. In particular, the action of the (complexified) maximal torus $\sT_{\vec\epsilon}$ is the evident descent of \eqref{eq:torusquiver}.

\subsection{Crepant Resolutions}\label{Orbcon}

While our constructions of orbifold instanton partition functions will hold quite generally for any toric Calabi--Yau four-orbifold, a central role is played by those orbifolds which admit a \emph{crepant resolution} \cite{crepant}: Recall that, for a singular variety $Y$, a proper algebraic map
 $\pi: X\longrightarrow Y$ is a {crepant resolution} if $X$ is smooth and $\pi$ is a birational morphism which preserves the canonical bundles, i.e. $K_X\simeq \pi^*K_{Y}$.
We are interested in the case where $Y=\FC^4/\sGamma$, or  more generally $Y=\FC^d/\sGamma$. The existence of a crepant resolution then requires the group $\sGamma$ to be a finite subgroup of $\sSL(d,\FC)$~\cite{crepant1,crepant2}. 

Type~II string theory on an orbifold is defined by imposing equivariance under the action of the finite group $\sGamma$. When all twisted NS--NS sectors have age one, a crepant resolution is induced by vacuum expectation values of scalar fields in these sectors. This allows one to continuously smoothen the quotient singularities while preserving the Calabi--Yau properties. The spectrum of the orbifold string theory is then the spectrum of type~II string theory on a compactification over a smooth Calabi--Yau space obtained by blow-up of the orbifold singularities~\cite{AsyOrb,Alexandrov:2011va}, whose sizes are controlled by the moduli of the scalar fields. In particular, the BPS states can be considered in `orbifold' and `large radius' phases, which are related by collapsing the compact cycles of the resolution~\cite{Douglas:2000qw}. For example, partition functions of fractional D-branes on orbifolds can be related to partition functions of D-branes wrapping compact cycles of a crepant resolution through changes of variables and wall-crossing formulas.

It is well-known that the existence of a crepant resolution of $\FC^d/\sGamma$ depends dramatically on the dimension $d$ of the orbifold \cite{crepant3}. In dimensions two and three such resolutions always exist~\cite{ItoNakamura,Bridgeland:2001xf}, and in particular in dimension two they are also unique. In these cases a crepant resolution is always provided by the Hilbert--Chow morphism from the
$\sGamma$-Hilbert scheme  $X=\Hilb^\sGamma(\FC^d)$  of 
$\sGamma$-invariant zero-dimensional subschemes $Z\subset \FC^d$ of length $|\sGamma|$ whose global sections $H^0(\mathcal{O}_Z)$ form the regular representation $R$ of $\sGamma$, which are called \emph{$\sGamma$-clusters} in $\FC^d$; it can be parametrized by $\sGamma$-invariant ideals $J\subset\FC[z_1,\dots,z_d]$ of codimension $|\sGamma|$,  where $\sGamma$ acts on the coordinate functions $z_1,\dots,z_d$. 
On the other hand, for $d\geq 4$ little is known: they exist in rather special cases.  A simple example for which no crepant resolution exists is given by the action of the generator of $\sGamma=\RZ_2$ on $\FC^4$ by simultaneous reflection $z_a\longmapsto-z_a$ for $a\in\ulfour$. It is unknown if \emph{a priori} there exists a crepant resolution given an orbifold group action, and even when one does it is not necessarily given by the $\sGamma$-Hilbert scheme, which generally behaves badly and can be singular. 

In the type~II string theory picture, the absence of a crepant resolution means that compactification on $\FC^d/\sGamma$ does not yield the requisite moduli fields from the twisted NS--NS sectors. For the worldvolume gauge theories of D-branes localized at the  singularity of a Calabi--Yau orbifold $\FC^d/\sGamma$ for $d=2$~\cite{Douglas:1996sw,Johnson:1996py} or $d=3$~\cite{Douglas:1997de}, the Higgs moduli space is a smooth Calabi--Yau manifold which is a resolution of the original orbifold, for generic choices of Fayet--Iliopoulos parameters. In marked contrast, for $d=4$ the Higgs moduli space is not necessarily a smooth Calabi--Yau fourfold even if the orbifold admits a smooth resolution: \smash{$\frM_{\vec r,\vec k}^\sGamma$} is a Calabi--Yau fourfold only if $\FC^4/\sGamma$ admits a crepant blow-up to it~\cite{Mohri:1997ef}.

Now suppose that $\sGamma$ is abelian, and that
\begin{align}
    \pi:X\longrightarrow \FC^d/\sGamma \label{cre}
\end{align}
is a crepant resolution. Under certain technical conditions which are spelled out in~\cite[Theorem~5.2]{Constellation}, the crepant resolution $X$ can be obtained as a fine moduli space of stable $\sGamma$-equivariant coherent sheaves $\CE$ on $\FC^d$ such that $H^0(\CE)\simeq R$ as a $\sGamma$-module, which are called \emph{$\sGamma$-constellations}; it can be parametrized by $\sGamma$-equivariant modules over the coordinate ring $\FC[z_1,\dots,z_d]$ which are isomorphic to $R$ as $\sGamma$-modules. This includes the cases where $X$ is the $\sGamma$-Hilbert scheme: the structure sheaf $\CO_Z$ of a $\sGamma$-cluster $Z$ in $\FC^d$ is a $\sGamma$-constellation. Physically $X$ can be interpreted as an orbifold resolution by regular D-branes~\cite{Douglas:1997de}, with Fayet--Iliopoulos parameters controlling the size of the blow-ups.

Specialising to $d=4$, if $\FC^4/\sGamma$ admits a suitable crepant resolution $X$, then $X$ can be realised as a moduli space of $\sGamma$-constellations, which in turn should be isomorphic to a moduli space of noncommutative instantons in the Coulomb branch. This is the starting point for a generalization of the $d=3$ construction of~\cite{Quiver3d}, which is based on the $\sGamma$-Hilbert scheme and uses Beilinson's theorem along with some homological algebra to interpret orbifold noncommutative instantons geometrically in terms of equivariant sheaves. This would lend a geometrical parametrization of the instanton moduli space \smash{$\frM_{r,k}^\sGamma$} in terms of $\sGamma$-equivariant coherent sheaves on $\FC^4$. Such a construction is briefly sketched in~\cite{Bon} for a special case where $X$ is the $\sGamma$-Hilbert scheme by directly applying the techniques of~\cite{Quiver3d}. We do not pursue these technical considerations in this paper, and instead proceed to look at these blow-ups in more detail for the types of theories we are interested in;  the orbifold example of~\cite{Bon} is contained in our general formalism as a special case.

\subsubsection*{Toric Orbifold Resolutions}
 \label{Orb_res}
 
Let $\sGamma$ be a finite abelian subgroup of $\sSL(d,\FC)$ acting linearly on  $\FC^d$, which we require to be a subgroup of the maximal torus of the holonomy group $\sSU(d)$. The space $\FC^d$ is a toric Calabi--Yau manifold and the coarse moduli space of the toric Calabi--Yau orbifold $\FC^d/\sGamma$  has quotient singularities.  These can be resolved through a blow-up process \cite{crepant2} which defines a crepant resolution. To discuss the
blow-up we will use the language of toric geometry \cite{TV1,TV2}.

We start by recalling the relevant notions from toric geometry. Let $N\simeq\RZ^d$ be a $\RZ$-lattice and let $\{v_1,\dots, v_p\}$ be a set of linearly independent vectors in the vector space $N_{\FR}=N\otimes_\RZ \FR\simeq\FR^d$. The set
\begin{align}
\sigma:=\FR_{\geq0}\,v_1+\cdots +\FR_{\geq0}\,v_p \ \subset \ N_\FR 
\end{align}
is called a polyhedral cone if $\sigma\cap(-\sigma)=\{0\}$. We say that $\sigma$ is rational if $v_i\in N$, and that the cone is smooth if its generators $v_1,\dots, v_p$ form part of an integral basis of $N$.
In the following we use the notation $\langle v_1,\dots,v_p\rangle$ to indicate the rational cone generated by the vectors $v_1,\dots,v_p$. For a lattice point $ \vec{\alpha}=(\alpha_1,\dots \alpha_d)$ of $N\cap\sigma,$ the age of $\vec{\alpha}$ is defined by \smash{${\rm age}(\vec{\alpha})=\sum_{a=1}^d\,\alpha_a$}.

A fan $\Delta\subseteq N_{\FR}$ is a finite collection of polyhedral cones in $N_{\FR}$ such that every face of any cone $\sigma\in\Delta$ is also a cone of $\Delta$, and the intersection of any two cones of $\Delta$ is also a cone of $\Delta$; it is smooth if all of its cones are smooth. For our purposes it suffices to say that a toric variety $X_\Delta$ is a $d$-dimensional complex variety defined combinatorially by a fan $\Delta$, which has an open covering by invariant subsets $X_\sigma\subset X_\Delta$ with $\sigma\in\Delta$ for the natural action of the algebraic torus $N_{\FC^\times}=N\otimes_\RZ\FC^\times\simeq(\FC^\times)^d$; it is smooth if $\Delta$ is smooth.

Our main example is the quotient variety $\FC^d/\sGamma := {\rm Spec}\big(\FC[z_1,\dots,z_d]^\sGamma\big)$, the spectrum of the $\sGamma$-invariant subring of $\FC[z_1,\dots,z_d]$, which can be represented by a fan containing a single rational polyhedral cone $\sigma$. Indeed, if $\sGamma$ is a finite abelian group of order $n$ whose generators act on the coordinates of $\FC^d$ as
\begin{align}
(z_1,\dots,z_d)\longmapsto (\e^{\,2\pi \,\ii\,s_1/n}\, z_1,\dots,\e^{\,2\pi \,\ii\,s_d/n}\, z_d) \ ,
\end{align} 
then the local coordinates of $X_\sigma$ are $U^a = (z_1)^{(v_1)_a}\cdots (z_d)^{(v_d)_a}$ for $a=1,\dots,d$, where $v_1,\dots, v_d$ are the generators of the cone $\sigma$~\cite{4d-fold}. We require the coordinates $U^a$ to be invariant
under the action of $\sGamma$, thus the cone is given by the solution of the system of equations
\begin{align}
s_1\,(v_1)_a+\cdots+ s_d\,(v_d)_a= 0\quad \mbox{mod} \ n \ , \label{vconditions}
\end{align}
for $a=1,\dots,d$.

Let $\Sigma$ be a subdivision of $\sigma$ through lattice points $\vec\alpha$ with ${\rm age}(\vec{\alpha})=1$, which defines a hyperplane $\Pi$ in $N_\FR$. If $\Sigma$ is smooth, then the
toric variety $X_\Sigma$ determined by $\Sigma$ is a crepant resolution of~$\FC^d/\sGamma$; it is Calabi--Yau because its canonical bundle $K_{X_\Sigma}$ is trivial, and so is acted on by a $d{-}1$-dimensional Calabi--Yau subtorus in $N_{\FC^\times}$. The McKay correspondence associates the prime exceptional divisors of $X_\Sigma$ to the elements of $\sGamma$ with age $\frac1n\,(s_1+\cdots+s_d)=1$~\cite{Constellation}.

In what follows we will consider two particular classes of orbifolds of $\FC^4$ explicitly for which crepant resolutions are always guaranteed. They both involve non-generic subgroups $\sGamma\subset\sSL(4,\FC)$ which also preserve a smaller holonomy subgroup of $\sSU(4)$.

\subsubsection*{$\mbf{(2,0)}$ Orbifolds $\mbf{\FC^2/\RZ_n\times \FC^2}$}

Consider the quotient singularity $\FC^2/\RZ_n\times \FC^2$, where $\sGamma=\RZ_n$ is the cyclic group of order $n$ in $\sSL(2,\FC)$ whose generator acts on $\FC^4$ as
\begin{align}\label{gact}
(z_1,z_2,z_3,z_4)\longmapsto(\omega\, z_1,\omega^{-1}\, z_2,z_3,z_4) \ ,
\end{align}
with $\omega=\e^{\,2\pi\,\ii/n}$ a primitive $n$-th root of unity.
We call these \emph{(2,0) orbifolds}. They have $\sSU(2)$ holonomy, leading to an enhanced $\CN=(0,4)$ supersymmetry in the two-dimensional D1-brane theory, which can be obtained by dimensional reduction of the six-dimensional $\CN=2$ worldvolume theories of D5-branes at ALE singularities $\FC^2/\RZ_n$. 

The conditions \eqref{vconditions} leave an arbitrariness in choosing the generators $v_3$ and $v_4$. Instead, the generators $v_1$ and $v_2$ define a plane which, with an appropriate choice of an orthonormal basis $\{e_a\}_{a\in\ulfour}$ of $\FR^4$, is spanned by $e_1$ and $e_2$. Consider now the rational cone $\sigma_{12}=\langle v_1,v_2\rangle$. Since $\FC^2/\RZ_n$ admits a crepant resolution (by $\Hilb^{\RZ_n}(\FC^2)$), there is a subdivision $\Sigma_{12}$ of $\sigma_{12}$ which is smooth. By arbitrariness of $v_3$ and $v_4$ we can pick them such that they define a smooth rational cone; for instance, $v_3=e_3$ and $v_4=e_4$ with respect to the basis $\{e_a\}_{a\in\ulfour}$.

The fan 
\begin{align}
\Sigma=\big\{\langle\beta_{12},v_3,v_4\rangle \ \big| \  \beta_{12} \in \Sigma_{12}\big\}
\end{align} 
is a smooth subdivision of $\sigma=\langle v_1,v_2,v_3,v_4\rangle$.
Therefore the orbifold $\FC^4/\RZ_n\simeq \FC^2/\RZ_n\times \FC^2$, with the action of $\RZ_n\subset  \sSL(2,\FC)\subset  \sSL(4,\FC)$ given by \eqref{gact}, admits a crepant resolution $X_\Sigma$, which is a (trivial) fibration of $A_{n-1}$ ALE spaces over the affine plane, or equivalently the total space of the rank three bundle $\CO_{\PP^1}(-n+1)\oplus\CO_{\PP^1}\oplus\CO_{\PP^1}$.

\begin{example}
Take the orbifold $\FC^2/\RZ_2\times \FC^2$ considered in~\cite{Bon}.
By the conditions \eqref{vconditions} we can choose the generators of the rational cone $\sigma$ to be
\begin{align}
v_1=(2,-1,0,0) \quad ,\quad
v_2=(0,1,0,0)\quad,\quad v_3=(0,0,1,0) \quad,\quad v_4=(0,0,0,1) \ .
\end{align}
Its subdivision $\Sigma$ using the hyperplane $\Pi$ is 
\begin{align}
\begin{split}
\Sigma=\big\{\langle v_1,v_3,v_4,v_5\rangle\,,\,\langle v_2,v_3,v_4,v_5\rangle \,,\,
 \langle v_1,v_3,v_4\rangle \,,\, \langle v_2,v_3,v_4\rangle \,,\, \langle v_3,v_4,v_5\rangle \,,\, \langle v_3,v_4\rangle
\big\} \ ,
\end{split}
\end{align}
where
\begin{align}
v_5=(1,0,0,0) \ .
\end{align}
Each tuple of vectors $v_i$ forms part of a basis of the lattice $\RZ^4$, so $\Sigma$ is smooth and $X_\Sigma$ is a crepant resolution of $\FC^2/\RZ_2\times \FC^2$.
\end{example}

\subsubsection*{$\mbf{(3,0)}$ Orbifolds $\mbf{\FC^3/\sGamma\times \FC}$}

Another tractable quotient singularity we will consider explicitly is $\FC^4/\sGamma\simeq \FC^3/\sGamma\times \FC$ where $\sGamma$ is a finite abelian subgroup of $\sSL(3,\FC)\subset  \sSL(4,\FC)$. We call these \emph{(3,0) orbifolds}. They have $\sSU(3)$ holonomy and again lead to an enhancement of the supersymmetry in the two-dimensional D1-brane theory, which can be obtained from dimensional reduction of the four-dimensional $\CN=2$ worldvolume theories of D3-branes at toric Calabi--Yau threefold singularities $\FC^3/\sGamma$.

The construction for $(2,0)$ orbifolds above can be generalized directly here, with obvious modifications, starting from the key fact that $\FC^3/\sGamma$ always admits a crepant resolution for any finite abelian subgroup $\sGamma$ of $\sSL(3,\FC$) (for example by $\Hilb^\sGamma(\FC^3)$). Note that any $(2,0)$ orbifold is also a $(3,0)$ orbifold.

\subsection{Noncommutative Crepant Resolutions and their ADHM Representations} \label{Quivers}

As in the example studied in \cite{Bon}, all information about the moduli space of orbifold instantons can be encoded in the data of a quiver $\ttQ$ which generalizes the McKay quiver for surface singularities and is analogous to the four-loop quiver ${\sf L}_4$ of \eqref{eq:ADHMquiverC4}. It is determined by the representation theory data of the $\sGamma$-action, with ${\sf L}_4$ corresponding to the trivial group $\sGamma=\{\ident\}$, and the generalized McKay correspondence asserts that it captures the geometry of the orbifold $\FC^4/\sGamma$. The quiver $\ttQ$ encodes the isotopical decomposition of the usual ADHM data according to the $\sGamma$-action and it has the orbifold ADHM equations \eqref{eq:orbADHM} as relations. For background, see~\cite{Savage06} for a concise overview of the theory of quivers.

One starts from the irreducible representations $\widehat{\sGamma}$ of $\sGamma\subset  \sSL(4,\FC)$. To each representation $\rho_s$ in $\widehat{\sGamma}$, including the trivial representation $\rho_0$ of weight zero, we associate a node of a quiver $\ttQ$. A node $s$ is connected to a node $s'$ by a number of arrows $a_{ss'}$ from $s$ to $s'$ determined by the adjacency  matrix $A=(a_{ss'})$ in the decomposition of $\sGamma$-modules 
\begin{align}
    Q\otimes \rho_s=\bigoplus_{s'\in\widehat{\sGamma}} \, a_{ss'} \, \rho_{s'} \ ,
\end{align}
where the multiplicities are given by
\begin{align}
    a_{ss'}=\dim_{\FC} \sHom_{\FC[\sGamma]}(\rho_{s'},Q\otimes\rho_s) = \frac1{|\sGamma|} \, \sum_{g\in\sGamma} \, \chi_Q(g) \, \chi_s(g) \, \overline{\chi_{s'}(g)} \ .
\end{align}

In this graphical representation the sum over the irreducible representations $\widehat{\sGamma}$ of $\sGamma$ in \eqref{eq:ADHMdecompB} becomes a sum over the nodes $\ttQ_0$ of the quiver and the matrices $B_a^{s}$ are associated to the arrows $\ttQ_1$. The resulting quiver is known as the (generalized) bounded McKay quiver $(\ttQ,\ttR)$ and it is associated with a two-sided ideal of relations $\langle \ttR\rangle$ in
the corresponding path algebra $\FC\,\ttQ$ generated by the arrows $\ttQ_1$; recall that the product in $\FC\,\ttQ$ is the concatenation of paths in $\ttQ$ whenever this makes sense and $0$ otherwise. The remaining data in \eqref{eq:ADHMdecompI} define a framing of the quiver. 

Representations of the McKay quiver $\ttQ=(\ttQ_0,\ttQ_1)$ with relations \eqref{eq:orbADHM} are functorially equivalent to finitely-generated left modules over its path algebra $\ttA := \FC\,\ttQ/\langle\ttR\rangle$. This can be identified as the standard \emph{noncommutative} crepant resolution of the abelian quotient singularity $\FC^4/\sGamma$, in the sense that $\ttA$ is Morita equivalent to the skew group ring $\FC[z_1 , z_2 , z_3,z_4 ] \rtimes \FC[\sGamma]$ of the quotient stack $[\FC^4 /\sGamma]$, whereas the centre $\ttZ(\ttA)$ of $\ttA$ is isomorphic to the $\sGamma$-invariant subring $\FC[z_1 , z_2 , z_3,z_4 ]^\sGamma$. Thus while $\FC^4/\sGamma$ may not have a geometric crepant resolution, it always possesses a noncommutative crepant resolution by $\ttA$. This can be seen as a consequence of the fact that the derived category of coherent sheaves of $\ttA$-modules on $\FC^4/\sGamma$ is a \emph{categorical} crepant resolution of $\FC^4/\sGamma$~\cite{Abuaf16}, in the sense that it mimicks the functorial behaviour one expects from the derived category of a (geometric) crepant resolution; see e.g.~\cite{Wemyss:2012ee,vandenbergh} for introductions. 

The moduli space \smash{$\frM_{r,k}^\sGamma$} of instantons on $\FC^4/\sGamma$, viewed as the moduli space of $\sGamma$-equivariant instantons on $\FC^4$, is then identified as the moduli space of stable framed representations of the bounded quiver $(\ttQ,\ttR)$. The connected component \smash{$\frM_{\vec r,\vec k}$} is the moduli space of stable framed quiver representations with fixed dimension vector~\smash{$(\vec r,\vec k\,)$}, that is, the (stacky) quotient by the action of the group \eqref{eq:sGvec} of the subvariety of the framed quiver representation space
\begin{align}
{\mathfrak{Rep}}_{\vec r,\vec k} = \sHom_{\FC[\sGamma]}(V,V\otimes Q) \, \oplus \, \sHom_{\FC[\sGamma]}(W,V)
\end{align}
cut out by the  ADHM equations \eqref{eq:orbADHM}.
 In particular, regular instantons with $V\simeq R$ the regular representation of $\sGamma$ correspond to orbits in \smash{$\frM_{\vec r,\vec k}$} for $r=1$ and $\vec k = (1,\dots,1)$, which is a moduli space parametrizing isomorphism classes of $\sGamma$-constellations by~\cite{CMT2007,Constellation}. 
 
From this perspective, the (complex) obstruction bundle over the instanton moduli space has the explicit description~\cite{Cao:2014bca} as the pullback along \smash{$\frM_{\vec r,\vec k}\lhook\joinrel\longrightarrow{\mathfrak{Rep}}_{\vec r,\vec k} $} of the bundle
\begin{align}
\Ob_{\vec r,\vec k} = {\mathfrak{Rep}}_{\vec r,\vec k}  \, \times_{\sG_{\vec k}} \, {\sf Ext}_\ttA^2\Big(\textstyle\bigoplus\limits_{s\in\shGamma}\, D_s\otimes V_s\,,\,\bigoplus\limits_{s\in\shGamma}\, D_s\otimes V_s\Big) \ ,
\end{align}
whose fibre computes the relations $\ttR$ of the quiver $\ttQ$. Here $D_s$ is the one-dimensional simple $\ttA$-module defined by $(D_s)_{s'}=\delta_{s,s'}\,\FC$ for $s,s'\in\shGamma$.

Recall from Section~\ref{sec:ADHMC4} that the introduction of a flat $B$-field on $\FC^4$ leads to a representation of D-branes by complexes of modules over the noncommutative algebra ${\sf Mat}_{r{\times}r}(\CA)$. This description is valid when the $B$-field is large compared to the metric on $\FC^4$. For generic orbifold groups $\sGamma\subset\sSL(4,\FC)$, the theory on $\FC^4/\sGamma$ is regarded as a limit where the volume of a subspace containing the orbifold fixed point shrinks to zero but the $B$-field is non-vanishing, which regularises the string worldsheet theory and leads to a noncommutative crepant resolution of the quotient singularity; this is a highly non-geometric limit of the worldvolume gauge theory.
The vacua of the D1-brane theory of Section~\ref{sec:NCorbifold}, and therefore the stable D-brane configurations, then correspond to stable representations of the noncommutative algebra $\ttA$ associated to~$(\ttQ,\ttR)$. In particular, a regular D-brane corresponds to a module over $\ttA$ of dimension vector $\vec k=(1,\dots,1)$ which is the $\ttA$-module of global sections $H^0(\CE)$ of a $\sGamma$-constellation~$\CE$.

\begin{example}\label{ex:orbZ4}
Let $\sGamma=\RZ_4$ acting on $\FC^4$ with generator
\begin{align}
(z_1,z_2,z_3,z_4)\longmapsto(\ii\,z_1,-\ii\,z_2,z_3,z_4) \ . \label{Gaction}
\end{align}
This $\RZ_4$-action defines a $(2,0)$ quotient singularity $\mathbbm{C}^4/\mathbbm{Z}_4\simeq \mathbbm{C}^2/\mathbbm{Z}_4\times \mathbbm{C}^2$. From the general discussion of Section~\ref{Orbcon}, it admits a geometric crepant resolution by a fibration of $A_3$ ALE spaces over~$\FC^2$.

With respect to this $\RZ_4$-action, the fundamental representation $Q\simeq\FC^4$ decomposes into irreducible $\RZ_4$-modules as  $Q=\rho_1\oplus\rho_3\oplus\rho_0\oplus\rho_0$, which implies
\begin{align}
\begin{split}
    Q\otimes\rho_0=\rho_1\oplus\rho_3\oplus\rho_0\oplus\rho_0 \qquad , \qquad Q\otimes\rho_2=\rho_3\oplus\rho_1\oplus\rho_2\oplus\rho_2 \ , \\[4pt]
    Q\otimes\rho_1=\rho_2\oplus\rho_0\oplus\rho_1\oplus\rho_1 \qquad , \qquad Q\otimes\rho_3=\rho_0\oplus\rho_2\oplus\rho_3\oplus\rho_3 \ .
\end{split}
\end{align}
The finite abelian group $\shGamma$ is defined by the character table of $\RZ_4$:
\begin{equation}
{\small
\begin{tabular}{|l|c|c|c|c|l}
\cline{1-5}
                  & \multicolumn{1}{l|}{\textbf{$\rho_0$}} & \multicolumn{1}{l|}{\textbf{$\rho_1$}} & \multicolumn{1}{l|}{\textbf{$\rho_2$}} & \multicolumn{1}{l|}{\textbf{$\rho_3$}} &  \\ \cline{1-5}
\textbf{$\chi_0$} & $1$                                      & $1$                                      & $1$                                      & $1$                                      &  \\ \cline{1-5}
\textbf{$\chi_1$} & $1$                                      & $\ii$                                      & $ -1$                                     & $-\ii$                                     &  \\ \cline{1-5}
\textbf{$\chi_2$} & $1$                                      & $-1$                                     & $1$                                      & $-1$                                     &  \\ \cline{1-5}
\textbf{$\chi_3$} & $1$                                      & $-\ii$                                     & $-1$                                     & $\ii$                                      &  \\ \cline{1-5}
\end{tabular} }
\normalsize
\end{equation}
The generalized McKay quiver $\ttQ$ is
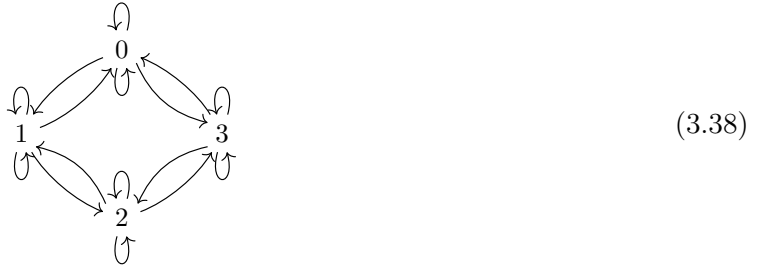
\begin{equation}\label{eq:McKayZ4}
{\small
\begin{tikzcd}
	& 0\arrow[,out=75,in=105,loop,swap,]    
	\arrow[,out=255,in=285,loop,swap,] \\
	1 \arrow[,out=75,in=105,loop,swap,]    
	\arrow[,out=255,in=285,loop,swap,]&& 3\arrow[,out=75,in=105,loop,swap,]    
	\arrow[,out=255,in=285,loop,swap,] \\
	& 2\arrow[,out=75,in=105,loop,swap,]    
	\arrow[,out=255,in=285,loop,swap,]
	\arrow[curve={height=6pt}, from=1-2, to=2-1]
	\arrow[curve={height=6pt}, from=2-1, to=3-2]
	\arrow[shift left=1, curve={height=6pt}, from=3-2, to=2-1]
	\arrow[curve={height=6pt}, from=3-2, to=2-3]
	\arrow[shift left=1, curve={height=6pt}, from=2-3, to=3-2]
	\arrow[curve={height=6pt}, from=2-3, to=1-2]
	\arrow[shift left=1, curve={height=6pt}, from=1-2, to=2-3]
	\arrow[curve={height=6pt}, from=2-1, to=1-2]
\end{tikzcd}}
\normalsize
\end{equation}

The generalized ADHM equations, which determine the relations $\ttR$ of the quiver \eqref{eq:McKayZ4},  assume the form
\begin{align}
    B^{s+s_a}_b\, B_a^s=B^{s+s_b}_a \, B_b^s \qquad \text{with} \quad \begin{cases} \ B_{1}^s: V_s \longrightarrow V_{s+1}\\
    \ B_{2}^s: V_s \longrightarrow V_{s+3}\\
    \ B_{3,4}^s: V_s \longrightarrow V_{s}
    \end{cases} \ .
\end{align}
Explicitly
\begin{equation}
{\small
\begin{tabular}{llll}
$B_2^1\, B_1^0=B_1^3\, B_2^0$ \ ,  & $B_3^1\, B_1^0=B_1^0\, B_3^0$  \ , & $B_4^1\, B_1^0=B_1^0\, B_4^0$  \ ,  & $B_3^0\, B_2^0=B_2^0\, B_3^0$  \ , \\[4pt]
$B_4^3\, B_2^0=B_2^0\, B_4^0$  \ ,  & $B_4^0\, B_3^0=B_3^0\, B_4^0$  \ ,  & $B_2^2\, B_1^1=B_1^0\, B_2^1$  \ ,  & $B_3^2\, B_1^1=B_1^1\, B_3^1$  \ , \\[4pt]
$B_4^2\, B_1^1=B_1^1\, B_4^1$  \ ,  & $B_3^0\, B_2^1=B_2^1\, B_3^1$  \ ,  & $B_4^0\, B_2^1=B_2^1\, B_4^1$ \ ,   & $B_4^1\, B_3^1=B_3^1\, B_4^1$  \ , \\[4pt]
$B_2^3\, B_1^2=B_1^1\, B_2^2$  \ ,  & $B_3^3\, B_1^2=B_1^2\, B_3^2$  \ ,  & $B_4^3\, B_1^2=B_1^2\, B_4^2$  \ ,  & $B_3^1\, B_2^2=B_2^2\, B_3^2$  \ , \\[4pt]
$B_4^1\, B_2^2=B_2^2\, B_4^2$  \ ,  & $B_4^2\, B_3^2=B_3^2\, B_4^2$  \ ,  & $B_2^0\, B_1^3=B_1^2\, B_2^3$  \ ,  & $B_3^0\, B_1^3=B_1^3\, B_3^3$  \ , \\[4pt]
$B_4^0\, B_1^3=B_1^3\, B_4^3$  \ ,  & $B_3^0\, B_2^1=B_2^1\, B_3^1$  \ ,  & $B_4^2\, B_2^3=B_2^3\, B_4^3$  \ ,  & $B_4^3\, B_3^3=B_3^3\, B_4^3 \ \  . $
\end{tabular} }
\normalsize
\end{equation}

The center $\ttZ(\ttA)$ of the path algebra $\ttA$ of $(\ttQ,\ttR)$  is generated as a ring by elements
\begin{align}
\begin{split}
    \ttX_{11}=B^3_1\, B_1^2\,B_1^1\,B_1^0 \quad , \quad & \ttX_{22}=B^1_2\,B^2_2\,B^3_2\,B^0_2 \quad,\quad
     \ttX_{12}=B^3_1\,B^0_2 \ , \\[4pt]
    \ttY_{3}&=B^0_3 \quad,\quad \ttY_{4}=B^0_4 \ .
\end{split}
\end{align}
Then given the $\RZ_4$-action \eqref{Gaction}, we can identify these generators with the $\RZ_4$-invariant elements in the ring  $\FC[z_1,z_2,z_3,z_4]$ through
\begin{align}
\begin{split}
\ttX_{11}\leadsto z_1^4\quad,\quad &\ttX_{22}\leadsto z_2^4\quad, \quad \ttX_{12}\leadsto z_1\,z_2 \ ,\\[4pt]
&\ttY_{3}\leadsto z_3\quad, \quad \ttY_{4}\leadsto z_4 \ .
\end{split}
\end{align}
It follows that ${\rm Spec}(\ttZ(\ttA))\simeq\FC^2/\RZ_4\times \FC^2$ and  the path algebra $\ttA$ is a noncommutative crepant resolution of the quotient singularity $\FC^4/\RZ_4$.

Finally, the generalized ADHM data $(B_a,I)_{a\in\ulfour}$ define a framed representation of the quiver \eqref{eq:McKayZ4}, which we depict by the oriented graph
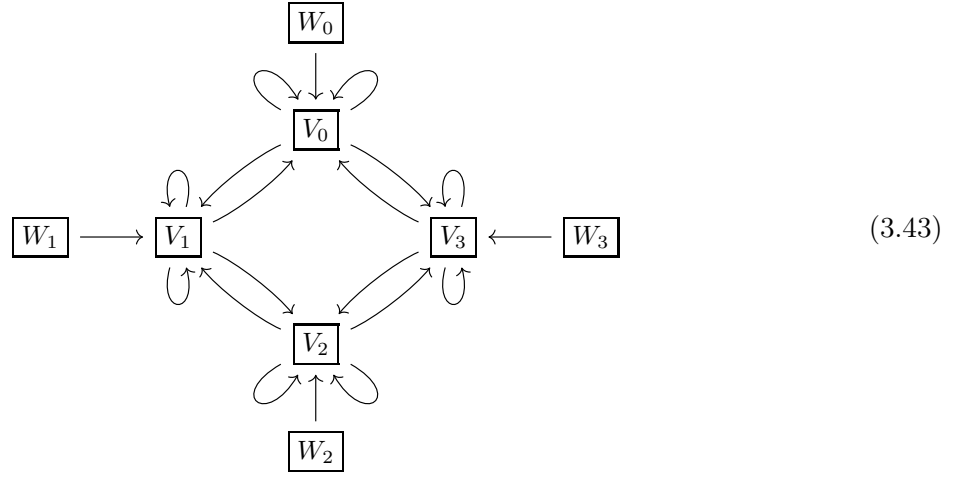
\begin{equation}
{\small
\begin{tikzcd}
	&& {\boxed{W_0}} \\
	&& {\boxed{V_0}}\arrow[,out=30,in=60,loop,swap,]    
	\arrow[,out=150,in=120,loop,swap,] \\
	{\boxed{W_1}} & {\boxed{V_1}}\arrow[,out=75,in=105,loop,swap,]    
	\arrow[,out=255,in=285,loop,swap,] && {\boxed{V_3}}\arrow[,out=75,in=105,loop,swap,]    
	\arrow[,out=255,in=285,loop,swap,] & {\boxed{W_3}} \\
	&& {\boxed{V_2}}\arrow[,out=-30,in=-60,loop,swap,]    
	\arrow[,out=-150,in=-120,loop,swap,] \\
	&& {\boxed{W_2}}
	\arrow[curve={height=6pt}, from=2-3, to=3-2]
	\arrow[curve={height=6pt}, from=3-2, to=2-3]
	\arrow[curve={height=6pt}, from=4-3, to=3-4]
	\arrow[curve={height=-6pt}, from=3-2, to=4-3]
	\arrow[curve={height=-6pt}, from=4-3, to=3-2]
	\arrow[curve={height=6pt}, from=3-4, to=4-3]
	\arrow[curve={height=-6pt}, from=3-4, to=2-3]
	\arrow[curve={height=-6pt}, from=2-3, to=3-4]
	\arrow[from=3-1, to=3-2]
	\arrow[from=1-3, to=2-3]
	\arrow[from=5-3, to=4-3]
	\arrow[from=3-5, to=3-4]
\end{tikzcd}
}
\normalsize
\end{equation}
\end{example}

\begin{example}\label{ex_orbZ2Z2}
 Let $\sGamma=\RZ_2\times\RZ_2$ be the subgroup of order four in $\sSL(4,\FC)$ acting non-trivially on all four coordinates of $\mathbbm{C}^4$ with weights $s_1=(1,1,0,0)$, $s_2=(0,0,1,1)$ and $s_3= s_1+ s_2=(1,1,1,1)$. Although this does not define a $(3,0)$ orbifold, the existence of a geometric crepant resolution of $\FC^4/\RZ_2\times\RZ_2$ is guaranteed by \cite[Proposition~3.1]{TV2}, which may be taken to be $\Hilb^{\RZ_2\times\RZ_2}(\FC^4)$ by  \cite[Proposition~3.2]{TV2}. 
 
 The generators of $\RZ_2\times\RZ_2$ with respect to this action are
\begin{align}
    g_1={\small \begin{pmatrix}
    -1&&&\\
    &-1&&\\
    &&1&\\
    &&&1
    \end{pmatrix}} \normalsize \qquad \mbox{and} \qquad  g_2={\small \begin{pmatrix}
    1&&&\\
    &1&&\\
    &&-1&\\
    &&&-1
    \end{pmatrix} } \normalsize \ .
\end{align}
The group action has four irreducible representations $\rho_s$ where $\rho_0$ is the trivial representation, $\rho_1$ and $\rho_2$ have   weights $s_1$ and $s_2$, respectively, while $\rho_3=\rho_1\otimes\rho_2$ has   weight $s_3$.
The character table of $\RZ_2\times\RZ_2$ is
\begin{equation}
{\small
\begin{tabular}{|l|c|c|c|c|l}
\cline{1-5}
                  & \multicolumn{1}{l|}{\textbf{$\rho_0$}} & \multicolumn{1}{l|}{\textbf{$\rho_1$}} & \multicolumn{1}{l|}{\textbf{$\rho_2$}} & \multicolumn{1}{l|}{\textbf{$\rho_3$}} &  \\ \cline{1-5}
\textbf{$\chi_0$} & $1$                                      & $1$                                      & $1$                                      & $1$                                      &  \\ \cline{1-5}
\textbf{$\chi_1$} & $1$                                      & $1$                                      & $ -1$                                     & $-1$                                     &  \\ \cline{1-5}
\textbf{$\chi_2$} & $1$                                      & $-1$                                     & $1$                                      & $-1$                                     &  \\ \cline{1-5}
\textbf{$\chi_3$} & $1$                                      & $-1$                                     & $-1$                                     & $1$                                      &  \\ \cline{1-5}
\end{tabular} }
\normalsize
\end{equation}

The ADHM equations are
\begin{equation}\label{eq:Z2Z2ADHM}
{\small
\begin{tabular}{llll}
$B_2^1\, B_1^0=B_1^1\, B_2^0$ \ , & $B_3^1\, B_1^0=B_1^2\, B_3^0$ \ , & $B_4^1\, B_1^0=B_1^2\, B_4^0$  \ ,  & $B_3^1\, B_2^0=B_2^2\, B_3^0$  \ ,  \\[4pt]
$B_4^1\, B_2^0=B_2^2\, B_4^0$  \ ,  & $B_4^2\, B_3^0=B_3^2\, B_4^0$  \ ,  & $B_2^0\, B_1^1=B_1^0\, B_2^1$  \ ,  & $B_3^0\, B_1^1=B_1^3\, B_3^1$  \ ,  \\[4pt]
$B_4^0\, B_1^1=B_1^3\, B_4^1$  \ ,  & $B_3^0\, B_2^1=B_2^3\, B_3^1$  \ ,  & $B_4^0\, B_2^1=B_2^3\, B_4^1$  \ ,  & $B_4^3\, B_3^1=B_3^3\, B_4^1$ \ ,   \\[4pt]
$B_2^3\, B_1^1=B_1^3\, B_2^2$  \ ,  & $B_3^3\, B_1^2=B_1^0\, B_3^2$  \ ,  & $B_4^3\, B_1^2=B_1^0\, B_4^2$  \ ,  & $B_3^3\, B_2^2=B_2^0\, B_3^2$  \ ,  \\[4pt]
$B_4^3\, B_2^2=B_2^0\, B_4^2$  \ ,  & $B_4^0\, B_3^2=B_3^0\, B_4^2$  \ ,  & $B_2^2\, B_1^3=B_1^2\, B_2^3$  \ ,  & $B_3^2\, B_1^3=B_1^1\, B_3^3$  \ ,   \\[4pt]
$B_4^2\, B_1^3=B_1^1\, B_4^3$  \ ,  & $B_3^2\, B_2^3=B_2^1\, B_3^3$  \ ,  & $B_4^2\, B_2^3=B_2^1\, B_4^3$  \ ,  & $B_4^1\, B_3^3=B_3^1\, B_4^3$ \ .
\end{tabular} }
\normalsize
\end{equation}
These are relations $\ttR$ for the corresponding McKay quiver $\ttQ$ given by
\begin{equation}\label{eq:Z2Z2quiver}
{\small
\begin{tikzcd}
	& 0 \\
	1 && 3 \\
	& 2
	\arrow[curve={height=12pt}, from=1-2, to=2-1]
	\arrow[curve={height=12pt}, from=2-1, to=3-2]
	\arrow[shift left=1, curve={height=12pt}, from=3-2, to=2-1]
	\arrow[curve={height=12pt}, from=3-2, to=2-3]
	\arrow[shift left=1, curve={height=12pt}, from=2-3, to=3-2]
	\arrow[curve={height=12pt}, from=2-3, to=1-2]
	\arrow[shift left=1, curve={height=12pt}, from=1-2, to=2-3]
	\arrow[curve={height=12pt}, from=2-1, to=1-2]
	\arrow[curve={height=6pt}, from=1-2, to=2-1]
	\arrow[curve={height=6pt}, from=2-1, to=1-2]
	\arrow[curve={height=6pt}, from=2-1, to=3-2]
	\arrow[shift left=1, curve={height=6pt}, from=3-2, to=2-1]
	\arrow[curve={height=6pt}, from=3-2, to=2-3]
	\arrow[shift left=1, curve={height=6pt}, from=2-3, to=3-2]
	\arrow[shift left=1, curve={height=6pt}, from=1-2, to=2-3]
	\arrow[curve={height=6pt}, from=2-3, to=1-2]
\end{tikzcd} }
\end{equation}

The centre $\ttZ(\ttA)$ of the path algebra $\ttA$ of the quiver \eqref{eq:Z2Z2quiver} with relations \eqref{eq:Z2Z2ADHM} is generated as a ring by the elements
\begin{align}
\begin{split}
    \ttX_{abcd}&=B^2_a\, B_b^3\, B_c^1\, B_d^0 \qquad \mbox{with} \quad b\leq d<3\leq a\leq c\leq 4 \ , \\[4pt]
    \ttX_{ab}&=B^2_a\, B^0_b \qquad \mbox{with} \quad b\leq 2<a\leq 4 \ ,
\end{split}
\end{align}
which are again identified  with the $\RZ_2\times\RZ_2$-invariant elements in $\FC[z_1,z_2,z_3,z_4]$. It follows that ${\rm Spec}(\ttZ(\ttA))\simeq\FC^4/\RZ_2\times\RZ_2$.

Similarly to Example \ref{ex:orbZ4}, the ADHM  data define a framed quiver representation
\begin{equation}
{\small
\begin{tikzcd}
	&& {\boxed{W_0}} \\
	&& {\boxed{V_0}} \\
	{\boxed{W_1}} & {\boxed{V_1}} && {\boxed{V_2}} & {\boxed{W_2}} \\
	&& {\boxed{V_3}} \\
	&& {\boxed{W_3}}
	\arrow[shift right=1, curve={height=6pt}, from=2-3, to=3-2]
	\arrow[shift right=1, curve={height=6pt}, from=3-2, to=4-3]
	\arrow[shift right=1, curve={height=6pt}, from=4-3, to=3-4]
	\arrow[shift right=1, curve={height=6pt}, from=3-4, to=2-3]
	\arrow[shift right=2, curve={height=6pt}, from=2-3, to=3-2]
	\arrow[shift right=2, curve={height=6pt}, from=3-2, to=4-3]
	\arrow[shift right=2, curve={height=6pt}, from=4-3, to=3-4]
	\arrow[shift right=2, curve={height=6pt}, from=3-4, to=2-3]
	\arrow[curve={height=6pt}, from=3-2, to=2-3]
	\arrow[curve={height=6pt}, from=2-3, to=3-4]
	\arrow[curve={height=6pt}, from=3-4, to=4-3]
	\arrow[curve={height=6pt}, from=4-3, to=3-2]
	\arrow[shift right=1, curve={height=6pt}, from=3-2, to=2-3]
	\arrow[shift right=1, curve={height=6pt}, from=2-3, to=3-4]
	\arrow[shift right=1, curve={height=6pt}, from=3-4, to=4-3]
	\arrow[shift right=1, curve={height=6pt}, from=4-3, to=3-2]
	\arrow[from=1-3, to=2-3]
	\arrow[from=3-5, to=3-4]
	\arrow[from=5-3, to=4-3]
	\arrow[from=3-1, to=3-2]
\end{tikzcd} }
\normalsize
\end{equation}
of the bounded McKay quiver (\ref{eq:Z2Z2quiver},\ref{eq:Z2Z2ADHM}).
\end{example}

\subsection{Cohomological Field Theory on Noncommutative Crepant Resolutions}\label{Orb_coho}

We shall now modify the construction of Section~\ref{BRST} to calculate the instanton partition function of the cohomological field theory on orbifolds $\FC^4/\sGamma$, which we regard as a gauge theory on the quotient stack $[\FC^4/\sGamma]$ following~\cite{Quiver3d}. We begin with a framed representation of the generalized McKay quiver associated with the abelian quotient singularity, which is uniquely determined by the decomposition of the fundamental representation $Q=\rho_{s_1}\oplus\cdots\oplus\rho_{s_4}$ of the orbifold group into irreducible $\sGamma$-modules. The topological field theory on $[\FC^4/\sGamma]$ is invariant under a set of $\sGamma$-equivariant BRST transformations and it localises onto the relations ${\ttR}$ of the generalized McKay quiver $\ttQ$. The resulting orbifold instanton partition functions encode noncommutative Donaldson--Thomas invariants associated to $(\ttQ,\ttR)$, which count semistable representations of the noncommutative path algebra $\ttA=\FC\,\ttQ/\langle\ttR\rangle$. 

If the quotient singularity $\FC^4/\sGamma$ admits a geometric crepant resolution $\pi:X\longrightarrow\FC^4/\sGamma$, it is natural to conjecture that the noncommutative crepant resolution $\ttA$ of $\FC^4/\sGamma$ is realised in the stringy K\"ahler moduli space of $X$. In particular, the orbifold instanton partition function should be related to the large radius partition functions which compute the Donaldson--Thomas and Pandharipande--Thomas invariants of $X$ via a version of the McKay correspondence, through suitable changes of variables and wall-crossing formulas in the derived category of $X$ using the methods of~\cite{Quiver3d}. As they presently stand, our techniques do not immediately extend to these resolutions and so we defer further discussion to future work. A conjectural mathematical treatment of the rank one partition functions in this setting is discussed in~\cite{CKMpreprint}.

We begin from the BRST transformations with $\sGamma$-module structure, where now the matrix fields decompose according to the irreducible representations of $\sGamma$. The action of the symmetry group \smash{$\big(\timesbig_{s\in\shGamma}\,\sU(k_s)\times\sU(r_s)_{\rm col}\big) \times\sSU(4)$} is described in Section~\ref{sec:NCorbifold}, and the BRST transformations of the ADHM variables are given by
 \begin{equation}\begin{split}
     \CQ_\sGamma B_a^s=\psi_a^s \qquad &, \qquad \CQ_\sGamma\psi_a^s=\phi^{s+s_a}\,B_a^s - B_a^s\,\phi^s-\epsilon_a\, B_a^s \ , \\[4pt]
     \CQ_\sGamma I^s=\varrho^s \qquad &, \qquad \CQ_\sGamma\varrho^s=\phi^s\, I^s-I^s \, \mbf a^s \ ,
     \end{split}
\end{equation}
for $a\in\ulfour$ and $s\in\shGamma$. Here $\phi^s$ parametrizes $\sU(k_s)$ gauge transformations, while the vector $\mbf a^s$ collects the $r_s$ Higgs field eigenvalues $a_l$ associated with the irreducible
representation $\rho_s$. The latter defines a map $l\longmapsto s(l)\in\shGamma$ for $l\in\{1,\dots,r\}$.

The $\sGamma$-module structure of the antighost fields in $\sEnd_{\FC[\sGamma]}(V)$ are dictated by the generalized ADHM equations \eqref{eq:ADHMOrb1}, where again we choose $\mu_{\alpha\beta}^{\FC s}=0$ for $(\alpha,\beta)\in\ulthree^{\perp}:=\big\{(1,2)\,,\,(1,3)\,,\,(2,3)\big\}$ as the independent complex moment map equations. Their decomposition into equivariant maps is given by
\begin{align}
\chi_{\alpha\beta}^{\FC s}: V_s\longrightarrow V_{s+s_{\alpha\beta}} \qquad \mbox{and} \qquad \chi^{\FR s}: V_s\longrightarrow V_s \ .
\end{align}
Together with the auxiliary fields their BRST transformations are
 \begin{equation}\begin{split}
     \CQ_\sGamma\chi_{\alpha\beta}^{\FC s}=H_{\alpha\beta}^{\FC s} \qquad &, \qquad \CQ_\sGamma H_{\alpha\beta}^{\FC s}=\phi^{s+s_{\alpha\beta}}\,\chi_{\alpha\beta}^{\FC s} - \chi_{\alpha\beta}^{\FC s}\,\phi^s -\epsilon_{\alpha\beta}\,\chi^{\FC s}_{\alpha\beta} \ , \\[4pt]
     \CQ_\sGamma\chi^{\FR s}=H^{\FR s} \qquad &, \qquad \CQ_\sGamma H^{\FR s}=[\phi^{s},\chi^{\FR s}] \ ,
\end{split}
 \end{equation}
 while the gauge multiplet closes the BRST algebra
\begin{align}
\CQ_\sGamma\phi^{s}=0 \ , \quad \CQ_\sGamma\Bar{\phi}^{s}=\eta^s \qquad \mbox{and} \qquad \CQ_\sGamma\eta^s=[\phi^{s},\Bar{\phi}^{s}] \ .
\end{align}

For the fundamental matter fields, we treat the global flavour symmetry identically to the global colour symmetry, which is thereby broken to the subgroup $\timesbig_{s\in\shGamma}\,\sU(r_s)_{\rm fla}$ by the $\sGamma$-action; its maximal torus is 
\begin{align}
\sT_{\vec m}=\Timesbig_{s\in\shGamma}\,\sT_{\vec m^s} \ ,
\end{align}
where $\sT_{\vec m^s}$ is the maximal torus of $\sU(r_s)_{\rm fla}$. From the string theory perspective of Section~\ref{sec:ADHMC4}, this aligns the Chan--Paton factors of D9-branes and $\overline{\rm D9}$-branes on the orbifold such that they still annihilate. Since both the framing of the gauge bundle and the fundamental matter fields transform in the fundamental representation $\mbf r$ of $\sU(r)$, it decomposes into irreducible representations of the orbifold group $\sGamma$ as
\begin{align}
\mbf r = \bigoplus_{s\in\shGamma} \, \mbf r_s \otimes\rho_s^* \ ,
\end{align}
with multiplicity spaces $\mbf r_s=\sHom_{\FC[\sGamma]}(\rho_s,\mbf r)$ and
\begin{align}
r = \sum_{s\in\shGamma} \,  r_s = \sum_{s\in\shGamma} \, \dim_\FC \mbf r_s \ .
\end{align}

We use Schur's lemma to decompose the Fermi multiplet in $\sHom_{\FC[\sGamma]}(\mbf r,V)$ into isotopical components
\begin{align}
\bar I^s:\mbf r_s\longrightarrow V_s \ .
\end{align}
Then the BRST transformations are given by
 \begin{equation}\begin{split}
   \mathcal{Q}_\sGamma\bar{I}^s=\bar{\varrho}^s\qquad \mbox{and} \qquad \mathcal{Q}_\sGamma\bar{\varrho}^s=\phi^s\,\bar{I}^s-\bar{I}^s\,\mbf{m}^s \ , \end{split}
 \end{equation}
where the vector $\mbf m^s$ collects the $ r_s$ masses $m_l$ associated with the irreducible representation $\rho_s$.  

Following the same steps as in Section~\ref{BRST} we can now write down the instanton partition function. We introduce a set of fugacities $\vec\qu=(\qu_s)_{s\in\shGamma}$ for the corresponding fractional instanton sectors \smash{$\vec k=(k_s)_{s\in\shGamma}\in\RZ_{\geq0}^{|\sGamma|}$} and define
\begin{align}
Z_{[\FC^4/\sGamma]}^{\vec r}(\vec\qu;\vec a,\vec\epsilon,\vec m) = \sum_{\vec k\in\RZ_{\geq0}^{|\sGamma|}} \, \vec\qu^{\,\vec k} \ Z_{[\FC^4/\sGamma]}^{\vec r,\vec k}(\vec a,\vec\epsilon,\vec m) \ ,
\end{align}
where 
\begin{align}
\vec\qu^{\,\vec k}:=\prod_{s\in\shGamma}\,\qu_s^{k_s}
\end{align}
and the quiver matrix model is defined by the integral
\begin{align}
\begin{split}\label{eq:Orb_Zin}
Z_{[\FC^4/\sGamma]}^{\vec r,\vec k}(\vec a,\vec\epsilon,\vec m) = \oint_{\varGamma_{\vec r,\vec k}} \ \prod_{s\in\shGamma} \, \frac1{k_s!} \ & \prod_{i=1}^{k_s} \, \frac{\dd\phi_i^s}{2\pi\,\ii} \, \frac{\CP_{ r_s}(\phi_i^s|\vec m^s)}{\CP_{r_s}(\phi_i^s|\vec a^s)} \ \prod_{\stackrel{\scriptstyle i,j=1}{\scriptstyle i\neq j}}^{k_s} \, \big(\phi^s_i-\phi^s_j\big) \\
& \hspace{3cm} \times \ \prod_{i,j=1}^{k_s} \, \frac{\displaystyle \prod_{(\alpha,\beta)\in\ulthree^\perp}\, \big(\phi_i^{s+s_{\alpha\beta}}-\phi_j^s-\epsilon_{\alpha\beta}\big)}{\displaystyle \prod_{a\in\ulfour} \, \big(\phi_i^{s+s_a}-\phi_j^s-\epsilon_a\big)} \ ,
\end{split}
\end{align}
which again we make sense of via a suitable contour integral prescription. 

\subsubsection*{Dimensional Reduction}

From the matrix model representation \eqref{eq:Orb_Zin} we immediately obtain the orbifold version of Proposition~\ref{prop:ZDTgeb}, which reads

\begin{proposition}\label{prop3}
The equivariant instanton partition function on a $(3,0)$ orbifold $\FC^3/\sGamma\times\FC$ is related to the partition function \smash{$Z_{[\FC^3/\sGamma]}^{\vec r}(\vec\qu;\vec a,\epsilon_1,\epsilon_2,\epsilon_3)$} for noncommutative Donaldson--Thomas invariants of type~$\vec r$ for the toric K\"ahler orbifold $\FC^3/\sGamma$ through the mass specialisation
\begin{align}
Z_{[\FC^3/\sGamma]\times\FC}^{\vec r}(\vec\qu;\vec a,\vec\epsilon,m_l^s=a_l^s+\epsilon_4) =  Z_{[\FC^3/\sGamma]}^{\vec r}(\vec\qu\,';\vec a,\epsilon_1,\epsilon_2,\epsilon_3) \ ,
\end{align}
where $\vec\qu\,'=\big((-1)^{r_s+1}\,\qu_s\big)_{s\in\shGamma}$.
\end{proposition}

\begin{proof}
Since any $(3,0)$ orbifold has weight $s_4=0$, and $s_{123}=0$ since $\sGamma\subset \sSL(3,\FC)$, using the Calabi--Yau condition $\epsilon_4=-\epsilon_{123}$ on $\FC^3/\sGamma\times\FC$ one finds that the matrix integral \eqref{eq:Orb_Zin} gives
\begin{align}
\begin{split}\label{Orb_Zin-spec_Orb}
&Z_{[\FC^3/\sGamma]\times\FC}^{\vec r,\vec k}(\vec a,\vec\epsilon, m_l^s=a_l^s-\epsilon_{123}) \\[4pt]
& \quad = \oint_{\varGamma_{\vec r,\vec k}} \ \prod_{s\in\shGamma} \, \frac1{k_s!} \ \prod_{i=1}^{k_s} \, \frac{\dd\phi_i^s}{2\pi\,\ii} \, \frac{\CP_{ r_s}(\phi_i^s+\epsilon_{123}|\vec a^s)}{\CP_{r_s}(\phi_i^s|\vec a^s)} \ \prod_{\stackrel{\scriptstyle i,j=1}{\scriptstyle i\neq j}}^{k_s} \, \big(\phi^s_i-\phi^s_j\big) \\
& \hspace{1.7cm} \times \ \prod_{i,j=1}^{k_s} \, \frac{\displaystyle  \big(\phi_i^{s+s_{12}}-\phi_j^s-\epsilon_{12}\big)\,\big(\phi_i^{s+s_{13}}-\phi_j^s-\epsilon_{13}\big)\,\big(\phi_i^{s+s_{23}}-\phi_j^s-\epsilon_{23}\big)}{\displaystyle \big(\phi_i^s-\phi_j^s-\epsilon_{123}\big)\, \big(\phi_i^{s+s_1}-\phi_j^s-\epsilon_1\big)\,\big(\phi_i^{s+s_2}-\phi_j^s- \epsilon_2\big)\,\big(\phi_i^{s+s_3}-\phi_j^s- \epsilon_3\big)} \ .
\end{split}
\end{align}
Up to an overall sign \smash{$\prod_{s\in\shGamma}\,(-1)^{(r_s+1)\,k_s}$}, this is the same as the straightforward generalization of the matrix integral representation from~\cite[eq.~(5.23)]{coho} for the instanton partition function on $\FC^3/\sGamma$ using the orbifold quiver matrix model from~\cite[Section~5.5]{Quiver3d}, slightly modified to include a generic point of the $\varOmega$-deformation on $\FC^3$ with $\epsilon_{123}\neq0$.
\end{proof}

\begin{remark}
The noncommutative crepant resolution of $\FC^3/\sGamma$ in this dimensional reduction is obtained by erasing a single loop from each node of the McKay quiver for the quotient singularity~$\FC^3/\sGamma\times\FC$.
\end{remark}

\subsubsection*{Fixed Points and Coloured Solid Partitions}

The integrand of \eqref{eq:Orb_Zin} has poles along the hyperplanes
  \begin{align}
\phi^{s+s_a}_i-\phi^s_j-\epsilon_a=0\qquad \mbox{and} \qquad \phi^s_i-a^s_l=0 
  \end{align}
in $\FR^k$, for $a\in\ulfour$ and $s\in\widehat{\sGamma}$. In completely analogy with the matrix model of Section \ref{BRST}, these are the fixed points of the orbifold ADHM data \smash{$(B_a^s,I^s)_{a\in\ulfour,\,  s\in\shGamma}$} under the equivariant action of the symmetry group $\big(\timesbig_{s\in\shGamma}\,\sU(k_s)\times\sU(r_s)_{\rm col}\big) \times\sSU(4)$. They reside on the locus of fixed points of the BRST charge $\CQ_\sGamma$ of the cohomological gauge theory on $[\FC^4/\sGamma]$.

Since the actions of $\sGamma$ and $\sT$ commute, we can argue as in Section~\ref{BRST} that these are parametrized  by arrays of solid partitions $\vec{\sigma}=(\sigma_1,\dots,\sigma_r)$, where the splitting of the ADHM data into irreducible representations of $\sGamma$ induces a $\shGamma$-colouring of the solid partitions according to \eqref{coloring}. A coloured solid partition \smash{$\sigma_l$} is in  one-to-one correspondence with the fixed points
\begin{align}
\phi^s_{(a_l^{s(l)},\vec p\,)}=a^{s(l)}_l+\vec{p} \cdot \vec{\epsilon} \ ,
\end{align} 
for $\vec{p}=(p_1,p_2,p_3,p_4)\in\RZ_{>0}^4$, which carry an irreducible representation $\rho_s$ of the orbifold group $\sGamma$ given by
\begin{align}
\rho_s=\rho_{s(l)}\otimes \rho_{\vec p} \qquad \mbox{with} \quad \rho_{\vec p}:=\rho_{s_1}^{\otimes p_1}\otimes\rho_{s_2}^{\otimes p_2}\otimes\rho_{s_3}^{\otimes p_3}\otimes\rho_{s_4}^{\otimes p_4} \ .
\end{align}
The $\shGamma$-colouring of the array of solid partitions defines the total number of boxes of colour $\rho_s$ in $\vec{\sigma}$ for each $s\in\shGamma$ as the fractional instanton number $|\vec{\sigma}|_s=k_s$; we write this condition as $|\vec{\sigma}|_{\shGamma} = \vec k$.

\subsection{Orbifold Instanton Partition Functions}\label{sec:Orb_def_complex}

The partition function of the cohomological theory on $[\FC^4/\sGamma]$ can be computed explicitly by considering the $\sGamma$-equivariant version of the instanton deformation complex \eqref{complex}, which reads
\begin{align}
    0\longrightarrow \sEnd_{\FC[\sGamma]}(V_{\vec{\sigma}})\xrightarrow{ \ \dd_1^\sGamma \ }\begin{matrix}
        \sHom_{\FC[\sGamma]}(V_{\vec{\sigma}}, V_{\vec{\sigma}} \otimes  Q ) \\[1ex] \oplus\\[1ex] \sHom_{\FC[\sGamma]}(W_{\vec{\sigma}},V_{\vec{\sigma}})
    \end{matrix}\xrightarrow{ \ \dd_2^\sGamma \ } \sHom_{\FC[\sGamma]}(V_{\vec{\sigma}},V_{\vec{\sigma}}\otimes\midwedge_-^{0,2}\,  Q  )\longrightarrow 0 \ , 
\end{align}
where the map $\dd_1^\sGamma$ is an infinitesimal \smash{$\sG_{\vec k}$} gauge transformation, while $\dd_2^\sGamma$ is the linearization of the holomorphic ADHM equations \smash{$B_\alpha^{s+s_\beta}\,B_\beta^{s}=B_\beta^{s+s_\alpha}\,B_\alpha^s$} for $(\alpha,\beta)\in\ulthree^\perp$. 

Since the actions of the groups $\sGamma$ and $\sT$ commute, the character we wish to calculate  is now the $\sGamma$-invariant part of the character of the complex \eqref{complex}, that is
\begin{align}\begin{split}
    \chi^\sGamma_{\vec{\sigma}}:\!&= \sqrt{\ch^\sGamma_\sT}\big(T^{\rm vir}_{\vec{\sigma}}\mathfrak{M}_{\vec r,\vec k}\big)-\ch^\sGamma_\sT\big((\mathscr{V}_{r,k})_{\vec{\sigma}}\otimes\mbf r\big) = \sqrt{\ch^\sGamma_\sT}\big(T^{\rm vir}_{\vec{\sigma}}\mathfrak{M}_{\vec r,\vec k}\big) -\big[\mbf r^*\otimes V_{\vec{\sigma}}\big]^\sGamma \ , \end{split} \label{chi_Gamma}
\end{align}
with
\begin{align}\label{eq:chTvir_Gamma}
\begin{split}
\sqrt{\ch^\sGamma_{\sT}}\big(T_{\vec{\sigma}}^{\rm vir}\frM_{\vec r,\vec k}\big)&=\big[V_{\vec{\sigma}}^*\otimes V_{\vec{\sigma}} \, \big(t^{-1}_1+t^{-1}_2+t^{-1}_3+t^{-1}_4\big)+W_{\vec{\sigma}}^*\otimes V_{\vec{\sigma}}\\
    & \qquad \, -V_{\vec{\sigma}}^*\otimes V_{\vec{\sigma}}\,\big(1+t^{-1}_1\,t^{-1}_2+t^{-1}_1\,t^{-1}_3+t^{-1}_2\,t^{-1}_3\big) \big]^\sGamma \ . \end{split}
\end{align}
Since the dual involution commutes with the $\sGamma$-action, by taking the $\sGamma$-invariant part of \eqref{eq:chTvirfull} it follows that
\begin{align}
\ch^\sGamma_{\sT}\big(T_{\vec{\sigma}}^{\rm vir}\frM_{\vec r,\vec k}\big) = \sqrt{\ch^\sGamma_{\sT}}\big(T_{\vec{\sigma}}^{\rm vir}\frM_{\vec r,\vec k}\big) + \sqrt{\ch^\sGamma_{\sT}}\big(T_{\vec{\sigma}}^{\rm vir}\frM_{\vec r,\vec k}\big)^* \ .
\end{align}

The subgroup inclusion $\sGamma\lhook\joinrel\longrightarrow\sT_{\vec\epsilon}$ defines the irreducible representations of $\sGamma$ associated to the toric generators $t_a$ for $a\in\ulfour $. Consequently the vector spaces $V$ and $W$  at a fixed point $\vec{\sigma}$ decompose into
\begin{align}\begin{split}   
  V_{\vec{\sigma}}=\sum_{l=1}^r\, e_l \ \sum_{\vec p\,\in {\sigma}_l} \, t_1^{p_1}\,t_2^{p_2}\,t_3^{p_3}\,t_4^{p_4}\otimes \rho^*_{\vec p}\otimes\rho^*_{s(l)} \qquad \mbox{and} \qquad W_{\vec{\sigma}}=\sum_{l=1}^r\, e_l\otimes \rho^*_{s(l)} \ ,
   \end{split}\label{decompositionVW}
\end{align}
and the $\sU(r)_{\rm fla}$-module $\mbf r$ into
\begin{align}
\mbf r=\sum_{l=1}^r\, f_l\otimes\rho^*_{s(l)} \ ,
\end{align}
as elements of the representation ring of the group $\sT\times\sGamma$. The index \eqref{chi_Gamma} is calculated by projecting the character \eqref{chi} onto those terms which carry the trivial representation $\rho_0$, giving an element in the representation ring of $\sT$.

Then the orbifold instanton partition function is evaluated by using the top form operation \eqref{eq:top} to get the combinatorial formula
\begin{align}\label{eq:Orb_pf}
Z^{\vec r}_{[\mathbb{C}^4/\sGamma]}(\vec\qu;\vec{a},\vec{\epsilon},\vec{m})=\sum_{\vec k\in\RZ_{\geq0}^{|\sGamma|}}\, \vec\qu^{\,\vec k} \ \sum_{|\vec{\sigma}|_\shGamma=\vec k} \, (-1)^{{\tt O}^\sGamma_{\vec{\sigma}}} \ \widehat{\tt e}\big[- \chi^\sGamma_{\vec{\sigma}}\,\big]  \ ,
\end{align}
where
\begin{align}
\begin{split}
\widehat{\tt e}\big[-\chi_{\vec{\sigma}}^\sGamma\,\big] = \prod_{l=1}^r \ \prod_{\vec p_l\in{\sigma}_l}^{\neq0} \, \frac{\CP_r\circ\delta^\sGamma_0(a_l+\vec p_l\cdot\vec\epsilon\,|\vec m)}{\CP_r\circ\delta^\sGamma_0(a_l+\vec p_l\cdot\vec\epsilon\,|\vec a)} \ \prod_{ l'=1}^r \ \prod_{\vec p^{\,\prime}_{l'}\in{\sigma}_{l'}}^{\neq0} \, \CR\circ\delta^\sGamma_0\big(a_l-a_{l'}+(\vec p_l-\vec p^{\,\prime}_{l'})\cdot\vec\epsilon\,\big|\vec\epsilon\,\big) \ .
\end{split}\label{chiorb}
\end{align}
Here the operation $\delta_0^\sGamma$ acts on a combination of equivariant parameters $x$ as the identity if $x$ is associated to the trivial representation $\rho_0$ and returns $1$ otherwise; for example
\begin{align}
\delta_0^\sGamma\big(a_l-a_{l'}+(\vec p_l-\vec p^{\,\prime}_{l'})\cdot\vec\epsilon\,\big) = \begin{cases} a_l-a_{l'}+(\vec p_l-\vec p^{\,\prime}_{l'})\cdot\vec\epsilon \quad \text{if} \ \ \rho_{\vec p}\otimes\rho_{s(l)}\otimes\rho^*_{\vec p^{\,\prime}}\otimes\rho_{s(l')}^*\simeq\rho_0 \ , \\[4pt]
1 \quad \text{otherwise} \ .
\end{cases}
\end{align}
As previously, we understand \eqref{chiorb} as a residue contribution to the matrix integral \eqref{eq:Orb_Zin}, regarded as a contour integral over $\varGamma_{\vec r,\vec k}\subset\FC^k$ where $k=|\vec k\,|$ is the size of \smash{$\vec k\in\RZ_{\geq0}^{|\sGamma|}$}.

\begin{remark}
Since the character \eqref{eq:chTvir_Gamma} is obtained by projection to the $\sGamma$-invariant part of the character \eqref{eq:chTvir}, we believe that the sign factors $\ttO_{\vec{\sigma}}^\sGamma$ in \eqref{eq:Orb_pf}, which should come from a careful residue calculation of \eqref{eq:Orb_Zin}, do not depend on the $\sGamma$-colourings of the solid partitions and coincide with \eqref{eq:signfactor}. This same assertion is made in~\cite{CKMpreprint}.
\end{remark}

\subsubsection*{$\mbf{\sU(1)}$ Gauge Theories}

Let us consider the rank one case $r=1$. Then there are only single equivariant parameters $a^s$ and $m^s$ for the gauge and flavour symmetry, which are both associated to the same irreducible representation $\rho_s$ of~$\sGamma$.
Since $\rho_s\otimes\rho_s^*\simeq\rho_0$, the partition function \eqref{eq:Orb_pf} does not depend on the choice of framing vector $\vec{r}=(r_s)_{s\in\shGamma}$, where we use the convention that $r_0$ is the first entry of $\vec r$.
In other words, all framing vectors of the form $\vec{r}=(0,\dots,0,1,0,\dots,0)$, with zeroes in all but one entry, give the same partition function. Moreover, the partition function \eqref{eq:Orb_pf} depends only on the combination $m:=m^s-a^s$ and hence is effectively independent of the Coulomb parameter. The rank one partition functions are therefore simply denoted as
\begin{align}
Z_{[\mathbb{C}^4/\sGamma]}(\vec\qu;\vec{\epsilon},m) = \sum_{\sigma} \, (-1)^{{\tt O}^\sGamma_{{\sigma}}} \ \vec\qu^{\,|{\sigma}|_\shGamma} \ \prod_{\vec p\,\in{\sigma}}^{\neq0} \, \frac{\delta^\sGamma_0(\vec p\cdot\vec\epsilon-m)}{\delta^\sGamma_0(\vec p\cdot\vec\epsilon\,)} \ \prod_{\vec p^{\,\prime}\in{\sigma}}^{\neq0} \, \CR\circ\delta^\sGamma_0\big((\vec p-\vec p^{\,\prime})\cdot\vec\epsilon\,\big|\vec\epsilon\,\big) \ .
\end{align}

\begin{example}
Let $\sGamma=\RZ_3$ act on $\FC^4$ with generator
\begin{align}
    (z_1,z_2,z_3,z_4)\longmapsto (\xi\, z_1,\xi\, z_2, \xi\, z_3, z_4) \ ,
\end{align}
where $\xi=\e^{\,2\pi\,\ii /3}$ is a primitive third root of unity. This defines a $(3,0)$ orbifold $\FC^3/\RZ_3\times\FC$, whose natural geometric crepant resolution induced by the $\RZ_3$-Hilbert scheme \smash{$\Hilb^{\RZ_3}(\FC^3)$} is a fibration of local del~Pezzo surfaces of degree zero over the affine line, or equivalently the total space of the rank two bundle $\CO_{\PP^2}(-3)\oplus\CO_{\PP^2}$. Its McKay quiver is
\begin{equation}
{\small
\begin{tikzcd}
	& 0\arrow[,out=75,in=105,loop,swap,] \\
	1\arrow[,out=165,in=195,loop,swap,] && 2\arrow[,out=15,in=-15,loop,swap,]
	\arrow[curve={height=6pt}, from=1-2, to=2-1]
	\arrow[curve={height=-6pt}, from=1-2, to=2-1]
	\arrow[from=1-2, to=2-1]
	\arrow[shift right=1, curve={height=6pt}, from=2-1, to=2-3]
	\arrow[shift right=1, curve={height=-6pt}, from=2-1, to=2-3]
	\arrow[shift right=1, from=2-1, to=2-3]
	\arrow[curve={height=6pt}, from=2-3, to=1-2]
	\arrow[curve={height=-6pt}, from=2-3, to=1-2]
	\arrow[from=2-3, to=1-2]
\end{tikzcd}
} \normalsize
\end{equation}
which is obtained from the Beilinson quiver $\sf B$ by adding a loop at each node.
The fractional instanton contributions to the corresponding $\sU(1)$ partition function are given by
\begin{align}
\begin{split}
& Z_{[\FC^3/\RZ_3]\times\FC}^{\vec k}(\vec\epsilon,m) \\[4pt]
& \quad = \sum_{|{\sigma}|_{\RZ_3}=\vec k} \, (-1)^{{\tt O}^{\RZ_3}_{{\sigma}}} \  \prod_{\substack{\vec p\,\in\sigma\\ p_1+p_2+p_3\,\equiv_3\, 0}}^{\neq0} \, \frac{\vec p\cdot\vec \epsilon-m}{\vec p\cdot\vec \epsilon}  \ \prod_{\substack{\vec p,\vec p\,'\in\sigma\\ p_1+p_2+p_3\,\equiv_3\, p_1'+p_2'+p_3'}}^{\neq0} \, \frac{(\vec p-\vec p\,')\cdot\vec \epsilon}{(\vec p-\vec p\,')\cdot\vec \epsilon-\epsilon_{4}} \\
& \hspace{1cm} \times \prod_{\substack{\vec p,\vec p\,'\in\sigma\\ p_1+p_2+p_3+1\,\equiv_3\, p_1'+p_2'+p_3'}}^{\neq0} \, \frac{\big( (\vec p-\vec p\,')\cdot\vec \epsilon -\epsilon_{12}\big)\,\big( (\vec p-\vec p\,')\cdot\vec \epsilon -\epsilon_{13}\big)\,\big( (\vec p-\vec p\,')\cdot\vec \epsilon-\epsilon_{23}\big)}{\big( (\vec p\,'-\vec p\,)\cdot\vec \epsilon-\epsilon_{1}\big)\,\big( (\vec p\,'-\vec p\,)\cdot\vec \epsilon-\epsilon_{2}\big)\,\big( (\vec p\,'-\vec p\,)\cdot\vec \epsilon-\epsilon_{3}\big)} \ ,
\end{split}
\end{align}
where $\equiv_3$ denotes congruence modulo~$3$.
\end{example}

\subsubsection*{Higher Rank Gauge Theories}

In contrast to the rank one case, the partition function for higher rank $r>1$ depends explicitly on the choice of decomposition $\vec r=(r_s)_{s\in\shGamma}$ of the rank $r$ according to the irreducible representations of $\sGamma$, and different framing $\sGamma$-modules $W\simeq\FC^r$ generally lead to inequivalent theories.

\begin{example}\label{ex:Z2Z2higherrank}
Consider the $\RZ_2\times\RZ_2$-action of Example~\ref{ex_orbZ2Z2}. Three $\vec k=(2,0,0,0)$ contributions to rank two partition functions $Z_{[\FC^4 / \RZ_2\times\RZ_2]}^{\vec r}(\vec \qu;\vec a,\vec{\epsilon},\vec{m})$ are
\begin{align}
\begin{split}
 Z_{[\FC^4 / \RZ_2\times\RZ_2]}^{\vec r=(2,0,0,0),\vec k=(2,0,0,0)}(\vec a,\vec{\epsilon},\vec{m})&=\frac{\epsilon_{12}^2\,m_1\,m_2}{(a_1-a_2)^2}\,(a_1-a_2+m_1)\,(a_1-a_2-m_2) \ ,\\[4pt]
 Z_{[\FC^4 / \RZ_2\times\RZ_2]}^{\vec r=(1,1,0,0),\vec k=(2,0,0,0)}(\vec a,\vec{\epsilon},\vec{m})&=\epsilon_{12}^2\,m_1\,m_2 =
 Z_{[\FC^4 / \RZ_2\times\RZ_2]}^{\vec r=(1,0,1,0),\vec k=(2,0,00)}(\vec a,\vec{\epsilon},\vec{m}) \ .
 \end{split}
\end{align}
\end{example}

However, some higher rank theories are equivalent. Looking at the contribution \eqref{chiorb} from the index $\chi_{\vec{\sigma}}^{\sGamma}$, we see that all framing vectors $\vec r$ of the form $(0,\dots,0, r,0,\dots,0)$ yield equivalent partition functions, irrespecitve of the location of the entry $r$. Indeed, in these cases the equivariant parameters $\vec a$ and $\vec m$ are associated the same irreducible $\sGamma$-representation. Thus the differences $a_l-a_{l'}$ and $a_l-m_{l'}$ are associated to the trivial representation $\rho_0$ and one always counts the same contributions \smash{$ Z^{\vec r,\vec k}_{[\FC^4/\sGamma]}(\vec{a},\vec{\epsilon},\vec{m})$}.

\begin{example}
Consider $\sGamma=\RZ_4$ acting on $\FC^4$ with generator
\begin{align}
(z_1,z_2,z_3,z_4)\longmapsto(\ii\,z_1,\ii\,z_2,\ii\,z_3,\ii\,z_4) \ .
\end{align}
The quotient singularity $\FC^4/\RZ_4$ admits a geometric crepant resolution by~\cite[Claim~2]{Mohri:1997ef} and~\cite[Proposition~3.1]{TV2}. Its McKay quiver is
\begin{equation}
{\small
\begin{tikzcd}
	& 0 \\
	1 && 3 \\
	& 2
	\arrow[shift right=1, curve={height=12pt}, from=1-2, to=2-1]
	\arrow[shift right=1, curve={height=12pt}, from=2-1, to=3-2]
	\arrow[shift right=1, curve={height=-6pt}, from=2-1, to=3-2]
	\arrow[shift right=1, curve={height=12pt}, from=3-2, to=2-3]
	\arrow[shift right=1, curve={height=-6pt}, from=3-2, to=2-3]
	\arrow[shift right=1, curve={height=12pt}, from=2-3, to=1-2]
	\arrow[shift right=1, curve={height=-6pt}, from=2-3, to=1-2]
	\arrow[shift right=1, curve={height=-6pt}, from=1-2, to=2-1]
	\arrow[shift right=1, curve={height=6pt}, from=1-2, to=2-1]
	\arrow[shift right=1, from=1-2, to=2-1]
	\arrow[shift right=1, curve={height=6pt}, from=2-1, to=3-2]
	\arrow[shift right=1, from=2-1, to=3-2]
	\arrow[shift right=1, curve={height=6pt}, from=3-2, to=2-3]
	\arrow[shift right=1, from=3-2, to=2-3]
	\arrow[shift right=1, from=2-3, to=1-2]
	\arrow[shift right=1, curve={height=6pt}, from=2-3, to=1-2]
\end{tikzcd}
}
\normalsize
\end{equation}
and the noncommutative crepant resolution can be associated with the total space of the canonical bundle $\CO_{\PP^3}(-4)$~\cite{Cao:2014bca}.
The leading contributions for $\vec r=(2,0,0,0)$ and $\vec r=(0,2,0,0)$ to the rank two partition functions \smash{$Z_{[\FC^4 / \RZ_4]}^{\vec r}(\vec \qu;\vec a,\vec{\epsilon},\vec{m})$} coincide and are given by
\begin{align}
\begin{split}
Z_{[\FC^4 / \RZ_4]}^{\vec r=(2,0,0,0)}(\vec \qu;\vec a,\vec{\epsilon},\vec{m}) &=Z_{[\FC^4 / \RZ_4]}^{\vec r=(0,2,0,0)}(\vec \qu;\vec a,\vec{\epsilon},\vec{m}) \\[4pt] &= -(m_1+m_2)\,\qu_0\\
& \hspace{1cm} \quad\,+\frac{\epsilon_{12}^2\,m_1\,m_2}{(a_1-a_2)^2}\,(a_1-a_2+m_1)\,(a_1-a_2-m_2)\,\qu_0^2+\cdots \ . \end{split}
\end{align}
\end{example}

From \eqref{eq:Orb_Zin} it follows that \smash{$ Z^{\vec r,\vec k}_{[\FC^4/\sGamma]}(\vec{a},\vec{\epsilon},\vec{m})$} is invariant under the action of the Weyl group \smash{$\timesbig_{s\in\shGamma}\, S_{r_s}$} of the colour and flavour symmetries by permuting the entries of the parameters $\vec a$ and $\vec m$.
It is natural to ask if the permutation symmetry of the type $\vec r=(0,\dots,0,r,0,\dots,0)$ instanton partition functions observed above persists in more generality. That is, whether the theories with rank vectors $\vec r=(r_s)_{s\in\shGamma}$ and ${}^\varsigma\vec r:=(r_{\varsigma(s)})_{s\in\shGamma}$ are equivalent, for a permutation $\varsigma\in S_{|\sGamma|}$ of degree $|\sGamma|$. Despite the appearance of the symmetric case in Example~\ref{ex:Z2Z2higherrank}, it is easy to see that this is not true in general.

\begin{example}
For the $\RZ_4$-action of Example~\ref{ex:orbZ4}, consider the rank two partition functions with framing vectors $\vec r=(1,1,0,0)$ and $\vec r=(1,0,1,0)$ for $\vec k=(1,1,0,0)$. In this case there are only two pairs of $\RZ_4$-coloured solid partitions that contribute to \eqref{eq:Orb_pf} which are given by
\begin{align}
\vec{\sigma}_1=(\sigma,\emptyset) \qquad \mbox{and} \qquad \vec{\sigma}_2=(\emptyset,\sigma) \ ,
\end{align}
where
\begin{align}
\sigma=\big\{(1,1,1,1) \,,\, (2,1,1,1)\big\} \ .
\end{align}
The corresponding charge two contributions are given by
\begin{align}
\begin{split}
Z_{[\FC^2/\RZ_4]\times\FC^2}^{\vec r=(1,1,0,0),\vec k=(1,1,0,0)}(\vec a,\vec\epsilon,\vec m) &= -\frac{\epsilon_{12}}{\epsilon_3\,\epsilon_4}\,\Big(m_1+m_2-\frac{m_1\,m_2}{a_1-a_2+\epsilon_1}\Big) \ , \\[4pt]
Z_{[\FC^2/\RZ_4]\times\FC^2}^{\vec r=(1,0,1,0),\vec k=(1,1,0,0)}(\vec a,\vec\epsilon,\vec m) &= -\frac{\epsilon_{12}}{\epsilon_3\,\epsilon_4}\,(m_1+m_2) \ .
\end{split}
\end{align}
\end{example}

\subsection{Pure $\mathcal{N}_{\textrm{\tiny T}}=2$ Gauge Theory on $[\FC^4/\sGamma]$}

By dropping the matter bundle contribution to the matrix integral \eqref{eq:Orb_Zin}, we obtain the quiver matrix model for the pure gauge theory on $[\FC^4/\sGamma]$ given by
\begin{align}
\begin{split}\label{Orb_Zpure}
& Z_{[\FC^4/\sGamma]}^{\vec r,\vec k}(\vec a,\vec\epsilon\,)^{\rm pure} = \oint_{\varGamma_{\vec r,\vec k}} \ \prod_{s\in\shGamma} \, \frac1{k_s!} \ \prod_{i=1}^{k_s} \, \frac{\dd\phi_i^s}{2\pi\,\ii} \ \frac{\displaystyle\prod_{\stackrel{\scriptstyle j=1}{\scriptstyle j\neq i}}^{k_s} \, \big(\phi^s_i-\phi^s_j\big)}{\CP_{r_s}(\phi_i^s|\vec a^s)}  \\
& \hspace{8cm} \times \prod_{i,j=1}^{k_s} \, \frac{\displaystyle \prod_{(\alpha,\beta)\in\ulthree^\perp}\, \big(\phi_i^{s+s_{\alpha\beta}}-\phi_j^s-\epsilon_{\alpha\beta}\big)}{\displaystyle \prod_{a\in\ulfour} \, \big(\phi_i^{s+s_a}-\phi_j^s-\epsilon_a\big)} \ .
\end{split}
\end{align}
The contour integral \eqref{Orb_Zpure} can be evaluated by using the index \eqref{eq:chTvir_Gamma} to get
\begin{align}\label{eq:orbpurechar}
\begin{split}
Z_{[\FC^4/\sGamma]}^{\vec r,\vec k}(\vec a,\vec\epsilon\,)^{\rm pure} &= \sum_{|\vec{\sigma}|_\shGamma=\vec k} \, (-1)^{{\tt O}^\sGamma_{\vec{\sigma}}} \ \prod_{l=1}^r \ \prod_{\vec p_l\in{\sigma}_l}^{\neq0} \, \frac{1}{\CP_r\circ\delta^\sGamma_0(a_l+\vec p_l\cdot\vec\epsilon\,|\vec a)} \\
& \hspace{4cm} \times \prod_{ l'=1}^r \ \prod_{\vec p^{\,\prime}_{l'}\in{\sigma}_{l'}}^{\neq0} \, \CR\circ\delta^\sGamma_0\big(a_l-a_{l'}+(\vec p_l-\vec p^{\,\prime}_{l'})\cdot\vec\epsilon\,\big|\vec\epsilon\,\big) \ .
\end{split}
\end{align}

The corresponding orbifold instanton partition function is 
\begin{align}
Z_{[\FC^4/\sGamma]}^{\vec r}(\vec\Lambda;\vec a,\vec\epsilon\,)^{\rm pure} = \sum_{\vec k\in\RZ_{\geq0}^{|\sGamma|}} \, \vec\Lambda^{\,\vec k} \ Z_{[\FC^4/\sGamma]}^{\vec r,\vec k}(\vec a,\vec\epsilon\,)^{\rm pure} \ ,
\end{align}
where $\vec\Lambda=(\Lambda_s)_{s\in\shGamma}$ and 
\begin{align}
\vec\Lambda^{\,\vec k} = \prod_{s\in\shGamma}\, \Lambda_s^{k_s} \ .
\end{align}

Similarly to the pure gauge theory on flat space $\FC^4$ from Section~\ref{sec:puregaugetheory}, the pure gauge theory on the quotient stack $[\FC^4/\sGamma]$ can be computed as an appropriate limit of the cohomological gauge theory on $[\FC^4/\sGamma]$ with massive matter. The orbifold version of Proposition~\ref{prop:puremassiverel} reads

\begin{proposition}\label{Prop4}
The instanton partition function for the pure cohomological gauge theory is related to the partition function with a massive fundamental hypermultiplet on $[\FC^4/\sGamma]$ through the double scaling limit
\begin{align}\label{eq:pure_limit_orb}
Z^{\vec r}_{[\FC^4/\sGamma]}(\vec \Lambda;\vec a,\vec{\epsilon}\,)^{\rm pure}=\lim_{m_1,\dots,m_r\to\infty} \ \lim_{\qu_0\to0} \,Z^{\vec r}_{[\mathbb{C}^4/\sGamma]}(\vec \qu;\vec a,\vec{\epsilon},\vec m)\,\Big|_{\Lambda_0=(-1)^r\,m_1\cdots m_r\,\qu_0} \ ,
\end{align}
with $\Lambda_s=\qu_s$ for $s\neq 0$.
\end{proposition}

\proof 
This follows immediately from \eqref{chiorb} and \eqref{eq:orbpurechar}.
\endproof

\section{The (2,0) Orbifolds $\FC^2/\RZ_n\times\FC^2$}
\label{sec:20C2ZnC2}

Proposition~\ref{prop3} suggests that the equivariant partition functions on $(2,0)$ orbifolds $\FC^2/\RZ_n \times\FC^2$ can be expressed as a generalization of equivariant partition functions on the toric three-orbifolds $\FC^2/\RZ_n\times\FC$, when the latter are known explicitly, similarly to the uplifting from Proposition~\ref{prop:ZDTgeb} to {Conjecture}~\ref{Prop1}. In this section we will propose uplifts of the rank one instanton partition functions, building on the known generating function for $\RZ_n$-coloured plane partitions from~\cite{Young:2008hn} and its extension to the toric three-orbifolds $\FC^2/\RZ_n\times\FC$ with $\sU(3)$ holonomy from~\cite{Zhou18}. We shall also propose conjectural closed formulas for the higher rank instanton partition functions, as well as consider the corresponding pure gauge theories. 

\subsection{Noncommutative $\sU(1)$ Instantons on $\FC^2/\RZ_{n}\times\FC^2$}\label{sec:C2ZnC2}

Generalizing Example~\ref{ex:orbZ4}, consider noncommutive $\sU(1)$ instantons on the toric Calabi--Yau four-orbifold $\FC^2/\RZ_n\times\FC^2$. The $\RZ_n$-action is generated by \eqref{gact}.
The cyclic group $\RZ_n$ has $n$ irreducible representations $\rho_s$, with $s=0,1,\dots,n-1$ and $\rho_s=\rho_1^{\otimes s}$. The fundamental representation of $\sSU(4)$ restricts to $\RZ_n$ as $Q=\rho_1\oplus\rho_{n-1}\oplus\rho_0\oplus\rho_0$, and the McKay quiver assumes the form
\begin{equation}
{\small
\begin{tikzcd}
	&& 0\arrow[,out=75,in=105,loop,swap,]    
	\arrow[,out=255,in=285,loop,swap,] \\
	1\arrow[,out=75,in=105,loop,swap,]    
	\arrow[,out=255,in=285,loop,swap,] &&&& {n-1}\arrow[,out=75,in=105,loop,swap,]    
	\arrow[,out=255,in=285,loop,swap,] \\
	& 2\arrow[,out=75,in=105,loop,swap,]    
	\arrow[,out=255,in=285,loop,swap,] && 3\arrow[,out=75,in=105,loop,swap,]    
	\arrow[,out=255,in=285,loop,swap,]
	\arrow[curve={height=6pt}, from=1-3, to=2-1]
	\arrow[curve={height=6pt}, from=2-1, to=1-3]
	\arrow[curve={height=6pt}, from=2-1, to=3-2]
	\arrow[curve={height=6pt}, from=3-2, to=2-1]
	\arrow[curve={height=6pt}, from=3-2, to=3-4]
	\arrow[curve={height=6pt}, from=3-4, to=3-2]
	\arrow[curve={height=6pt}, dotted, from=3-4, to=2-5]
	\arrow[curve={height=6pt}, dotted, from=2-5, to=3-4]
	\arrow[curve={height=6pt}, from=2-5, to=1-3]
	\arrow[curve={height=6pt}, from=1-3, to=2-5]
\end{tikzcd} }
\normalsize
\end{equation}
This is obtained as the double of the cyclic quiver ${\sf C}_n$ with arrows $s\to s+1$, whose underlying graph is the extended Dynkin diagram of type \smash{$\hat A_{n-1}$}, and adding a pair of loops at each node.

The fractional instanton contributions to the corresponding $\sU(1)$ partition function are given by
\begin{align}\label{eq:ZC2Znk}
\begin{split}
& Z_{[\FC^2/\RZ_n]\times\FC^2}^{\vec k}(\vec\epsilon,m) \\[4pt]
& \hspace{1cm} = \sum_{|{\sigma}|_{\RZ_n}=\vec k}  \, (-1)^{\ttO_{\sigma}^{\RZ_n}} \  \prod_{\substack{\vec p\,\in\sigma \\ p_1\,\equiv_n\, p_2}}^{\neq0} \, \frac{\vec p\cdot\vec \epsilon-m}{\vec p\cdot\vec \epsilon}  \\
& \hspace{3cm} \times \prod_{\substack{\vec p,\vec p\,'\in\sigma \\ p_1+p_2+p_3\,\equiv_n\, p_1'+p_2'+p_3'}}^{\neq0} \, \frac{\big((\vec p-\vec p\,')\cdot\vec \epsilon\,\big) \, \big((\vec p-\vec p\,')\cdot\vec\epsilon - \epsilon_{12}\big)}{\big((\vec p-\vec p\,')\cdot\vec \epsilon-\epsilon_{3}\big)\,\big((\vec p-\vec p\,')\cdot\vec \epsilon-\epsilon_{4}\big)} \\
& \hspace{5cm} \times  \prod_{\substack{\vec p,\vec p\,'\in\sigma \\ p_1-p_2+1\,\equiv_n\, p_1'-p_2'}}^{\neq0} \, \frac{\big( (\vec p-\vec p\,')\cdot\vec \epsilon -\epsilon_{13}\big)\,\big( (\vec p\,'-\vec p\,)\cdot\vec \epsilon-\epsilon_{23}\big)}{\big( (\vec p-\vec p\,')\cdot\vec \epsilon-\epsilon_{1}\big)\,\big( (\vec p\,'-\vec p\,)\cdot\vec \epsilon-\epsilon_{2}\big)} \ ,
\end{split}
\end{align}
where $\equiv_n$ denotes congruence modulo~$n$.

The equivariant partition function for rank one noncommutative Donaldson--Thomas invariants of the toric three-orbifold $\FC^2/\RZ_n\times\FC$ with $\sU(3)$ holonomy can be written in the closed form~\cite{Zhou18}
\begin{align}\label{orb6dZn}
Z_{[\FC^2/\RZ_n]\times\FC}(\vec\qu;\epsilon_1,\epsilon_2,\epsilon_3)=M(-\Qu)^{-n\,\frac{\epsilon_{12}\,\epsilon_{13}\,\epsilon_{23}}{\epsilon_1\,\epsilon_2\,\epsilon_3}+\frac{n^2-1}{n}\,\frac{\epsilon_{12}\,\epsilon_{123}}{\epsilon_1\,\epsilon_2}} \ \prod_{0<p\leq s<n}\,\widetilde{M}(\qu_{[p,s]},-\Qu)^{-\frac{\epsilon_{12}}{\epsilon_3}} \ ,
\end{align} 
where 
\begin{align}
\Qu=\qu_0\,\qu_1\cdots \qu_{n-1} \qquad \mbox{and} \qquad \qu_{[p,s]}=\qu_p\,\qu_{p+1}\cdots \qu_{s-1}\,\qu_s \ , 
\end{align}
while
 \begin{align}
 M(x,q)=\prod_{k=1}^\infty\,\frac{1}{(1-x\,q^k)^k} \qquad \mbox{and} \qquad \widetilde{M}(x,q)=M(x,q)\, M(x^{-1},q) \ .
 \end{align}
The generalized MacMahon function $M(x,q)$ counts weighted plane partitions (cf. Remark~\ref{not:solidpartitions}) and specialises to \eqref{eq:MacMahon} at $x=1$: $M(1,q)=M(q)$.

At the Calabi--Yau specialization $\epsilon_1+\epsilon_2+\epsilon_3=0$ of the $\varOmega$-deformation, the formula \eqref{orb6dZn} reduces to the instanton partition function on $\FC^2/\RZ_n\times\FC$ with $\sSU(3)$ holonomy and agrees with the closed formula found in \cite{Young:2008hn,Quiver3d}:
\begin{align}
\begin{split}
Z_{[\FC^2/\RZ_n]\times\FC}(\vec \qu;\epsilon_1,\epsilon_2,\epsilon_3)\big|_{\epsilon_{123}=0}&=M(-\Qu)^n \ \prod_{0<p\leq s<n}\, \widetilde{M}(\qu_{[p,s]},-\Qu)\\[4pt]
&=\sum_{\pi}\,(-\qu_0)^{|\pi|_0}\,\qu_1^{|\pi|_1}\, \cdots \qu_{n-1}^{|\pi|_{n-1}} \ .\end{split}\label{orb6dZnCY}
\end{align}
In the second equality the sum runs through plane partitions $\pi$ which are $\RZ_n$-coloured with respect to \eqref{coloring}. 

Together with Proposition~\ref{prop3}, the forms of the partition functions \eqref{orb6dZn} and \eqref{orb6dZnCY} suggest that the partition function for $\sU(1)$ noncommuative instantons on $\FC^2/\RZ_n\times \FC^2$ assumes a closed form as a combination of generalized MacMahon functions $M(x,q)$. This is encapsulated in

\begin{conjecture}\label{con2}
The equivariant  instanton partition function  of the cohomological $\sU(1)$ gauge theory with massive fundamental matter on $[\mathbb{C}^2/\mathbb{Z}_n]\times \mathbb{C}^2$ is given by
\begin{align}\label{eq:rank1Zn}
\begin{split}
Z_{[\FC^2/\RZ_n]\times \FC^2}(\vec \qu;\vec{\epsilon},m) & = M\big((-1)^n\,\Qu\big)^{-n\,\frac{m\,\epsilon_{12}\,\epsilon_{13}\,\epsilon_{23}}{\epsilon_1\,\epsilon_2\,\epsilon_3\,\epsilon_4}-\frac{n^2-1}{n}\,\frac{m\,\epsilon_{12}}{\epsilon_1\,\epsilon_2}} \\
& \quad \, \times \prod_{0<p\leq s<n}\,\widetilde{M}\big((-1)^{p-s+1}\,\qu_{[p,s]},(-1)^n\,\Qu\big)^{-\frac{m\,\epsilon_{12}}{\epsilon_3\,\epsilon_4}} \ .
\end{split}
\end{align}
\end{conjecture}

Using the combinatorial expression \eqref{eq:ZC2Znk} we explicitly checked

\begin{proposition}
Conjecture~\ref{con2} is true for $n=2,3,4$ and $k=|\vec k\,|=1,2,3$.
\end{proposition}

\begin{remark}
In Appendix~\ref{app:con2} we sketch the steps of a possible direct proof of Conjecture~\ref{con2}, analogous to the sketch of the analytic proof of {Conjecture}~\ref{Prop1} .
\end{remark}

\subsection{Higher Rank Generalizations} \label{sec:higher_rank}

{Conjecture}~\ref{Prop1} shows that  the equivariant instanton partition function of the rank $r>1$ cohomological
gauge theory with a massive fundamental hypermultiplet on $\FC^4$ is given by the simple closed formula \eqref{PfRank}, expressing it as the $r$-th power of the rank one partition function with the mass parameter $m$ given by coordinates on the center $\sU(1|1)$ of the global symmetry $\sU(r|r)\supset\sU(r)_{\rm col}\times\sU(r)_{\rm fla}$. We would now like to generalize Conjecture~\ref{con2} to the case of higher rank framing vectors $\vec r$. Unfortunately, we do not have available any higher rank noncommutative Donaldson--Thomas partition functions on the orbifold $\FC^2/\RZ_n\times\FC$ with generic $\sU(3)$ holonomy to uplift, so our results will be limited to a set of well-substantiated conjectures.

\subsubsection*{Instantons on $\mbf{\FC^2/\RZ_n\times\FC^2}$ of Type $\vec r=(0,\dots,0,r,0,\dots,0)$}

For the framing vectors of the form $\vec r=(r,0,\dots,0)$, we propose

\begin{conjecture}\label{con4}
Set
\begin{align}
 m=\frac{1}{r} \, \sum_{l=1}^r\,(m_l-a_l) \ .
\end{align}
Then the equivariant  instanton partition function of type $\vec{r}=(r,0,\dots,0)$ for the cohomological $\sU(r)$ gauge theory with massive fundamental matter on $[\FC^2/\RZ_n]\times \FC^2$ is given by
\begin{align}\label{eq:Orb_pf_C2ZnC2}
\begin{split}
Z^{\vec r=(r,0,\dots,0)}_{[\FC^2/\RZ_n]\times \FC^2}(\vec \qu;\vec{\epsilon},m) & = M\big((-1)^n\,\Qu\big)^{-n\,\frac{m\,r\,\epsilon_{12}\,\epsilon_{13}\,\epsilon_{23}}{\epsilon_1\,\epsilon_2\,\epsilon_3\,\epsilon_4}-\frac{n^2-1}{n}\,\frac{m\,r\,\epsilon_{12}}{\epsilon_1\,\epsilon_2}} \\
& \quad \, \times \prod_{0<p\leq s<n}\,\widetilde{M}\big((-1)^{p-s+1}\,\qu_{[p,s]},(-1)^n\,\Qu\big)^{-\frac{m\,r\,\epsilon_{12}}{\epsilon_3\,\epsilon_4}} \ .
\end{split}
\end{align}
\end{conjecture}

Using the combinatorial expression \eqref{eq:Orb_pf} we explicitly checked

\begin{proposition}
Conjecture~\ref{con4} is true for $r=2,3$, $k=|\vec k\,|=1,2$ and $n=2,3$.
\end{proposition}

\begin{remark}
The partition function \eqref{eq:Orb_pf_C2ZnC2} is invariant under permutation of the location of the entry $r$ in the array $\vec r=(r,0,\dots,0)$ (cf. Section~\ref{sec:Orb_def_complex}). 
\end{remark} 

\begin{proposition} \label{prop:C2ZnU3}
Assume Conjecture~\ref{con4} is true. Then the partition function for rank $r$ noncommutative Donaldson--Thomas invariants of type~$\vec r=(r,0,\dots,0)$ for the orbifold  $\FC^2/\RZ_n\times\FC$ with $\sU(3)$ holonomy is given by
\begin{align}\label{eq:C2ZnU3}
\begin{split}
Z^{\vec{r}=(r,0,\dots,0)}_{[\FC^2/\RZ_n]\times\FC}(\vec\qu;\epsilon_1,\epsilon_2,\epsilon_3)&=M\big((-1)^{r}\,\Qu\big)^{-n\,\frac{r\,\epsilon_{12}\,\epsilon_{13}\,\epsilon_{23}}{\epsilon_1\,\epsilon_2\,\epsilon_3}+\frac{n^2-1}{n}\,\frac{r\,\epsilon_{12}\,\epsilon_{123}}{\epsilon_1\,\epsilon_2}} \\
& \quad \, \times \prod_{0<p\leq s<n}\,\widetilde{M}\big(\qu_{[p,s]},(-1)^{r}\,\Qu\big)^{-\frac{r\,\epsilon_{12}}{\epsilon_3}} \ .\end{split}
\end{align}
\end{proposition}

\proof
This follows immediately from \eqref{eq:Orb_pf_C2ZnC2} by using Proposition~\ref{prop3}. \endproof

\subsubsection*{Instantons on $\mbf{\FC^2/\RZ_2\times\FC^2}$ of Type $\vec r=(r_0,r_1)$}

It appears to be difficult to make a concrete conjecture for the closed form of the instanton partition function on $\FC^2/\RZ_n\times\FC^2$ for generic framing vectors $\vec r=(r_0,r_1,\dots,r_{n-1})$, as the calculations become much more cumbersome in general. Here we consider the simplest case $n=2$, and propose a closed formula which holds in full generality.

\begin{example}\label{Ex3}
The leading terms in the expansions of the rank two partition functions for the orbifold $\FC^2/\RZ_2\times\FC^2$ are
\begin{align}\begin{split}
 Z_{[\FC^2/ \RZ_2]\times\FC^2}^{\vec r=(2,0)}(\qu_0,\qu_1;\vec a,\vec{\epsilon},\vec{m})&=2\,\frac{m\,\epsilon_{12}}{\epsilon_3\,\epsilon_4}\,\qu_0+\frac{m\,\epsilon_{12}}{\epsilon_1\,\epsilon_2}\,\Big(4\,\frac{\epsilon_{13}\,\epsilon_{3}}{\epsilon_3\,\epsilon_4}+3\Big)\,\qu_0\,\qu_1 \\
 &\quad\,+\frac{m\,\epsilon_{12}}{\epsilon_3\,\epsilon_4}\,\Big(2\,\frac{m\,\epsilon_{12}}{\epsilon_3\,\epsilon_4}-1\Big)\,\qu^2_0+\cdots \ , \\[4pt]
Z_{[\FC^2/ \RZ_2]\times\FC^2}^{\vec r=(1,1)}(\qu_0,\qu_1;\vec a,\vec{\epsilon},\vec{m})&=
2\,\frac{m\,\epsilon_{12}}{\epsilon_3\,\epsilon_4}\,\qu_0+\Big(4\,\frac{m\,\epsilon_{12}\,\epsilon_{13}\,\epsilon_{3}}{\epsilon_1\,\epsilon_2\,\epsilon_3\,\epsilon_4}+3\,\frac{m\,\epsilon_{12}}{\epsilon_1\,\epsilon_2} \\
& \hspace{5cm} +\frac{(m_1-a_1)\,(m_2-a_2)\,\epsilon_{12}}{\epsilon_3\,\epsilon_4}\Big)\,\qu_0\,\qu_1 \\
&\quad\, +\frac{m\,\epsilon_{12}}{\epsilon_3\,\epsilon_4}\Big(2\,\frac{m\,\epsilon_{12}}{\epsilon_3\,\epsilon_4}-1\Big)\,\qu^2_0+\cdots \ ,
\end{split}
\end{align}
where $m=\frac12\,(m_1-a_1+m_2-a_2)$.
\end{example}

Since we are working with the $\RZ_2$-action of Section~\ref{sec:C2ZnC2} in evaluating the contribution \eqref{chiorb} from the index $\chi^{\RZ_2}_{\vec{\sigma}}$, we can use
\begin{align}
\rho_s\otimes\rho_{s'}=\begin{cases}
\rho_0 \quad \text{if} \ s = s' \ ,\\
\rho_1 \quad \mbox{otherwise} \ ,
\end{cases}
\end{align}
for $s,s'\in\{0,1\}$.
Then the partition function \eqref{eq:Orb_pf} is symmetric under permutation of $r_0$ and $r_1$, and it must reduce to \eqref{eq:Orb_pf_C2ZnC2} with $n=2$ for $\vec{r}=(r,0)$ and $\vec{r}=(0,r)$. These properties together with Example~\ref{Ex3} prompt us to formulate

\begin{conjecture}\label{con:r0r1}
The equivariant  instanton partition function  of type $\vec{r}=(r_0,r_1)$ for the cohomological gauge theory with massive fundamental matter on $[\FC^2/\RZ_2]\times \FC^2$ is given by
\begin{align}\label{eq:C2Z2r0r1}
\begin{split}
Z_{[\FC^2/\RZ_2]\times\FC^2}^{\vec{r}=(r_0,r_1)}(\vec \qu;\vec a,\vec{\epsilon},\vec{m}) &= M(\Qu)^{-2\,\frac{m\,r\,\epsilon_{12}\,\epsilon_{13}\,\epsilon_{3}}{\epsilon_1\,\epsilon_2\,\epsilon_3\,\epsilon_4}-\frac{3}{2}\,\frac{m\,r\,\epsilon_{12}}{\epsilon_1\,\epsilon_2}} \\
& \quad\, \times \widetilde{M}(-\qu_1,\qu_0\,\qu_1)^{-\frac{m^0\,r_0\,\epsilon_{12}}{\epsilon_3\,\epsilon_4}} \, \widetilde{M}(-\qu_0,\qu_0\,\qu_1)^{-\frac{m^1\,r_1\, \epsilon_{12}}{\epsilon_3\,\epsilon_4}} \ ,
\end{split}
\end{align}
where $r=r_0+r_1$ and
\begin{align}
m^s=\frac1{r_s}\,\sum_{s(l)=s} \, (m_l-a_l)
\end{align}
are coordinates on the center of the global symmetry supergroup $\sU(r_s|r_s)$ associated to the irreducible representation $\rho_s$ of $\RZ_2$, for $s\in\{0,1\}$.
\end{conjecture}

\begin{proposition} \label{prop:C2Z2CU3}
Assume Conjecture~\ref{con:r0r1} is true. Then the partition function for noncommutative Donaldson--Thomas invariants of type~$\vec r=(r_0,r_1)$ for the orbifold  $\FC^2/\RZ_2\times\FC$ with $\sU(3)$ holonomy is given by
 \begin{align}\label{eq:C2Z2CU3}
 \begin{split}
Z_{[\FC^2/\RZ_2]\times\FC}^{\vec{r}=(r_0,r_1)}(\vec\qu;\epsilon_1,\epsilon_2,\epsilon_3)&=
M\big((-1)^{r}\,\Qu\big)^{-2\,\frac{r\,\epsilon_{12}\,\epsilon_{13}\,\epsilon_{23}}{\epsilon_1\,\epsilon_2\,\epsilon_3}+\frac{3}{2}\,\frac{r\,\epsilon_{12}\,\epsilon_{123}}{\epsilon_1\,\epsilon_2}} \\
& \hspace{1cm} \times
\widetilde{M}\big((-1)^{r_1}\,\qu_1,(-1)^{r}\, \qu_0\,\qu_1\big)^{-\frac{r_0\,\epsilon_{12}}{\epsilon_3}}\\
& \hspace{2cm} \times \widetilde{M}\big((-1)^{r_0}\,\qu_0,(-1)^{r}\, \qu_0\,\qu_1\big)^{-\frac{r_1\,\epsilon_{12}}{\epsilon_3}} \ .
\end{split}
 \end{align}
\end{proposition}

\begin{proof}
This follows immediately from \eqref{eq:C2Z2r0r1} by using Proposition~\ref{prop3}. 
\end{proof}

\begin{remark}
At the Calabi--Yau specialization $\epsilon_{123}=0$, both partition functions \eqref{eq:C2ZnU3} and \eqref{eq:C2Z2CU3} agree with the generating functions for Coulomb branch invariants found in \cite[Section~7]{Quiver3d}.
\end{remark}

\subsection{Pure $\mathcal{N}_{\textrm{\tiny T}}=2$ Gauge Theory on $[\FC^2/\RZ_n]\times\FC^2$}

Using Proposition \ref{Prop4} together with the results of Sections~\ref{sec:C2ZnC2} and~\ref{sec:higher_rank} we can immediately infer corresponding closed formulas for the partition functions of the pure gauge theories on the quotient stack $[\FC^2/\RZ_n]\times\FC^2$. They read as

\begin{proposition}\label{prop:pureorbZn}
If Conjecture~\ref{con2} is true, then the equivariant instanton partition function of the pure cohomological $\sU(1)$ gauge theory on $[\FC^2/\RZ_n]\times\FC^2$ is given by
\begin{align}\begin{split}\label{Pure_orb_Zn}
Z_{[\FC^2/\RZ_n]\times\FC^2}(\vec \Lambda;\vec{\epsilon}\,)^{\rm pure}&=\exp \,(-1)^n\,\varLambda\,\Big(n\,\frac{\epsilon_{12}\,\epsilon_{13}\,\epsilon_{23}}{\epsilon_1\,\epsilon_2\,\epsilon_3\,\epsilon_4}+\frac{n^2-1}{n}\,\frac{\epsilon_{12}}{\epsilon_3\,\epsilon_4}\\
& \hspace{3cm} -\frac{\epsilon_{12}}{\epsilon_3\,\epsilon_4}\,\sum_{0<p\leq s<n}\,(-1)^{p-s}\,\big(\Lambda_{[p,s]}+\Lambda_{[p,s]}^{-1}\big)\Big) \ ,\end{split} 
\end{align}
where
\begin{align}
\varLambda=\Lambda_0\,\Lambda_1\cdots\Lambda_n \qquad \mbox{and} \qquad\Lambda_{[p,s]}=\Lambda_p\,\Lambda_{p+1}\cdots\Lambda_{s-1}\,\Lambda_s \ .
\end{align}
If in addition Conjecture~\ref{con4} is true, then the higher rank partition functions of type $\vec r=(r,0,\dots,0)$ are all trivial:
\begin{align}
Z^{\vec r=(r,0,\dots,0)}_{[\FC^2/\RZ_n]\times\FC^2}(\vec\Lambda;\vec{\epsilon}\,)^{\rm pure} = 1 \qquad \mbox{for} \quad r> 1 \ .
\end{align}
\end{proposition}

\begin{proof}
The proof is completely analogous to the proof of Corollary~\ref{prop:pureC4} using the series representation
\begin{align}
\log M(x,q) = \sum_{k=1}^\infty\,\frac{x^k}{k}\,\frac{q^k}{(1-q^k)^2}
\end{align}
for the logarithm of the generalized MacMahon function.
\end{proof}

\section{The (3,0) Orbifold $\FC^3/(\RZ_2\times\RZ_2)\times\FC$}
\label{sec:30C3Z2Z2C}

As another explicit example, we repeat our treatment from Section~\ref{sec:20C2ZnC2} in the case of the $(3,0)$ orbifold $\FC^3/\sGamma\times\FC$ for the action of the group $\sGamma=\RZ_2\times\RZ_2$ in $\sSL(3,\FC)$ defined below. By the results of~\cite{Wemyss18} it admits four geometric projective toric crepant resolutions related to each other by flops, which can each be constructed as fine moduli spaces of stable $\sGamma$-constellations; for the symmetric resolution induced by $\Hilb^{\RZ_2\times\RZ_2}(\FC^3)$, the geometric crepant resolution of $\FC^3/(\RZ_2\times\RZ_2)\times\FC$ is a fibration of closed topological vertex geometries over $\FC$. Whereas the dimensionally reduced partition functions are well understood for the Calabi--Yau three-orbifolds $\FC^3/\RZ_2\times\RZ_2$ by the results of~\cite{Young:2008hn,Quiver3d}, we are currently lacking closed expressions for generic $\sU(3)$ holonomy, even in the rank one case. Hence in this section our considerations will again be limited to conjectural but well-substantiated closed formulas.

\subsection{Noncommutative $\sU(1)$ Instantons on $\FC^3/(\RZ_{2}\times\RZ_{2})\times\FC$}

Consider the toric Calabi--Yau four-orbifold $\FC^3/(\RZ_{2}\times\RZ_{2})\times\FC$, where the action of the group $\mathbbm{Z}_{2}\times\mathbbm{Z}_{2}=\{\ident,g_1,g_2,g_3\}$ on $\mathbbm{C}^4$  is given by the $\sSU(4)$ matrices
\begin{align}
    g_1={\small \begin{pmatrix}
    -1&&&\\
    &-1&&\\
    &&1&\\
    &&&1
    \end{pmatrix}} \normalsize \qquad \mbox{and} \qquad  g_2={\small \begin{pmatrix}
    -1&&&\\
    &1&&\\
    &&-1&\\
    &&&1
    \end{pmatrix} } \normalsize \ ,
\end{align}
together with $g_3=g_1\,g_2$.
The four irreducible representations $\shGamma=\{\rho_{0},\rho_1,\rho_2,\rho_3\}$ have weights $s_0=(0,0,0,0)$, $s_1=(1,1,0,0)$, $s_2=(1,0,1,0)$ and $s_3=s_1+s_2=(0,1,1,0)$, respectively.

The tensor product decomposition of the
fundamental representation $Q$ gives an adjacency matrix
\begin{align}
A=(a_{ss'})= {\small \begin{pmatrix}
1&1&1&1\\
1&1&1&1\\
1&1&1&1\\
1&1&1&1
\end{pmatrix} } \normalsize \ .
\end{align}
The generalized McKay quiver constructed from the representation theory of $\RZ_2\times\RZ_2$ is  
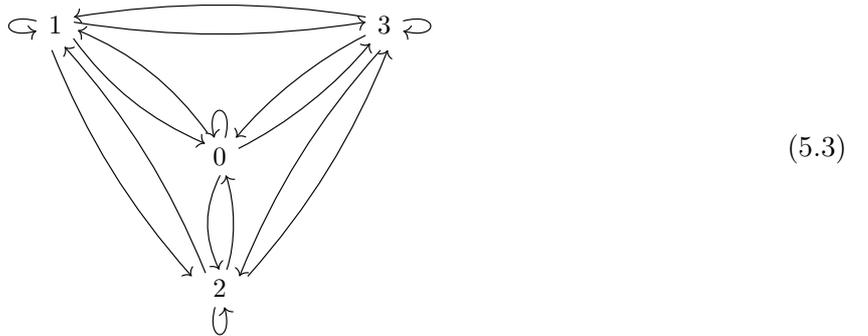
\begin{equation}
{\small
\begin{tikzcd}
	1\arrow[,out=165,in=195,loop,swap,] &&&& 3\arrow[,out=15,in=-15,loop,swap,] \\
	\\
	&& 0\arrow[,out=75,in=105,loop,swap,] \\
	\\
	&& 2\arrow[,out=255,in=285,loop,swap,]
	\arrow[shift left=1, curve={height=6pt}, from=1-1, to=1-5]
	\arrow[shift right=1, curve={height=6pt}, from=1-5, to=1-1]
	\arrow[curve={height=6pt}, from=1-5, to=3-3]
	\arrow[curve={height=6pt}, from=3-3, to=1-5]
	\arrow[shift left=2, curve={height=6pt}, from=1-5, to=5-3]
	\arrow[shift right=2, curve={height=6pt}, from=5-3, to=1-5]
	\arrow[shift left=1, curve={height=6pt}, from=3-3, to=5-3]
	\arrow[curve={height=6pt}, from=5-3, to=3-3]
	\arrow[shift right=2, curve={height=6pt}, from=1-1, to=5-3]
	\arrow[shift left=1, curve={height=6pt}, from=5-3, to=1-1]
	\arrow[shift left=1, curve={height=6pt}, from=1-1, to=3-3]
	\arrow[shift right=1, curve={height=6pt}, from=3-3, to=1-1]
\end{tikzcd}
} \normalsize
\end{equation}

To write the corresponding $\sU(1)$ instanton partition function explicitly, it is convenient to use the imaginary unit quaternions $\ii$, $\jj$ and $\kk$ which satisfy the relations
\begin{align}
    \ii^2=\jj^2=\kk^2=\ii\,\jj\,\kk=-1 \ .
\end{align}
Then the fractional instanton contributions are given by
\begin{align}
\begin{split}
& Z_{[\FC^3/\RZ_2\times\RZ_2]\times\FC}^{\vec k}(\vec\epsilon,m) \\[4pt]
& \quad = \sum_{|{\sigma}|_{\RZ_2\times\RZ_2}=\vec k} \, (-1)^{\ttO_{\sigma}^{\RZ_2\times\RZ_2}} \ \prod_{\substack{\vec p\,\in\sigma \\ \ii^{p_1}\,\jj^{p_2}\,\kk^{p_3}=\pm 1}}^{\neq0} \, \frac{\vec p\cdot\vec \epsilon-m}{\vec p\cdot\vec \epsilon} \ \prod_{\substack{\vec p,\vec p\,'\in\sigma \\ \ii^{p_1-p_1'}\,\jj^{p_2-p_2'}\,\kk^{p_3-p_3'}=\pm 1}}^{\neq0} \, \frac{(\vec p-\vec p\,')\cdot\vec \epsilon}{(\vec p-\vec p\,')\cdot\vec \epsilon-\epsilon_{4}} \\
& \hspace{2cm} \times \prod_{\substack{\vec p,\vec p\,'\in\sigma \\ \ii^{p_1-p_1'}\,\jj^{p_2-p_2'}\,\kk^{p_3-p_3'}=\pm \ii}}^{\neq0} \, \frac{(\vec p-\vec p\,')\cdot\vec \epsilon -\epsilon_{23}}{(\vec p-\vec p\,')\cdot\vec \epsilon-\epsilon_{1}} \
\prod_{\substack{\vec p,\vec p\,'\in\sigma \\ \ii^{p_1-p_1'}\,\jj^{p_2-p_2'}\,\kk^{p_3-p_3'}=\pm \jj}}^{\neq0} \, \frac{(\vec p-\vec p\,')\cdot\vec \epsilon -\epsilon_{13}}{(\vec p-\vec p\,')\cdot\vec \epsilon-\epsilon_{2}} \\
& \hspace{4cm} \times \prod_{\substack{\vec p,\vec p\,'\in\sigma \\ \ii^{p_1-p_1'}\,\jj^{p_2-p_2'}\,\kk^{p_3-p_3'}=\pm \kk}}^{\neq0} \, \frac{(\vec p-\vec p\,')\cdot\vec \epsilon -\epsilon_{12}}{(\vec p-\vec p\,')\cdot\vec \epsilon-\epsilon_{3}} \ .
\end{split}
\end{align}

The rank one Donaldson--Thomas partition function for the three-orbifold $\FC^3/\RZ_2\times\RZ_2$ with holonomy group $\sSU(3)$ is discussed in detail in~\cite{Young:2008hn,Quiver3d}. It can be expressed as the closed formula
\begin{align}\begin{split}
Z_{[\mathbbm{C}^3/\mathbbm{Z}_2\times\mathbbm{Z}_2]}(\vec\qu;\epsilon_1,\epsilon_2,\epsilon_3)\big|_{\epsilon_{123}=0} &= \frac{M(-\Qu)^4}{L(\qu_1,\qu_2,\qu_3,-\Qu)} \ \prod_{1\leq p<s\leq 3}\,\widetilde{M}(\qu_p\,\qu_s,-\Qu)\\[4pt]
&=\sum_{\pi} \, (-1)^{|\pi|_1+|\pi|_2+|\pi|_3} \, \qu_0^{|\pi|_0} \, \qu_1^{|\pi|_1} \, \qu_2^{|\pi|_2} \, \qu_3^{|\pi|_3} \ ,
\end{split}\label{Orb_6dZ2Z2}
\end{align}
where the sum runs through $\RZ_2\times\RZ_2$-coloured plane partitions $\pi$. We have set $\Qu=\qu_0\,\qu_1\,\qu_2\,\qu_3$ and
\begin{align}
L(x_1,x_2,x_3,q)=\widetilde{M}(x_1,q)\,\widetilde{M}(x_2,q)\,\widetilde{M}(x_3,q)\,\widetilde{M}(x_1\,x_2\,x_3,q) \ .
\end{align}

With an argument analogous to that of Section \ref{sec:C2ZnC2}, we expect that the rank one instanton partition function for $\FC^3/(\RZ_{2}\times\RZ_{2})\times\FC$ assumes a closed form in terms of a generalization of the formula \eqref{Orb_6dZ2Z2}. Using \eqref{eq:Orb_pf} we evaluate the leading terms in the expansion of $Z_{[\FC^3/\RZ_2\times\RZ_2]\times\FC}(\vec \qu;\vec{\epsilon},m)$ to be
\begin{align}\begin{split}
Z_{[\FC^3/\RZ_2\times\RZ_2]\times\FC}(\vec \qu;\vec{\epsilon},m)&=\frac{m}{\epsilon_4}\,\qu_0+\frac{(m-\epsilon_{4})\,m}{2\,\epsilon_4^2}\,\qu_0^2+\frac{(\epsilon_1-\epsilon_{23})\,m}{2\,\epsilon_1\,\epsilon_4}\,\qu_0\,\qu_1+\frac{(\epsilon_2-\epsilon_{13})\,m}{2\,\epsilon_2\,\epsilon_4}\,\qu_0\,\qu_2\\
& \quad \, +\frac{(\epsilon_3-\epsilon_{12})\,m}{2\,\epsilon_2\,\epsilon_4}\,\qu_0\,\qu_3+\frac{m}{\epsilon_4}\,\qu_0\,\qu_1\,\qu_2 +\cdots \ . \end{split}
\end{align}
Dimensionally reducing this result according to Proposition~\ref{prop3} at the Calabi--Yau specialization of the $\varOmega$-deformation on $\FC^3$ yields the leading contributions predicted by \eqref{Orb_6dZ2Z2}:
\begin{align}
\begin{split}
Z_{[\FC^3/\RZ_2\times\RZ_2]\times\FC}(\vec \qu;\vec{\epsilon},m=\epsilon_4)\big|_{\epsilon_{123}=0} &= \qu'_0-\qu'_0\,\qu'_1+\qu'_0\,\qu'_2+\qu'_0\,\qu'_3+\qu'_0\,\qu'_1\,\qu'_2+\cdots \\[4pt]
&= Z_{[\mathbbm{C}^3/\mathbbm{Z}_2\times\mathbbm{Z}_2]}(\vec\qu\,';\epsilon_1,\epsilon_2,\epsilon_3)\big|_{\epsilon_{123}=0} \ ,
\end{split}
\end{align}
where $\vec\qu\,'=(\qu_0,-\qu_1,-\qu_2,-\qu_3)$. 

This prompts us to formulate

\begin{conjecture}\label{con3}
The equivariant  instanton partition function  of the cohomological $\sU(1)$ gauge theory with massive fundamental matter on $[\FC^3/\RZ_2\times\RZ_2]\times \FC$ is given by
\begin{align}\begin{split}
Z_{[\FC^3/\RZ_2\times\RZ_2]\times\FC}(\vec \qu;\vec{\epsilon},m)&=\frac{M(\Qu)^{m\,\frac{\epsilon_1\,\epsilon_2\,\epsilon_3-\epsilon_1^2\,\epsilon_2-\epsilon^2_1\,\epsilon_3-\epsilon^2_2\,\epsilon_3-\epsilon_1\,\epsilon_2^2-\epsilon_1\,\epsilon_3^2-\epsilon_2\,\epsilon^2_3}{\epsilon_1\,\epsilon_2\,\epsilon_3\,\epsilon_4}}}{L(-\qu_1,-\qu_2,-\qu_3,\Qu)^{\frac{m}{\epsilon_4}}} \\
& \quad \, \times \prod_{1\leq p<s\leq 3}\,\widetilde{M}(\qu_p\,\qu_s,\Qu)^{m\,\frac{\epsilon_{(ps)^-}-\epsilon_{ps}}{2\,\epsilon_4\,\epsilon_{(ps)^-}}} \ ,\end{split}\label{eq:ZC2C2}
\end{align}
where $(ps)^-\in\{1,2,3\}\setminus\{p,s\}$.
\end{conjecture}

\begin{proposition} \label{prop:6dZC2C2}
Assume Conjecture~\ref{con3} is true. Then the partition function for rank one noncommutative Donaldson--Thomas invariants of the orbifold $\FC^3/\RZ_2\times \RZ_2$ with holonomy group $\sU(3)$  is given by
\begin{align}\label{eq:6dZC2C2}
\begin{split}
Z_{[\FC^3/\RZ_2\times\RZ_2]}(\vec \qu;\epsilon_1,\epsilon_2,\epsilon_3)&=\frac{M(-\Qu)^{\frac{\epsilon_1\,\epsilon_2\,\epsilon_3-\epsilon_1^2\,\epsilon_2-\epsilon^2_1\,\epsilon_3-\epsilon^2_2\,\epsilon_3-\epsilon_1\,\epsilon_2^2-\epsilon_1\,\epsilon_3^2-\epsilon_2\,\epsilon^2_3}{\epsilon_1\,\epsilon_2\,\epsilon_3}}}{L(\qu_1,\qu_2,\qu_3,-\Qu)} \\
& \quad \, \times \prod_{1\leq p<s\leq 3}\,\widetilde{M}(\qu_p\,\qu_s,-\Qu)^{\frac{\epsilon_{(ps)^-}-\epsilon_{ps}}{2\,\epsilon_{(ps)^-}}} \ .
\end{split}
\end{align}
\end{proposition}

\begin{proof}
This follows immediately from \eqref{eq:ZC2C2} by using Proposition \ref{prop3}.
\end{proof}

\begin{remark}
For holonomy group $\sSU(3)$, the partition function \eqref{eq:6dZC2C2} reduces to \eqref{Orb_6dZ2Z2}. As noted by~\cite{CKMpreprint}, it is possible to adapt the techniques of~\cite{Zhou18} to the orbifold $\FC^3/\RZ_2\times\RZ_2$ and hence provide a direct proof of \eqref{eq:6dZC2C2}.
\end{remark}

\subsection{Higher Rank Generalization}

For higher rank noncommutative instantons on $\FC^3/(\RZ_2\times\RZ_2)\times\FC$ of type $\vec r=(0,\dots,0,r,0,\dots,0)$, the analogue of Conjecture~\ref{con4} reads as

\begin{conjecture}\label{con4a}
The equivariant  instanton partition function of type $\vec{r}=(r,0,\dots,0)$ for the cohomological $\sU(r)$ gauge theory with massive fundamental matter on $[\FC^3/\RZ_2\times\RZ_2]\times \FC$ is given by
\begin{align}\label{eq:Orb_pf_C3Z2Z2C}
\begin{split}
Z^{\vec{r}=(r,0\dots,0)}_{[\FC^3/\RZ_2\times\RZ_2]\times\FC}(\vec \qu;\vec{\epsilon},m)
&=\frac{M(\Qu)^{m\,r\,\frac{\epsilon_1\,\epsilon_2\,\epsilon_3-\epsilon_1^2\,\epsilon_2-\epsilon^2_1\,\epsilon_3-\epsilon^2_2\,\epsilon_3-\epsilon_1\,\epsilon_2^2-\epsilon_1\,\epsilon_3^2-\epsilon_2\,\epsilon^2_3}{\epsilon_1\,\epsilon_2\,\epsilon_3\,\epsilon_4}}}{L(-\qu_1,-\qu_2,-\qu_3,\Qu)^{\frac{m\,r}{\epsilon_4}}} \\
& \quad \, \times \prod_{1\leq p<s\leq 3}\,\widetilde{M}(\qu_p\,\qu_s,\Qu)^{m\,r\,\frac{\epsilon_{(ps)^-}-\epsilon_{ps}}{2\,\epsilon_4\,\epsilon_{(ps)^-}}} \ . \end{split}
\end{align}
\end{conjecture}

Using the combinatorial expansion \eqref{eq:Orb_pf} we explicitly checked

\begin{proposition}
Conjecture~\ref{con4a} is true for $r=2,3$ and $k=|\vec k\,|=1,2$.
\end{proposition}

\begin{proposition} \label{prop:C3Z2Z2U3}
Assume Conjecture~\ref{con4a} is true. Then the partition function for rank $r$ noncommutative Donaldson--Thomas invariants of type~$\vec r=(r,0,\dots,0)$ for the orbifold  $\FC^3/\RZ_2\times\RZ_2$ with $\sU(3)$ holonomy is given by
\begin{align}\label{eq:C3Z2Z2U3}
\begin{split}
Z^{\vec{r}=(r,0\dots,0)}_{[\FC^3/\RZ_2\times\RZ_2]}(\vec \qu;\epsilon_1,\epsilon_2,\epsilon_3)&=\frac{M\big((-1)^{r}\,\Qu\big)^{r\,\frac{\epsilon_1\,\epsilon_2\,\epsilon_3-\epsilon_1^2\,\epsilon_2-\epsilon^2_1\,\epsilon_3-\epsilon^2_2\,\epsilon_3-\epsilon_1\,\epsilon_2^2-\epsilon_1\,\epsilon_3^2-\epsilon_2\,\epsilon^2_3}{\epsilon_1\,\epsilon_2\,\epsilon_3}}}{L\big(\qu_1,\qu_2,\qu_3,(-1)^{r}\,\Qu\big)^r} \\
& \quad \, \times \prod_{1\leq p<s\leq 3}\,\widetilde{M}\big(\qu_p\,\qu_s,(-1)^{r}\,\Qu\big)^{r\,\frac{\epsilon_{(ps)^-}-\epsilon_{ps}}{2\,\epsilon_{(ps)^-}}} \ .\end{split}
\end{align}
\end{proposition}

\proof
This follows immediately from \eqref{eq:Orb_pf_C3Z2Z2C} by using Proposition \ref{prop3}. 
\endproof

\subsection{Pure $\mathcal{N}_{\textrm{\tiny T}}=2$ Gauge Theory on $[\FC^3/\RZ_2\times\RZ_2]\times\FC$}

For the pure gauge theories on the quotient stack $[\FC^3/\RZ_2\times\RZ_2]\times\FC$, obtained via the decoupling limit of Proposition~\ref{Prop4}, the analogue of Proposition~\ref{prop:pureorbZn} reads as

\begin{proposition}\label{prop:pureorbZ2Z2}
Assume Conjectures~\ref{con3} and~\ref{con4a} are true.
Then the equivariant instanton partition function of the pure cohomological $\sU(1)$ gauge theory on $[\FC^3/\RZ_2\times\RZ_2]\times\FC$ is given by
\begin{align}\begin{split}
Z_{[\FC^3/(\RZ_2\times\RZ_2)\times\FC]}(\vec \Lambda;\vec{\epsilon}\,)^{\rm pure} &=\exp-\varLambda\,\Big(\frac{\epsilon_1\,\epsilon_2\,\epsilon_3-\epsilon_1^2\,\epsilon_2-\epsilon^2_1\,\epsilon_3-\epsilon^2_2\,\epsilon_3-\epsilon_1\,\epsilon_2^2-\epsilon_1\,\epsilon_3^2-\epsilon_2\,\epsilon^2_3}{\epsilon_1\,\epsilon_2\,\epsilon_3\,\epsilon_4} \\
& \hspace{2.5cm} + \sum_{1\leq p<s\leq 3}\,\frac{\big(\epsilon_{(ps)^{-}}-\epsilon_{ps}\big)\,\big(\Lambda_p\,\Lambda_s+\Lambda_p^{-1}\,\Lambda_s^{-1}\big)}{2\,\epsilon_4\,\epsilon_{(ps)^-}} \\
& \hspace{2.5cm} +\frac{1}{\epsilon_4}\,\sum_{s=1}^3\,\big(\Lambda_s+\Lambda_s^{-1}\big)-\frac{1}{\epsilon_4}\,\big(\Lambda_{123}+\Lambda_{123}^{-1}\big)\Big) \ ,\end{split}
\end{align}
where $\varLambda=\Lambda_0\,\Lambda_1\,\Lambda_2\,\Lambda_3$ and $\Lambda_{123}=\Lambda_1\,\Lambda_2\,\Lambda_3$, while the higher rank partition functions of type $\vec r=(r,0,\dots,0)$ are all trivial:
\begin{align}
Z^{\vec r=(r,0,\dots,0)}_{[\FC^2/\RZ_2\times\RZ_2]\times\FC}(\vec\Lambda;\vec{\epsilon}\,)^{\rm pure} = 1 \qquad \mbox{for} \quad r>1 \ .
\end{align}
\end{proposition}

\begin{proof}
The proof is completely analogous to the proof of Proposition~\ref{prop:pureorbZn}.
\end{proof}

\appendix

\section{Generalized ADHM Construction} \label{app:ADHMconstruction}

In this appendix we construct the ADHM type finite-dimensional matrix model of the moduli space of finite action solutions to the noncommutative instanton equations \eqref{Zin}. To write these generalized ADHM equations, we first introduce two Hermitian vector spaces $V$ and $W$ of complex dimensions $k$ and $r$, respectively. Let $U$ be an $(8\,k{+}r){\times} r$ matrix which solves the Weyl equation
\begin{align}
    \Delta^{\dagger}\, U=0 \ ,
\end{align}
where $\Delta$ is the  $(8\,k+r){\times} 8\,k$ matrix
\begin{align}
    \Delta={\small\begin{pmatrix}
    b_1^\dagger& b_2& b_3& b_4&0&0&0&0\\
    b_2^\dagger& -b_1& 0& 0&b_3&-b_4&0&0\\
    b_3^\dagger& 0& -b_1& 0&b_2&0&b_4&0\\
    b_4^\dagger& 0& 0& b_1^\dagger&0&b_2&b_3&0\\
    0&0&0&0& b_4^\dagger&b_3^\dagger&b_2^\dagger&b_1\\
    0&0& -b_4^\dagger&b_3^\dagger&0&0&-b_1^\dagger&b_2\\
    0& b_4^\dagger&0&-b_2^\dagger&0&-b_2^\dagger&-b_1^\dagger&b_3\\
    0& -b_3^\dagger&-b_2^\dagger&0&-b_1^\dagger&0&0&b_4\\
    I^\dagger&0&0&0&0&0&0&0
    \end{pmatrix}} \ .
\normalsize
\end{align}
Here $b_a=B_a-z_a\,\ident_k$, with $B_{a} \in\sEnd_\mathbbm{C}(V)$ for $a\in\ulfour$ and $I \in \sHom_\mathbbm{C}(W,V)$.

The auxiliary matrix $\Delta$ is required to satisfy the equation
\begin{align}
\Delta^\dagger\,\Delta=\ident_{8}\otimes f_{k }^{-1} \ ,    
\end{align} 
where $f_{ k}$ is an invertible $k{\times} k$ matrix. This leads to the equations for the ADHM data $(B_{a},I)_{a\in\ulfour}$  given by \eqref{ADHMeq}:
\begin{align}
 [B_a,B_b]-\tfrac{1}{2}\,\epsilon_{ab\bar c\bar d}\,\big[B_{\bar c}^\dagger,B_{\bar d}^\dagger\big]=0 \qquad \mbox{and} \qquad   \sum_{a=1}^4\,\big[B_a,B_{\bar a}^\dagger\big]+I\,I^\dagger=\xi\, \ident_{ k} \ .
\end{align}
One asks that the matrix $U$ be normalized: $U^\dagger\, U=\ident_{r} $. Then the columns of $U$  together with $\Delta$ form a complete basis in $\mathbbm{C}^{8\,k+r}$, and therefore
\begin{align}
\ident_{8\,k+r}-U\, U^\dag=\Delta\,(\ident_{ 8}\otimes f_k)\,\Delta^\dagger \ .\label{condition}
\end{align}

The $\sSpin(7)$-instanton connection can now be written as $A=U^\dagger\,\dd U$. Indeed, using \eqref{condition} we compute the components of its curvature two-form to get
\begin{align}
\begin{split}
    F_{\mu\nu}&=\partial_\mu U^\dagger\,\partial_\nu U-\partial_\nu U^\dagger\,\partial_\mu U+\big[U^\dagger\,\partial_\mu U,U^\dagger\,\partial_\nu U\big] \\[4pt]
   &=\partial_{[\mu}U^\dagger\,\big(\ident_{8\,k+r}-U\,U^\dagger\big)\,\partial_{\nu]}U \\[4pt]
   &=U^\dagger\,\partial_{[\mu}\Delta\,(\ident_8\otimes f_k)\,\partial_{\nu]}\Delta^\dagger\,U =    U^\dagger\,\big(\Sigma_{\mu\nu}^{\textrm{\tiny$(+)$}}\otimes f_{ k}\big)\,U \ ,
\end{split}
\end{align}
which satisfy the self-duality equations \eqref{eq:spin7inst} $(\lambda=1)$ for $\sSpin(7)$-holonomy.  Here we have introduced the eight-dimensional counterparts of the 't Hooft symbols
\begin{align}
    \Sigma^{\textrm{\tiny$(+)$}}_{\mu\nu}=\bar\Sigma_{\mu}\, \Sigma_\nu -\bar\Sigma_\nu\, \Sigma_\mu \qquad \mbox{and} \qquad \Sigma^{\textrm{\tiny$(-)$}}_{\mu\nu}=\Sigma_{\mu}\, \bar\Sigma_\nu -\Sigma_\nu\, \bar\Sigma_\mu \ .
\end{align}
The $8{\times}8$ spin matrices $\Sigma_\mu$ and $\bar\Sigma_\mu$ for $\mu=1,\dots,8$ are generators of the Clifford algebra $\mathsf{C\ell}(8)$, where ${\Sigma}_8=\bar\Sigma_8=\ident_8$ and the matrices ${\Sigma}_a=-\bar\Sigma_a$ for $a=1,\dots,7$ satisfy the anticommutation relations $\{\Sigma_a,\Sigma_b\}=-2\,\delta_{ab}\,\ident_8$. 

Consequently, the complex connection
\begin{align}
    A_a=\tfrac{1}{\sqrt{2\,\xi}} \, U^\dagger\,\partial_a U 
\end{align}
for $a\in\ulfour$ satisfies the instanton equations \eqref{in}.

\section{Infinite Product Formulas for Instanton Partition Functions}
\label{app:closedformulas}

In this appendix we outline a possible alternative proof of {Conjecture}~\ref{Prop1} based on the quiver matrix model and its combinatorial evaluation, and then proceed to sketch how this can be extended to provide a potential similar proof of Conjecture~\ref{con2}.

\subsection{Evidence for {Conjecture}~\ref{Prop1}}
\label{app:Prop1}

We start by explicitly computing  the $k=1$ contribution to the instanton partition function \eqref{Zfull}, resulting in

\begin{lemma}\label{lem:1instC4}
For any rank $r\geq1$, the one-instanton contribution to $Z^{r}_{\FC^4}(\qu;\vec{a},\vec{\epsilon},\vec{m})$ is given by
\begin{align}\label{eq:1instC4}
Z^{r,1}_{\FC^4}(\vec{a},\vec{\epsilon},\vec{m})&=\frac{\epsilon_{12}\,\epsilon_{13}\,\epsilon_{23}}{\epsilon_1\,\epsilon_2\,\epsilon_3\,\epsilon_4}\,r\,m \qquad \mbox{with} \quad m=\frac{1}{r}\,\sum_{l=1}^r\,(m_l-a_l) \ .
\end{align}
\end{lemma}

\proof
Using the formula \eqref{Zk} we immediately see the result for $r=1$. So we assume $r>1$, and shifting the masses $m_l$ to $m_l':=m_l-a_l$ for $l=1,\dots,r$ we get
\begin{align}\label{eq:A2}
Z^{r,1}_{\FC^4}(\vec{a},\vec{\epsilon},\vec{m})&=\frac{\epsilon_{12}\,\epsilon_{13}\,\epsilon_{23}}{\epsilon_1\,\epsilon_2\,\epsilon_3\,\epsilon_4} \, \sum^r_{l=1}\, m'_{l} \ \prod^r_{\stackrel{\scriptstyle p=1}{\scriptstyle p\neq l}}\, \Big(1-\frac{m_p'}{a_{lp}}\Big) \ ,
\end{align}
where $a_{lp}=a_l-a_p=-a_{pl}$. Using \eqref{eq:A2} the result is easy to check for $r=2$, so henceforth we restrict to ranks $r>2$. The sum in \eqref{eq:A2} can be rewritten as
\begin{align}\begin{split}
\sum^r_{l=1}\, m'_{l} \ \prod^r_{\stackrel{\scriptstyle p=1}{\scriptstyle p\neq l}} \, \Big(1-\frac{m'_p}{a_{lp}}\Big)&= \sum_{l=1}^r\,m_l'\,\Big(1+\sum_{p=1}^{r-1} \ \sum_{\substack{1\leq i_1<\cdots< i_p\leq r \\ i_j\neq l}}\,\frac{m'_{i_1}\cdots m'_{i_p}}{a_{i_1l}\cdots a_{i_pl}}\Big) \\[4pt]
&= r\,m + \sum_{p=1}^{r-1} \ \sum_{l=1}^r\,m_l' \ \sum_{\substack{1\leq i_1<\cdots< i_p\leq r \\ i_j\neq l}}\,\frac{m'_{i_1}\cdots m'_{i_p}}{a_{i_1l}\cdots a_{i_pl}} \ .
\end{split}\label{induction}
\end{align}

For each $1\leq p\leq r-1$, the second sum of \eqref{induction} can be expressed in the form
\begin{align}
\begin{split}
& \sum_{l=1}^{r}\,m_l' \ \sum_{\substack{1\leq i_1<\cdots< i_p\leq r\\i_j\neq l}}\,\frac{m'_{i_1}\cdots m'_{i_p}}{a_{i_1 l}\cdots a_{i_pl}} \\[4pt]
& \hspace{2cm} =\Big(\sum_{1\leq l<i_1<\cdots <i_p\leq r}+\sum_{1\leq i_1<l<i_2<\cdots <i_p\leq r}+\cdots+\sum_{1\leq i_1<\cdots <i_p<l\leq r}\Big)\,\frac{m'_{i_1}\cdots m'_{i_p}\,m'_l}{a_{i_1l}\cdots a_{i_pl}} \\[4pt]
& \hspace{2cm} = \sum_{1\leq l<i_1<\cdots <i_p\leq r}\,m'_{i_1}\cdots m'_{i_p}\,m'_l\,\Big((a_{i_1l}\cdots a_{i_pl})^{-1} + \sum_{j=1}^p \, a_{li_j}^{-1} \ \prod_{\substack{n=1 \\ n\neq j}}^p\,a_{i_ni_j}^{-1}\Big) \\[4pt]
& \hspace{2cm} = \sum_{1\leq l<i_1<\cdots <i_p\leq r}\,\frac{m'_{i_1}\cdots m'_{i_p}\,m'_l}{a_{i_1l}\cdots a_{i_pl} \ \displaystyle \prod_{1\leq j<n\leq p}\,a_{i_ji_n}} \ \CA_{i_1\cdots i_pl} \ ,
\end{split}
\end{align}
where, for each increasing sequence $1\leq l<i_1<\cdots<i_p\leq r$, we set
\begin{align}\label{eq:CAil}
\begin{split}
\CA_{i_1\cdots i_pl} :\!&= \prod_{1\leq j<n\leq p} \, a_{i_ji_n} - \sum_{j=1}^{p} \, (-1)^{p-j} \ \prod_{\substack{j'=1 \\ j'\neq j}}^p \, a_{i_{j'}l} \ \prod_{\substack{1\leq q<n\leq p \\ n,q\neq j}} \, a_{i_{q}i_n} \\[4pt]
&=\sum_{t=1}^{p+1} \,(-1)^{p-t-1} \ \prod_{\substack{1\leq j<n\leq p+1 \\ j,n\neq t}}a_{i_ji_n} \ ,
\end{split}
\end{align}
with the convention $i_{p+1}:=l$. 

Although it should be possible to show directly that \eqref{eq:CAil} vanishes (indeed we have checked this explicitly up to $p=4$), a more straightforward proof uses Proposition~\ref{prop:ZDTgeb} to assert that \eqref{eq:1instC4} holds at the mass specialisations $m_l'=\epsilon_4$ for $l=1,\dots,r$. Since \eqref{eq:CAil} is independent of the shifted mass parameters, it follows that 
\begin{align}
\CA_{i_1\cdots i_pl} = 0
\end{align}
as required.
\endproof

We are now ready to sketch an argument that may prove {Conjecture}~\ref{Prop1}.
Consider the instanton partition function 
\begin{align}
Z_{\FC^4}^r(\qu;\vec a,\vec{\epsilon},\vec{m}')=1+\sum_{k=1}^\infty\,\qu^k \, Z_{\FC^4}^{r,k}(\vec{a},\vec{\epsilon},\vec{m}') 
\end{align}
with the explicit combinatorial expansion \eqref{Zk}, where $\vec m':=\vec m-\vec a$. By Proposition \ref{prop:ZDTgeb} we know that
\begin{align}
Z_{\FC^4}^{r}(\qu;\vec a,\vec\epsilon,m'_l=-\epsilon_{123}) =  Z_{\FC^3}^{r}\big((-1)^{r+1}\,\qu;\epsilon_1,\epsilon_2,\epsilon_3\big) = {M}(-\qu)^{-\frac{r\,\epsilon_{12}\,\epsilon_{13}\,\epsilon_{23}}{\epsilon_1\,\epsilon_2\,\epsilon_3}} \ .
\end{align}
Armed with this information, we can assume that the partition function takes a form given by
\begin{align}
\log Z_{\FC^4}^r(\qu;\vec a,\vec{\epsilon},\vec{m}') = f_r(\vec a,\vec\epsilon,\vec m')\log M(-\qu) + \log G_r(\qu;\vec a,\vec\epsilon,\vec m') \ ,
\end{align}
where $f_r$ is a rational function of the equivariant parameters with
\begin{align} \label{eq:frspecial}
f_r(\vec a,\vec\epsilon,m_l'=-\epsilon_{123}) = -\frac{r\,\epsilon_{12}\,\epsilon_{13}\,\epsilon_{23}}{\epsilon_1\,\epsilon_2\,\epsilon_3} \ ,
\end{align}
and the function $G_r$ has a power series expansion
\begin{align}
G_r(\qu;\vec{a},\vec{\epsilon},\vec{m}')=1+\sum_{k=1}^\infty\,\qu^k\,G_r^{(k)}(\vec{a},\vec{\epsilon},{\vec m'})
\end{align}
whose coefficients $G_r^{(k)}$ are rational functions of the equivariant parameters with
\begin{align}
G^{(k)}_r(\vec{a},\vec{\epsilon},m_l'=-\epsilon_{123})=0 \ .
\end{align}

From Lemma~\ref{lem:1instC4} it follows that
 \begin{align}
f_r(\vec{a},\vec{\epsilon},\vec{m}')-G_r^{(1)}(\vec{a},\vec{\epsilon},\vec{m}')=-\frac{r\,m\,\epsilon_{12}\,\epsilon_{13}\,\epsilon_{23}}{\epsilon_1\,\epsilon_2\,\epsilon_3\,\epsilon_4} \ .
 \end{align}
Hence, by redefining the functions $f_r$ and $G_r^{(1)}$ if necessary, we can assume that
\begin{align}
f_r(\vec a,\vec{\epsilon},\vec m')=-\frac{r\,m\,\epsilon_{12}\,\epsilon_{13}\,\epsilon_{23}}{\epsilon_1\,\epsilon_2\,\epsilon_3\,\epsilon_4} \qquad \mbox{and} \qquad G_r^{(1)}(\vec{a},\vec{\epsilon},\vec{m}')=0 \ .
\end{align}

The idea now is to prove that $G_r^{(k)}(\vec{a},\vec{\epsilon},\vec m')= 0$ by induction on $k$. We know this for $k=1$, so we suppose \smash{$G_r^{(n)}(\vec{a},\vec{\epsilon},\vec m,r)= 0$} for $1\leq n\leq k-1$ with $k>1$.
Then
\begin{align}\label{eq:ZnrFg}
Z^{r,k}_{\FC^4}(\vec{a},\vec{\epsilon},\vec{m}')=F_r^{(k)}(\vec a,\vec{\epsilon},\vec m')+G_r^{(k)}(\vec{a},\vec{\epsilon},\vec m') \ ,
\end{align}
where $F_r^{(k)}(\vec a,\vec{\epsilon},\vec m')$ is the coefficient of $\qu^k$ in the series expansion of $\exp(f_r(\vec a,\vec{\epsilon},\vec m')\log M(-\qu))$. Recalling the symmetries of the matrix integral \eqref{partition}, we know that \smash{$Z^{r,k}_{\FC^4}(\vec{a},\vec{\epsilon},\vec{m}')$} is invariant under permutation of $\epsilon_1$ and $\epsilon_4=-\epsilon_{123}$, as well as under permutation of the entries of $\vec m'=(m_1',\dots,m_r')$. Since \smash{$F_r^{(k)}(\vec a,\vec{\epsilon},\vec m')$} is invariant under these permutations, so is \smash{$G_r^{(k)}(\vec{a},\vec{\epsilon},\vec m')$}. 

Looking at \eqref{Zk}, we can decompose the $k$-instanton contributions for $k>1$ into
\begin{align}\label{eq:Zkrmasspoly}
Z^{r,k}_{\FC^4}(\vec{a},\vec{\epsilon},\vec{m}')= \sum_{\imath_1,\dots,\imath_r=1}^k \, \frZ_{\FC^4}^{r,k;\vec{\imath}}(\vec{a},\vec{\epsilon}\,) \ \prod_{l=1}^r\, (m'_l+\epsilon_{123})^{\imath_l} +(-1)^{(r+1)\,k}\,Z^{r,k}_{\FC^3}(\epsilon_1,\epsilon_2,\epsilon_3)
\end{align}
for some functions $\frZ_{\FC^4}^{r,k;\vec{\imath}}(\vec{a},\vec{\epsilon}\,)$. 
We now write
\begin{align}
f_r(\vec a,\vec{\epsilon},\vec m')=-\frac{r\,\epsilon_{12}\,\epsilon_{13}\,\epsilon_{23}}{\epsilon_1\,\epsilon_2\,\epsilon_3}-\sum_{l=1}^r\,\frac{\epsilon_{12}\,\epsilon_{13}\,\epsilon_{23}\,(m'_l+\epsilon_{123})}{\epsilon_1\,\epsilon_2\,\epsilon_3\,\epsilon_4}
\end{align}
and use this to separate out the polynomial mass dependence in $F_r^{(k)}(\vec a,\vec{\epsilon},\vec m')$, similarly to \eqref{eq:Zkrmasspoly}, as
\begin{align}
F_r^{(k)}(\vec a,\vec{\epsilon},\vec m') = \sum_{\imath_1,\dots,\imath_r=1}^k \, \frF_{r}^{(k);\vec{\imath}}(\vec{a},\vec{\epsilon}\,) \ \prod_{l=1}^r\, (m'_l+\epsilon_{123})^{\imath_l} +(-1)^{(r+1)\,k}\,Z^{r,k}_{\FC^3}(\epsilon_1,\epsilon_2,\epsilon_3)
\end{align}
with some functions \smash{$\frF_{r}^{(k);\vec{\imath}}(\vec{a},\vec{\epsilon}\,)$}. The coefficient functions \smash{$\frZ_{\FC^4}^{r,k;\vec{\imath}}(\vec a,\vec\epsilon\,)$}  and  \smash{$\frF_{r}^{(k);\vec{\imath}}(\vec a,\vec\epsilon\,)$} are independent of the masses $\vec m'$ and symmetric in the entries of $\vec\imath=(\imath_1,\dots,\imath_r)$.

Then \eqref{eq:ZnrFg} determines \smash{$G_r^{(k)}$} as the Taylor expansion
\begin{align}
G_r^{(k)}(\vec a,\vec \epsilon, \vec m')=\sum_{\imath_1,\dots, \imath_r=1}^k \, \big(\frZ_{\FC^4}^{r,k;\vec{\imath}}(\vec{a},\vec{\epsilon}\,) - \frF_{r}^{(k);\vec{\imath}}(\vec{a},\vec{\epsilon}\,) \big) \ \prod_{l=1}^r \, (m'_l+\epsilon_{123})^{\imath_l} \ .
\end{align}
At this stage one should be able to exploit the $\RZ_2$-symmetry $\epsilon_1\leftrightarrow-\epsilon_{123}$ of $G_r^{(k)}(\vec a,\vec \epsilon,\vec m')$, together with the analytic behaviour of \eqref{Zk} in $\vec\epsilon\,$, to infer that \smash{$\frZ_{\FC^4}^{r,k;\vec{\imath}}(\vec{a},\vec{\epsilon}\,) = \frF_{r}^{(k);\vec{\imath}}(\vec{a},\vec{\epsilon}\,)$} for each $\vec{\imath}$. It would be very interesting to understand this further and complete the proof of {Conjecture}~\ref{Prop1} along these lines.

\subsection{Evidence for Conjecture~\ref{con2}}
\label{app:con2}

A possible proof of Conjecture~\ref{con2} follows the same line of reasoning as in Appendix~\ref{app:Prop1}, starting with the dimensional reduction \eqref{orb6dZn} according to Proposition~\ref{prop3}. We choose $\vec r=(1,0,\dots, 0)$ without loss of generality. We can assume that
\begin{align}\begin{split}
Z_{[\FC^2/\RZ_n]\times \FC^2}(\vec \qu;\vec{\epsilon},m)&=
M\big((-1)^n\,\Qu\big)^{-n\,\frac{m\,\epsilon_{12}\,\epsilon_{13}\,\epsilon_{23}}{\epsilon_1\,\epsilon_2\,\epsilon_3\,\epsilon_4}-\frac{n^2-1}{n}\,\frac{m\,\epsilon_{12}}{\epsilon_1\,\epsilon_2}} \\
& \quad \, \times \prod_{0<p\leq s<n}\,\widetilde{M}\big((-1)^{p-s+1}\,\qu_{[p,s]},(-1)^n\,\Qu\big)^{-\frac{m\,\epsilon_{12}}{\epsilon_3\,\epsilon_4}} \ G(\vec \qu;\vec \epsilon, m)\end{split} \ ,
\end{align}
where the function $G$ has a power series expansion
\begin{align}
G(\vec \qu;\vec \epsilon, m)=1+ \sum_{\vec k\in\RZ_{\geq0}^n\setminus\,\vec 0} \, \vec\qu^{\,\vec k} \ G^{( \vec k\,)}(\vec \epsilon, m)
\end{align}
whose coefficients are rational functions of the equivariant parameters 
with $G^{(\vec k\,)}(\vec \epsilon, m=-\epsilon_{123})=0$.

The idea is to proceed by induction on the size of \smash{$\vec{k}\in\RZ_{\geq0}^n\setminus\vec 0$} to show that $G^{(\vec k\,)}(\vec \epsilon, m)=0$. For $|\vec k\,|=1$ the only contribution to the instanton partition function is
\begin{align}
\qu_0\,\frac{m\,\epsilon_{12}}{\epsilon_3\,\epsilon_4} \ .
\end{align}
Thus $G^{(\vec{k}\,)}(\vec{\epsilon},m)=0$ for all $\vec{k}$ of size one. Now suppose $G^{(\vec{k}\,)}(\vec \epsilon ,m)=0$ for all $\vec k$ of sizes \smash{$1\leq|\vec k\,|\leq k-1$} with $k>1$. Then for $\vec k$ of size $|\vec k\,|=k$ we can write
\begin{align}\label{eq:ZveckFG}
Z^{\vec k}_{[\FC^2/\RZ_n]\times\FC^2}(\vec{\epsilon},m) = F^{(\vec k\,)}(\vec \epsilon, m)+ G^{(\vec{k}\,)}(\vec \epsilon ,m) \ ,
\end{align}
where $F^{(\vec k\,)}(\vec \epsilon, m)$ is the coefficient of $\vec\qu^{\,\vec k}$ in the power series expansion of
\begin{align}\label{eq:MacMahonpowers}
M\big((-1)^n\,\Qu\big)^{-n\,\frac{m\,\epsilon_{12}\,\epsilon_{13}\,\epsilon_{23}}{\epsilon_1\,\epsilon_2\,\epsilon_3\,\epsilon_4}-\frac{n^2-1}{n}\,\frac{m\,\epsilon_{12}}{\epsilon_1\,\epsilon_2}} \ \prod_{0<p\leq s<n}\,\widetilde{M}\big((-1)^{p-s+1}\,\qu_{[p,s]},(-1)^n\,\Qu\big)^{-\frac{m\,\epsilon_{12}}{\epsilon_3\,\epsilon_4}} \ .
\end{align}

From the matrix integral \eqref{eq:Orb_Zin} with $\vec r=(1,0,\dots,0)$ for the $\RZ_n$-action of Section~\ref{sec:C2ZnC2}, with weights $s_1=1$, $s_2=n-1$ and $s_3=s_4=0$, it follows that the fractional instanton contribution \smash{$Z^{\vec k}_{[\FC^2/\RZ_n]\times\FC^2}(\vec{\epsilon},m)$} is invariant under the permutation of $\epsilon_3$ and $\epsilon_4= -\epsilon_{123}$. From the combinatorial expansion \eqref{eq:ZC2Znk} it follows that it can be decomposed into
\begin{align}
Z^{\vec k}_{[\FC^2/\RZ_n]\times\FC^2}(\vec{\epsilon},m)=\sum_{\imath=1}^{|\vec k\,|} \, (m+\epsilon_{123})^\imath \ \frZ^{\vec k;\imath}_{[\FC^2/\RZ_n]\times\FC^2}(\vec{\epsilon}\,) + (-1)^{|\vec k\,|+k_0} \, Z^{\vec k}_{[\FC^2/\RZ_n]\times\FC}(\epsilon_1,\epsilon_2,\epsilon_3) \ ,
\end{align} 
with some functions \smash{$\frZ^{\vec k;\imath}_{[\FC^2/\RZ_n]\times\FC^2}(\vec{\epsilon}\,)$} for $\imath=1,\dots, |\vec k\,|$ which are independent of the mass $m$.

The powers of the generalized MacMahon functions in \eqref{eq:MacMahonpowers} can be rewritten respectively as
\begin{align}
-n\,\frac{\epsilon_{12}\,\epsilon_{13}\,\epsilon_{23}}{\epsilon_1\,\epsilon_2\,\epsilon_3}+\frac{n^2-1}{n}\,\frac{\epsilon_{12}\,\epsilon_{123}}{\epsilon_1\,\epsilon_2}
 -\Big(n\,\frac{\epsilon_{12}\,\epsilon_{13}\,\epsilon_{23}}{\epsilon_1\,\epsilon_2\,\epsilon_3\,\epsilon_4}-\frac{n^2-1}{n}\,\frac{\epsilon_{12}}{\epsilon_1\,\epsilon_2}\Big)\,(m+\epsilon_{123})
 \end{align}
 and
 \begin{align}
-\frac{\epsilon_{12}}{\epsilon_3}-\frac{\epsilon_{12}\,(m+\epsilon_{123})}{\epsilon_3\,\epsilon_4} \ .
 \end{align}
Since $F^{(\vec k\,)}(\vec \epsilon, m)$ is symmetric under $\epsilon_3\leftrightarrow-\epsilon_{123}$, it follows from \eqref{eq:ZveckFG} that so is $G^{(\vec{k}\,)}(\vec \epsilon ,m)$, and that it is given as the Taylor expansion
\begin{align}
G^{(\vec{k}\,)}(\vec \epsilon ,m)=\sum_{\imath=1}^{|\vec k\,|}\, \big(\frZ_{[\FC^2/\RZ_n]\times\FC^2}^{\vec k;\imath}(\vec \epsilon\,)- \frF^{(\vec k\,);\imath}(\vec \epsilon\,)\big) \, (m+\epsilon_{123})^\imath \ ,
\end{align}
with some functions $\frF^{(\vec k\,);\imath}(\vec \epsilon\,)$ for $\imath=1,\dots, |\vec k\,|$ which are independent of the mass parameter $m$. At this stage one should be able to exploit invariance under the $\RZ_2$-action $\epsilon_3\leftrightarrow-\epsilon_{123}$, together with the analytic behaviour of \eqref{eq:ZC2Znk} in $\vec\epsilon\,$, to show that \smash{$\frF^{(\vec k\,);\imath}(\vec \epsilon\,) = \frZ_{[\FC^2/\RZ_n]\times\FC^2}^{\vec k;\imath}(\vec \epsilon\,)$} for all $\imath\in \{1,\dots,|\vec k\,|\}$. 


\bibliographystyle{ourstyle}
\bibliography{Orbifold-bibliography}

\end{document}